\pdfoutput=1

\documentclass[sigconf,nonacm]{acmart}
\setcopyright{none}
\usepackage{algorithmicx}
\usepackage[noend]{algpseudocode}
\usepackage{algorithm}

\algnewcommand{\algorithmicswitch}{\textbf{switch}}
\algdef{SE}[SWITCH]{Switch}{EndSwitch}[1]{\algorithmicswitch\ #1\ \algorithmicdo}{\algorithmicend\ \algorithmicswitch}%
\algtext*{EndSwitch}%

\algnewcommand{\algorithmiccase}{\textbf{case}}
\algdef{SE}[CASE]{Case}{EndCase}[1]{\algorithmiccase\ #1}{\algorithmicend\ \algorithmiccase}%
\algtext*{EndCase}%

\algnewcommand{\algorithmicon}{\textbf{on}}
\algdef{SE}[ON]{On}{EndOn}[1]{\algorithmicon\ #1\ \algorithmicdo}{\algorithmicend\ \algorithmicon}%
\algtext*{EndOn}%

\algrenewcommand{\algorithmicdo}{}
\algrenewcommand{\algorithmicthen}{}

\algnewcommand{\algorithmicgoto}{\textbf{goto}}%
\algnewcommand{\Goto}[1]{\algorithmicgoto~\ref{#1}}%

\algnewcommand{\algorithmicbreak}{\textbf{break}}%
\algnewcommand{\Break}[0]{\algorithmicbreak}%

\algnewcommand{\algorithmicassert}{\textbf{assert}}%
\algnewcommand{\Assert}[1]{\algorithmicassert~{#1}}%

\algnewcommand{\algorithmicgiveup}{\textbf{give up}}%
\algnewcommand{\GiveUp}[0]{\algorithmicgiveup}%

\algnewcommand{\algorithmicabort}{\textbf{abort}}%
\algnewcommand{\Abort}[0]{\algorithmicabort}%

\algnewcommand{\algorithmicwaiton}{\textbf{wait on}}%
\algnewcommand{\WaitOn}[1]{\algorithmicwaiton~{#1}}%

\newcommand*\circledScriptGray[1]{\tikz[baseline=(char.base)]{
            \node[shape=circle,draw=black,text=black,inner sep=2pt,fill=black!10] (char) {\scriptsize #1};}}
\newcommand*\circledScriptGraySlim[1]{\tikz[baseline=(char.base)]{
            \node[shape=circle,draw=black,text=black,inner sep=1pt,fill=black!10] (char) {\scriptsize #1};}}

\PassOptionsToPackage{hyphens}{url}
\usepackage{url}
\usepackage{hyperref}
\hypersetup{breaklinks=true}

\usepackage[utf8]{inputenc}
\usepackage[T1]{fontenc}

\usepackage{amsfonts,amsmath,amsthm}
\usepackage{graphicx}
\usepackage{xcolor}

\usepackage[caption=false,font=footnotesize]{subfig}
\usepackage{float}

\usepackage{IEEEtrantools}
\allowdisplaybreaks

\usepackage[shortlabels,inline]{enumitem}
\setlist[itemize]{leftmargin=1.25em}
\setlist[enumerate]{leftmargin=1.75em}

\usepackage{xstring}

\usepackage{datetime2}

\usepackage{multirow,booktabs}

\usepackage{siunitx}

\usepackage{varwidth}

\usepackage{tikz}
\usetikzlibrary{calc}
\usetikzlibrary{arrows}
\usetikzlibrary{arrows.meta}
\usetikzlibrary{patterns}
\usetikzlibrary{positioning}
\usetikzlibrary{decorations.pathreplacing}
\usetikzlibrary{shapes.misc}
\usetikzlibrary{spy}

\pgfdeclarelayer{bg1}
\pgfdeclarelayer{bg2}
\pgfsetlayers{bg1,bg2,main}

\usepackage{pgfplots}
\pgfplotsset{compat=1.14}
\usepgfplotslibrary{fillbetween}

\usepackage{xspace}
\usepackage{ifthen}

\newcommand{\ie}[0]{\emph{i.e.}\xspace}
\newcommand{\eg}[0]{\emph{e.g.}\xspace}
\newcommand{\cf}[0]{\emph{cf.}\xspace}

\theoremstyle{plain}
\newtheorem*{theorem*}{Theorem}
\newtheorem*{definition*}{Definition}

\definecolor{myParula01Blue}{RGB}{0,114,189}
\definecolor{myParula02Orange}{RGB}{217,83,25}
\definecolor{myParula03Yellow}{RGB}{237,177,32}
\definecolor{myParula04Purple}{RGB}{126,47,142}
\definecolor{myParula05Green}{RGB}{119,172,48}
\definecolor{myParula06LightBlue}{RGB}{77,190,238}
\definecolor{myParula07Red}{RGB}{162,20,47}

\tikzset{myparula11/.style={color=myParula01Blue,solid,mark=+,mark options={solid}}}
\tikzset{myparula12/.style={color=myParula01Blue,densely dashed,mark=x,mark options={solid}}}
\tikzset{myparula13/.style={color=myParula01Blue,densely dotted,mark=o,mark options={solid}}}
\tikzset{myparula14/.style={color=myParula01Blue,dashdotted,mark=triangle,mark options={solid}}}
\tikzset{myparula15/.style={color=myParula01Blue,dashdotdotted,mark=square,mark options={solid}}}

\tikzset{myparula21/.style={color=myParula02Orange,solid,mark=+,mark options={solid}}}
\tikzset{myparula22/.style={color=myParula02Orange,densely dashed,mark=x,mark options={solid}}}
\tikzset{myparula23/.style={color=myParula02Orange,densely dotted,mark=o,mark options={solid}}}
\tikzset{myparula24/.style={color=myParula02Orange,dashdotted,mark=triangle,mark options={solid}}}
\tikzset{myparula25/.style={color=myParula02Orange,dashdotdotted,mark=square,mark options={solid}}}

\tikzset{myparula31/.style={color=myParula03Yellow,solid,mark=+,mark options={solid}}}
\tikzset{myparula32/.style={color=myParula03Yellow,densely dashed,mark=x,mark options={solid}}}
\tikzset{myparula33/.style={color=myParula03Yellow,densely dotted,mark=o,mark options={solid}}}
\tikzset{myparula34/.style={color=myParula03Yellow,dashdotted,mark=triangle,mark options={solid}}}
\tikzset{myparula35/.style={color=myParula03Yellow,dashdotdotted,mark=square,mark options={solid}}}

\tikzset{myparula41/.style={color=myParula04Purple,solid,mark=+,mark options={solid}}}
\tikzset{myparula42/.style={color=myParula04Purple,densely dashed,mark=x,mark options={solid}}}
\tikzset{myparula43/.style={color=myParula04Purple,densely dotted,mark=o,mark options={solid}}}
\tikzset{myparula44/.style={color=myParula04Purple,dashdotted,mark=triangle,mark options={solid}}}
\tikzset{myparula45/.style={color=myParula04Purple,dashdotdotted,mark=square,mark options={solid}}}

\tikzset{myparula51/.style={color=myParula05Green,solid,mark=+,mark options={solid}}}
\tikzset{myparula52/.style={color=myParula05Green,densely dashed,mark=x,mark options={solid}}}
\tikzset{myparula53/.style={color=myParula05Green,densely dotted,mark=o,mark options={solid}}}
\tikzset{myparula54/.style={color=myParula05Green,dashdotted,mark=triangle,mark options={solid}}}
\tikzset{myparula55/.style={color=myParula05Green,dashdotdotted,mark=square,mark options={solid}}}

\tikzset{myparula61/.style={color=myParula06LightBlue,solid,mark=+,mark options={solid}}}
\tikzset{myparula62/.style={color=myParula06LightBlue,densely dashed,mark=x,mark options={solid}}}
\tikzset{myparula63/.style={color=myParula06LightBlue,densely dotted,mark=o,mark options={solid}}}
\tikzset{myparula64/.style={color=myParula06LightBlue,dashdotted,mark=triangle,mark options={solid}}}
\tikzset{myparula65/.style={color=myParula06LightBlue,dashdotdotted,mark=square,mark options={solid}}}

\tikzset{myparula71/.style={color=myParula07Red,solid,mark=+,mark options={solid}}}
\tikzset{myparula72/.style={color=myParula07Red,densely dashed,mark=x,mark options={solid}}}
\tikzset{myparula73/.style={color=myParula07Red,densely dotted,mark=o,mark options={solid}}}
\tikzset{myparula74/.style={color=myParula07Red,dashdotted,mark=triangle,mark options={solid}}}
\tikzset{myparula75/.style={color=myParula07Red,dashdotdotted,mark=square,mark options={solid}}}

\pgfplotsset{
    mysimpleplot/.style = {
        every axis plot/.prefix style={thick},
        width=1.0\linewidth,
        height=0.75\linewidth,
        title style={font=\footnotesize,align=center},
        legend cell align=left,
        legend style={font=\footnotesize},
        legend columns=3,
        legend style={
            at={(0.5,1)},
            yshift=0.3em,
            anchor=south,
            draw=none,
            /tikz/every even column/.append style={
                column sep=0.3em
            },
            cells={
                align=left
            }
        },
        grid=both,
        minor tick num=4,
        major grid style={solid,draw=gray!50},
        minor grid style={densely dotted,draw=gray!50},
        label style={font=\footnotesize,align=center},
        tick label style={font=\footnotesize},
    },
}

\pgfplotsset{
    semiavidprexperiment1/.style = {
        mysimpleplot,
        legend columns=2,
        height=4cm,
        width=3.8cm,
        xtick={0.5e6,1e6,2e6,4e6,8e6,16e6,32e6,64e6},
        xticklabels={$0.5$,$1$,$2$,$4$,$8$,$16$,$32$,$64$},
        ytick={0.01,0.1,1,10,100,1000},
        yticklabels={$0.01$,$0.1$,$1$,$10$,$100$,$1000$},
    },
}

\pgfplotsset{
    semiavidprexperiment2/.style = {
        semiavidprexperiment1,
        height=3.3cm,
    },
}

\tikzset{experiment-rate25-n128/.style={dashed,mark=+,mark options={solid}}}
\tikzset{experiment-rate25-n256/.style={dashed,mark=o,mark options={solid}}}
\tikzset{experiment-rate25-n1024/.style={dashed,mark=triangle,mark options={solid}}}

\tikzset{experiment-rate33-n128/.style={solid,mark=+,mark options={solid}}}
\tikzset{experiment-rate33-n256/.style={solid,mark=o,mark options={solid}}}
\tikzset{experiment-rate33-n1024/.style={solid,mark=triangle,mark options={solid}}}

\tikzset{experiment-rate45-n128/.style={densely dotted,mark=+,mark options={solid}}}
\tikzset{experiment-rate45-n256/.style={densely dotted,mark=o,mark options={solid}}}
\tikzset{experiment-rate45-n1024/.style={densely dotted,mark=triangle,mark options={solid}}}

\pgfplotsset{
    discard if neq/.style 2 args={
        x filter/.code={
            \edef\tempa{\thisrow{#1}}
            \edef\tempb{#2}
            \ifx\tempa\tempb
            \else
                
            \fi
        }
    }
}

\usepackage{pifont}%
\newcommand{\cmark}{\ding{52}}%
\newcommand{\xmark}{\ding{56}}%

\usepackage{shortsym}
\usepackage{mathtools}

\newcommand{\mystate}[0]{\ensuremath{\mathsf{state}}}
\newcommand{\txs}[0]{\ensuremath{\mathsf{txs}}}

\newcommand{\drawrandom}[0]{\ensuremath{\xleftarrow{\mathrm{R}}}}
\newcommand{\cryptoSK}[0]{\ensuremath{\mathsf{sk}}}
\newcommand{\cryptoPK}[0]{\ensuremath{\mathsf{pk}}}
\newcommand{\cryptoSP}[0]{\ensuremath{\mathsf{sp}}}
\newcommand{\cryptoPP}[0]{\ensuremath{\mathsf{pp}}}

\newcommand{\savid}[0]{\ensuremath{\Pi_{\mathrm{SAVIDPR}}}}
\newcommand{\savidCommit}[0]{\ensuremath{\mathsf{Commit}}}
\newcommand{\savidSetup}[0]{\ensuremath{\mathsf{Setup}}}
\newcommand{\savidDisperse}[0]{\ensuremath{\mathsf{Disperse}}}
\newcommand{\savidRetrieve}[0]{\ensuremath{\mathsf{Retrieve}}}
\newcommand{\savidVerify}[0]{\ensuremath{\mathsf{Verify}}}

\newcommand{\OURsavid}[0]{\ensuremath{\Pi^{\bigstar}}}
\newcommand{\OURsavidSUPER}[1]{\ensuremath{\Pi^{\bigstar,#1}}}

\newcommand{\savidLabelStored}[0]{\ensuremath{\mathtt{ack}}}
\newcommand{\savidLabelDisperse}[0]{\ensuremath{\mathtt{store}}}
\newcommand{\savidLabelRetrieve}[0]{\ensuremath{\mathtt{load}}}

\newcommand{\negl}[0]{\ensuremath{\mathrm{negl}}}

\newcommand{\slvc}[0]{\ensuremath{\mathsf{LVC}}}
\newcommand{\slvcSetup}[0]{\ensuremath{\slvc.\mathsf{Setup}}}
\newcommand{\slvcCommit}[0]{\ensuremath{\slvc.\mathsf{Commit}}}

\renewcommand{\VC}[0]{\ensuremath{\mathsf{VC}}}

\newcommand{\code}[0]{\ensuremath{\mathsf{Code}}}
\newcommand{\codeEncode}[0]{\ensuremath{\code.\mathsf{Encode}}}
\newcommand{\codeDecode}[0]{\ensuremath{\code.\mathsf{Decode}}}

\newcommand{\crhf}[0]{\ensuremath{\mathsf{HF}}}
\newcommand{\crhfGen}[0]{\ensuremath{\crhf.\mathsf{Gen}}}
\newcommand{\crhfH}[0]{\ensuremath{\crhf.\mathsf{H}}}

\newcommand{\CRHF}[0]{\ensuremath{\mathsf{CRHF}^s}}

\newcommand{\sig}[0]{\ensuremath{\mathsf{Sig}}}
\newcommand{\sigKeygen}[0]{\ensuremath{\sig.\mathsf{KeyGen}}}
\newcommand{\sigSign}[0]{\ensuremath{\sig.\mathsf{Sign}}}
\newcommand{\sigVerify}[0]{\ensuremath{\sig.\mathsf{Verify}}}

\newcommand{\Prob}[1]{\ensuremath{\operatorname{Pr}\!\left(\,#1\,\right)}}
\newcommand{\CProb}[2]{\Prob{#1\;\middle\vert\;#2}}

\newcommand{\KZG}[0]{\ensuremath{\mathsf{KZG}}}
\newcommand{\KZGsetup}[0]{\ensuremath{\mathsf{KZG.Setup}}}
\newcommand{\KZGcommit}[0]{\ensuremath{\mathsf{KZG.Commit}}}

\newcommand{\KZGproof}[0]{\ensuremath{\mathsf{KZG.CreateWitness}}}
\newcommand{\KZGverifyeval}[0]{\ensuremath{\mathsf{KZG.VerifyEval}}}

\newcommand{\TRUE}{\ensuremath{\mathtt{true}}}
\newcommand{\FALSE}{\ensuremath{\mathtt{false}}}

\newcommand{\gameAv}[0]{\ensuremath{\mathsf{AvG}}}   %
\newcommand{\gameCB}[0]{\ensuremath{\mathsf{CBG}}}   %
\newcommand{\gameEF}[0]{\ensuremath{\mathsf{EFG}}}   %
\newcommand{\gameCF}[0]{\ensuremath{\mathsf{CFG}}}   %
\newcommand{\gameVCB}[0]{\ensuremath{\mathsf{VCBG}}}   %

\newcommand{\eqA}[0]{\ensuremath{\overset{\tiny{\mathrm{(a)}}}{=}}}

\newcommand{\leqB}[0]{\ensuremath{\overset{\tiny{\mathrm{(b)}}}{\leq}}}
\newcommand{\leqC}[0]{\ensuremath{\overset{\tiny{\mathrm{(c)}}}{\leq}}}
\newcommand{\leqD}[0]{\ensuremath{\overset{\tiny{\mathrm{(d)}}}{\leq}}}
\newcommand{\leqE}[0]{\ensuremath{\overset{\tiny{\mathrm{(e)}}}{\leq}}}
\newcommand{\leqF}[0]{\ensuremath{\overset{\tiny{\mathrm{(f)}}}{\leq}}}

\begin{document}

\title[Information Dispersal with Provable Retrievability for Rollups]{Information Dispersal with Provable Retrievability for Rollups}

\author{Kamilla Nazirkhanova}
\email{nazirk@stanford.edu}
\affiliation{\country{}}
\author{Joachim Neu}
\email{jneu@stanford.edu}
\affiliation{\country{}}
\author{David Tse}
\email{dntse@stanford.edu}
\affiliation{\country{}}

\thanks{KN and JN contributed equally and are listed alphabetically.}

\begin{abstract}
The ability to verifiably retrieve transaction or state data stored off-chain is crucial to blockchain scaling techniques such as rollups or sharding. We formalize the problem and design a storage- and com\-mu\-ni\-cation-efficient protocol using linear erasure-correcting codes and homomorphic vector commitments.
Motivated by application requirements for rollups,
our solution
\emph{Semi-AVID-PR}
departs from earlier Verifiable Information Dispersal
schemes in that we do
not require comprehensive termination properties.
Compared to Data Availability Oracles,
under no circumstance do we fall back to returning empty blocks.
Distributing a file of $22\,\mathrm{MB}$ among $256$ storage nodes,
up to $85$ of which may be adversarial, requires in total
$\approx70\,\mathrm{MB}$ of communication and storage,
and $\approx41\,\mathrm{s}$ of single-thread runtime ($< 3\,\mathrm{s}$ on $16$ threads)
on an AMD Opteron 6378 processor
when using the BLS12-381 curve.
Our solution requires no modification to on-chain contracts of Validium rollups such as StarkWare's StarkEx. Additionally, it provides privacy of the dispersed data against honest-but-curious storage nodes.
Finally, we discuss an application of our Semi-AVID-PR scheme to data availability verification schemes based on random sampling.
\end{abstract}

\maketitle

\section{Introduction}
\label{sec:introduction}
\subsection{Rollups}

Ethereum, like many blockchains, suffers from poor transaction throughput and latency. To address this issue, various consensus-layer and \emph{off-chain} scaling methods were introduced. While consensus-layer solutions such as sharding \cite{polyshard,omniledger} or multi-chain protocols \cite{prism,ohie} aim at improving the base blockchain protocol, off-chain `layer 2' solutions such as payment channels \cite{decker,miller} and rollups \cite{arbitrum,validatingbridges} aim at moving transaction processing and storage off-chain.
The base blockchain 
then serves only
as a trust anchor, rollback prevention mechanism, and arbitrator in case of misbehavior and disputes among participants.
\emph{Rollups} in particular introduce an on-chain smart contract representing certain application logic, to and from which rollup users can transfer funds to enter and exit the rollup, 
and who watches over proper execution of the state machine that describes the rollup's application logic. Rollup users appoint an operator whose role is to execute the contract's state machine and keep track of updated state such as users' balances. 
For this purpose, the operator collects transactions issued by users and executes them off the main chain, but periodically posts a state snapshot to the main chain in order to irrevocably confirm transaction execution and, thus, inherit the main chain's \emph{safety} guarantee.
To ensure \emph{liveness}, rollup users need to be able to enforce application logic and to exit the rollup with their funds, even if the rollup operator turns uncooperative.
To this end, if a user presents proof of their balance according to the latest state snapshot, then the on-chain contract will pay out their funds to the user and thus enforce the user's exit.
Rollup designs differ in two crucial aspects. First, how to ensure that the state is only updated in accordance with the application logic. Second, how to guarantee that users are able to exit even if the operator turns malicious and withholds the information necessary for users to prove their balances on-chain.
\subsection{State Integrity and Data Availability}

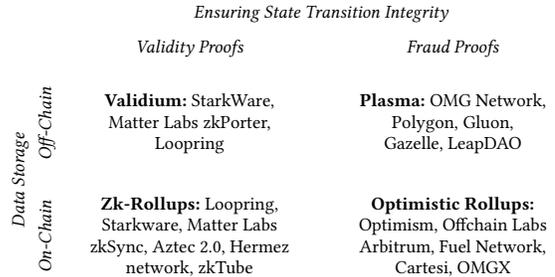
\begin{figure}[t]
    \centering
    \begin{tikzpicture}[x=3.5cm,y=1.5cm]
        \footnotesize
        
        \node[align=center] at (0,0) {\textbf{Zk-Rollups:} Loopring,\\Starkware, Matter Labs\\zkSync, Aztec 2.0, Hermez\\network, zkTube};
        \node[align=center] at (1,0) {\textbf{Optimistic Rollups:}\\Optimism,
        Offchain Labs\\Arbitrum, Fuel Network,\\
        Cartesi, OMGX};
        \node[align=center] at (0,1) {\textbf{Validium:} StarkWare,\\Matter Labs zkPorter,\\Loopring};
        \node[align=center] at (1,1) {\textbf{Plasma:} OMG Network,\\Polygon, Gluon,\\Gazelle, LeapDAO};
        
        \node[rotate=90,anchor=north] at (-0.6,0) {\emph{On-Chain}};
        \node[rotate=90,anchor=north] at (-0.6,1) {\emph{Off-Chain}};
        \node[rotate=90,anchor=north] at (-0.7,0.5) {\emph{Data Storage}};
        \node[anchor=north] at (0,1.8) {\emph{Validity Proofs}};
        \node[anchor=north] at (1,1.8) {\emph{Fraud Proofs}};
        \node[anchor=north] at (0.5,2.1) {\emph{Ensuring State Transition Integrity}};
    \end{tikzpicture}
    \caption{Layer $2$ and rollup projects grouped into four categories according to how validity of state transitions and data availability are ensured (fraud/validity proofs vs. data storage on/off chain). Source:
    \url{https://ethereum.org/en/developers/docs/scaling/}, \url{https://twitter.com/vitalikbuterin/status/1267455602764251138}}
    \label{fig:rollups}
\end{figure}

Rollup designs can be grouped into four categories, as illustrated in Figure~\ref{fig:rollups}, 
according to how they ensure validity of state transitions and availability of transaction information.
For the problem of ensuring that application logic is followed, one approach is to use \emph{fraud proofs}: anyone can re-execute the application logic on the inputs
at hand and check that the state transitions are correct. If they are not, they present proof of a fraudulent state transition to the on-chain contract which will step in as an arbitrator and enforce application logic. 
Rollups using fraud proofs are called \emph{optimistic rollups}. A second approach is based on \emph{validity proofs}, where instead of detecting fraud after the fact, fraud is prevented from the get-go by requiring the operator to provide cryptographic proof \cite{snark,stark} of proper state update.
This approach is used in \emph{zk-rollups}.

\begin{figure}[t]
    \centering
    \begin{tikzpicture}[x=0.95cm]
        \footnotesize
        
        \node[align=left,anchor=north west] at (1,0) {\textsc{Rollup Users}};
        
        \node (user1) at (2,-1.5) {\includegraphics[width=1cm]{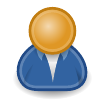}};
        \node (userDOT) at (3,-1) {...};
        \node (userN) at (4,-0.5) {\includegraphics[width=1cm]{figures/tangoicons/Emblem-person-blue.pdf}};
        
        \node[align=right,anchor=north east] at (10,0) {\textsc{Storage Nodes}};
        
        \node (storage1) at (7,-0.5) {\includegraphics[width=1cm]{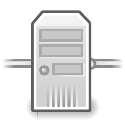}};
        \node (storageDOT) at (8,-1) {...};
        \node (storageN) at (9,-1.5) {\includegraphics[width=1cm]{figures/tangoicons/Network-server.pdf}};
        
        \node[align=left,anchor=north west] at (1,-3) {\textsc{Rollup Operator}};
        \node (stateT) at (7,-4.5) {$\mystate_t$};
        \node [gray] (stateTm1) at (4,-4.5) {$\mystate_{t-1}$};
        \draw [-Latex,dashed,gray] (2,-4.5) -- (stateTm1) node [midway,above] (txsTm1) {$\txs_{t-1}$};
        \draw [-Latex] (stateTm1) -- (stateT) node [midway,above] (txsT) {$\txs_t$};
        
        \node at ($(txsT.south)+(0,-0.5)$) {\circledScriptGray{2}};
        
        \draw [-Latex] (user1) -- (txsT);%
        \draw [draw=none] (userDOT) -- (txsT) node [pos=0.7] {...};
        \draw [-Latex] (userN) -- (txsT);%
        
        \draw [draw=none] (userDOT) -- (txsT) node [pos=0.5] {\circledScriptGray{1}};
        
        \begin{pgfonlayer}{bg1}
            \draw [fill=myParula01Blue!30,draw=none] ($(txsT.north west)-(0.3,-0.3)$) rectangle ($(stateT.south east)+(0.3,-0.3)$);
            \draw [fill=myParula01Blue!10,draw=none] ($(txsTm1.north west)-(0.1,-0.1)$) rectangle ($(stateTm1.south east)+(0.1,-0.1)$);
            \coordinate (blobT) at ($(stateT.west)+(-0.3,0.75)$);
        \end{pgfonlayer}
        
        \coordinate (blobTs) at ($(stateT.west)+(-0.3,-0.55)$);
        \coordinate (blobTm1s) at ($(stateTm1.west)+(-0.1,-0.35)$);
        
        \draw [-Latex,shorten <=0.5em] (blobT) -- (storage1) node [midway,left,xshift=-0.1em]  {\circledScriptGray{3}};
        \draw [draw=none] (blobT) -- (storageDOT) node [pos=0.25] {...};
        \draw [-Latex,shorten <=0.5em] (blobT) -- (storageN);
        
        \node (certT) at ($(stateT.east)+(0.3,0.5)$) {\includegraphics[width=1cm]{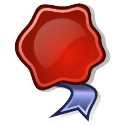}};
        
        \draw [-Latex] (storage1) -- (certT);
        \draw [draw=none] (storageDOT) -- (certT) node [pos=0.8] {...};
        \draw [-Latex] (storageN) -- (certT) node [midway,right,xshift=0.2em] {\circledScriptGray{4}};
        
        \node[align=left,anchor=north west] at (1,-5.5) {\textsc{Main Chain}};
        
        \node [fill=myParula03Yellow!20,draw=myParula03Yellow!50,minimum width=1.5cm,minimum height=0.75cm] (blkTm1) at (3.5,-6.75) {};
        \node [fill=myParula03Yellow!50,draw=myParula03Yellow!100,minimum width=1.5cm,minimum height=0.75cm] (blkT) at (6.5,-6.75) {};
        
        \draw [-Latex,ultra thick,myParula03Yellow!50,shorten >=0.5em,shorten <=0.5em,dashed] (blkTm1) -- (1.5,-6.75);
        \draw [-Latex,ultra thick,myParula03Yellow!100,shorten >=0.5em,shorten <=0.5em] (blkT) -- (blkTm1);
        
        \draw [-Latex,ultra thick,myParula01Blue!50,shorten >=0.5em,shorten <=0.5em] (blobTs) -- (blkT) node [midway,right,xshift=0.4em] {\circledScriptGray{5}};
        \draw [-Latex,ultra thick,myParula01Blue!20,shorten >=0.5em,shorten <=0.5em] (blobTm1s) -- (blkTm1);
        
        \node at ($(blkT)$) {\circledScriptGray{6}};
        
    \end{tikzpicture}
    \caption[]{Normal operating mode of a \emph{Validium rollup}, \ie, a zk-rollup with off-chain storage:
    \circledScriptGraySlim{1} Rollup operator collects transactions from rollup users.
    \circledScriptGraySlim{2} Transactions are executed by operator off-chain and new state of rollup is calculated.
    \circledScriptGraySlim{3} Transaction and state data is dispersed to storage nodes by operator.
    \circledScriptGraySlim{4} Storage nodes confirm receipt.
    Operator collects sufficient number of confirmations into certificate of data retrievability.
    \circledScriptGraySlim{5} Operator sends commitment to state and certificate of retrievability to main chain.
    \circledScriptGraySlim{6} State snapshot and certificate of retrievability are verified and if valid accepted by main chain.}
    \label{fig:validium}
\end{figure}

To ensure that rollup users are able to 
track proper execution of application rules and to prove their balances,
there are again two approaches. 
The relevant information
could be made available either on the main chain, perhaps in a condensed form, or stored off the main chain
but with some credible assurance that the data is in fact available for users to retrieve.
For the latter purpose, \emph{Validium rollups}, \ie, zk-rollups with off-chain storage
such as StarkWare's `StarkEx', introduce a committee of trusted storage nodes.
This way, Validium rollups avoid
both the settlement delays introduced by fraud proofs'
dispute period,
and
the cost and scalability bottleneck of on-chain data storage;
making them
a promising rollup construction
and attractive subject of study.
For the normal operating mode of a Validium rollup 
see
Figure~\ref{fig:validium}. The rollup operator deposits a copy of the relevant data with each storage node, who in turn confirm receipt. A valid state snapshot is accepted by the main chain only if enough storage nodes have confirmed receipt of the corresponding 
full data.
As long as enough storage nodes remain honest and available, rollup users can always
turn to them
to obtain the 
data necessary to prove fraud or balances,
should the rollup operator withhold it.

This solution, however, is not communication- or storage-efficient. The operator has to send a copy of the entire data to every storage node which in turn stores an entire copy. Therefore, this solution is not scalable and works only for a relatively small number of storage nodes (\eg, a current application of the StarkWare Validium rollup uses $8$ storage nodes \cite{dacvalidium}), 
which leads to heavy centralization.
Furthermore, the privacy of user data is violated, as storage nodes can view the entire state.

\subsection{Information Dispersal with Provable Retrievability}

A more communication- and storage-efficient solution is provided by
Verifiable Information Dispersal (VID) as embodied by Asynchronous Verifiable Information Dispersal (AVID \cite{avid}) and its successors AVID-FP \cite{avidfp} and AVID-M \cite{avidm}.
Generally speaking, AVID schemes encode the input data block into chunks
and every storage node has to store only one chunk
rather than the full block.
The correctness of the dispersal is verifiable, meaning that the consistency of chunks is ensured.
A VID scheme consists of two protocols,
$\mathsf{Disperse}$ and $\mathsf{Retrieve}$,
satisfying \cite{avidfp}, informally:
\begin{enumerate}[label={(\alph*)}]
    \item \textbf{Termination.} If $\mathsf{Disperse}(B)$ is initiated by an honest client, then $\mathsf{Disperse}(B)$ is eventually completed by all honest storage nodes.
    \item \textbf{Agreement.} If some honest storage node completes $\mathsf{Disperse}(B)$, all honest storage nodes eventually complete $\mathsf{Disperse}(B)$.
    \item \textbf{Availability.} If `enough' honest storage nodes complete $\mathsf{Disperse}(B)$, an honest client that initiates $\mathsf{Retrieve}()$ eventually reconstructs some block $B'$.
    \item \textbf{Correctness.} After `enough' honest storage nodes complete $\mathsf{Disperse}(B)$, all honest clients that initiate $\mathsf{Retrieve}()$ eventually retrieve the same block $B'$. If the client that initiated $\mathsf{Disperse}(B)$ was honest, then $B' = B$.
\end{enumerate}

Although some existing VID schemes can be used to ensure
data availability for rollups,
they miss properties that are required for this application, while having others that are not needed, resulting in unnecessary complexity
(\cf Figure~\ref{fig:development-avid-schemes}).
For Validium rollups, it is crucial 
that the on-chain rollup contract
can \emph{verify}
(%
\tikz[baseline=(char.base)]{
\node[text=myParula04Purple,fill=myParula04Purple!10,draw=myParula04Purple,shape=circle,inner sep=1pt] (char) {\scriptsize $4$};},
\tikz[baseline=(char.base)]{
\node[text=myParula05Green,fill=myParula05Green!10,draw=myParula05Green,shape=circle,inner sep=1pt] (char) {\scriptsize $5$};}%
)
the retrievability of the underlying data
before accepting a new state update, to ensure that
users have access to the data required to enforce the contract (or be able to exit)
on-chain in case of an uncooperative operator.
For this purpose, consistent retrieval of `some' block
$B' \neq B$ is not enough, the \emph{retrievability}
(%
\tikz[baseline=(char.base)]{
\node[text=myParula03Yellow,fill=myParula03Yellow!10,draw=myParula03Yellow,shape=circle,inner sep=1pt] (char) {\scriptsize $3$};}%
)
of \emph{the original block $B$}
needs to be ensured.
Oppositely, comprehensive termination properties
(%
\tikz[baseline=(char.base)]{
\node[text=myParula02Orange,fill=myParula02Orange!10,draw=myParula02Orange,shape=circle,inner sep=1pt] (char) {\scriptsize $2$};}%
)
such as Termination and Agreement are not needed,
and some VID schemes (here AVID) provide properties
exceeding VID that are not required for the rollup application
(%
\tikz[baseline=(char.base)]{
\node[text=myParula01Blue,fill=myParula01Blue!10,draw=myParula01Blue,shape=circle,inner sep=1pt] (char) {\scriptsize $1$};}%
).

We introduce the concept of \emph{Semi-AVID with Provable Retrievability} (Semi-AVID-PR)
to capture the requirements in the rollup application.
Besides $\savidDisperse$ and $\savidRetrieve$, a 
Semi-AVID-PR scheme provides $\savidCommit$
to succinctly and unequivocally identify data blocks
and $\savidVerify$ to verify certificates of retrievability.
\begin{definition*}[Semi-AVID-PR Security (Informal), \cf Definition~\ref{def:semiavidpr-security}]
If $f \leq t$ nodes are Byzantine, then
the Semi-AVID-PR scheme provides:
\begin{enumerate}[label={(\alph*)}]
    \item
        \textbf{Commitment-Binding.}
        $\savidCommit$ is a binding deterministic commitment to a block of data.
        
    \item
        \textbf{Correctness.}
        If an honest client initiates $\savidDisperse(B)$,
        then eventually it obtains a valid certificate of retrievability
        for $\savidCommit(B)$.

    \item
        \textbf{Availability.}
        If an honest client invokes $\savidRetrieve(P, C)$
        with a valid certificate of retrievability $P$
        for commitment $C$,
        then eventually it obtains a block $B$
        such that $\savidCommit(B) = C$.
        
\end{enumerate}
\end{definition*}
Note that availability is ensured even for
certificates of retrievability generated
under \emph{adversarial} dispersal
(\ie, malicious rollup operator).
We provide formal game-based definitions of commitment-binding (\cf Definition~\ref{alg:game-savid-binding}) and availability (\cf Definition~\ref{alg:game-savid-availability}).

\begin{figure}[t]
    \centering
    \begin{tikzpicture}[x=2.5cm,y=0.6cm]
        \footnotesize
        
        \node (avid) at (0,0) {AVID \cite{avid}};
        \node at (0,0.8) {%
            \tikz[baseline=(char.base)]{
            \node[text=myParula01Blue,fill=myParula01Blue!10,draw=myParula01Blue,shape=circle,inner sep=1pt] (char) {\scriptsize $1$};}
            \tikz[baseline=(char.base)]{
            \node[text=myParula02Orange,fill=myParula02Orange!10,draw=myParula02Orange,shape=circle,inner sep=1pt] (char) {\scriptsize $2$};}
            \tikz[baseline=(char.base)]{
            \node[text=myParula03Yellow,fill=myParula03Yellow!10,draw=myParula03Yellow,shape=circle,inner sep=1pt] (char) {\scriptsize $3$};}
            (\tikz[baseline=(char.base)]{
            \node[text=myParula04Purple,fill=myParula04Purple!10,draw=myParula04Purple,shape=circle,inner sep=1pt] (char) {\scriptsize $4$};})%
        };
        
        \node (avidfp) at (1,0) {AVID-FP \cite{avidfp}};
        \node at (1,0.8) {%
            \tikz[baseline=(char.base)]{
            \node[text=myParula02Orange,fill=myParula02Orange!10,draw=myParula02Orange,shape=circle,inner sep=1pt] (char) {\scriptsize $2$};}
            \tikz[baseline=(char.base)]{
            \node[text=myParula03Yellow,fill=myParula03Yellow!10,draw=myParula03Yellow,shape=circle,inner sep=1pt] (char) {\scriptsize $3$};}
            (\tikz[baseline=(char.base)]{
            \node[text=myParula04Purple,fill=myParula04Purple!10,draw=myParula04Purple,shape=circle,inner sep=1pt] (char) {\scriptsize $4$};})%
        };
        
        \node (aced) at (2,2) {ACeD \cite{aced}};
        \node [align=left,anchor=west] at (2.35,2) {%
            \tikz[baseline=(char.base)]{
            \node[text=myParula02Orange,fill=myParula02Orange!10,draw=myParula02Orange,shape=circle,inner sep=1pt] (char) {\scriptsize $2$};}
            \tikz[baseline=(char.base)]{
            \node[text=myParula05Green,fill=myParula05Green!10,draw=myParula05Green,shape=circle,inner sep=1pt] (char) {\scriptsize $5$};}%
        };
        
        \node[align=center] (dumbomvba) at (2,1) {APDB \cite{dumbo}};
        \node [align=left,anchor=west] at (2.35,1) {%
            \tikz[baseline=(char.base)]{
            \node[text=myParula04Purple,fill=myParula04Purple!10,draw=myParula04Purple,shape=circle,inner sep=1pt] (char) {\scriptsize $4$};}
        };
        
        \node (avidm) at (2,0) {AVID-M \cite{avidm}};
        \node [align=left,anchor=west] at (2.35,0) {%
            \tikz[baseline=(char.base)]{
            \node[text=myParula02Orange,fill=myParula02Orange!10,draw=myParula02Orange,shape=circle,inner sep=1pt] (char) {\scriptsize $2$};}
            (\tikz[baseline=(char.base)]{
            \node[text=myParula04Purple,fill=myParula04Purple!10,draw=myParula04Purple,shape=circle,inner sep=1pt] (char) {\scriptsize $4$};})%
        };
        
        \node [align=center] (savidpr) at (2,-1.2) {Semi-AVID-PR\\(this work)};
        \node [align=left,anchor=west] at (2.35,-1.2) {%
            \tikz[baseline=(char.base)]{
            \node[text=myParula03Yellow,fill=myParula03Yellow!10,draw=myParula03Yellow,shape=circle,inner sep=1pt] (char) {\scriptsize $3$};}
            \tikz[baseline=(char.base)]{
            \node[text=myParula04Purple,fill=myParula04Purple!10,draw=myParula04Purple,shape=circle,inner sep=1pt] (char) {\scriptsize $4$};}%
        };
        
        \draw [-Latex] (avid) -- (avidfp);
        \draw [-Latex] (avidfp) -- (avidm);
        \draw [-Latex] (avidfp) -- (savidpr);

    \end{tikzpicture}
    \caption[]{
        Related protocols and supported properties:
        \tikz[baseline=(char.base)]{
        \node[text=myParula01Blue,fill=myParula01Blue!10,draw=myParula01Blue,shape=circle,inner sep=1pt] (char) {\scriptsize $1$};}
        Retrieval from \emph{any} sufficiently large set of storage nodes
        \tikz[baseline=(char.base)]{
        \node[text=myParula02Orange,fill=myParula02Orange!10,draw=myParula02Orange,shape=circle,inner sep=1pt] (char) {\scriptsize $2$};}
        Comprehensive termination guarantees
        \tikz[baseline=(char.base)]{
        \node[text=myParula03Yellow,fill=myParula03Yellow!10,draw=myParula03Yellow,shape=circle,inner sep=1pt] (char) {\scriptsize $3$};}
        Retrievability guaranteed
        \tikz[baseline=(char.base)]{
        \node[text=myParula04Purple,fill=myParula04Purple!10,draw=myParula04Purple,shape=circle,inner sep=1pt] (char) {\scriptsize $4$};}
        Issues certificates of retrievability
        \tikz[baseline=(char.base)]{
        \node[text=myParula05Green,fill=myParula05Green!10,draw=myParula05Green,shape=circle,inner sep=1pt] (char) {\scriptsize $5$};}
        Dispersal verifiable on-chain%
    }
    \label{fig:development-avid-schemes}
\end{figure}

\begin{figure}[t]
    \centering
    \begin{tikzpicture}[x=0.9cm]
        \footnotesize
        
        \node[align=left,anchor=north west] at (0,0) {\textsc{Rollup Operator}};
        
        \node [minimum width=2.25mm,minimum height=1.5cm,draw=black] (u1) at (0.3,-1.5) {};
        \node [minimum width=2.25mm,minimum height=1.5cm,draw=black] (u2) at (0.6,-1.5) {};
        \node [minimum width=2.25mm,minimum height=1.5cm,draw=black] (u3) at (0.9,-1.5) {};
        \node [minimum width=2.25mm,minimum height=1.5cm,draw=black] (u4) at (1.2,-1.5) {};
        \node [minimum width=2.25mm,minimum height=1.5cm,draw=black] (u5) at (1.5,-1.5) {};
        
        \node [minimum width=2.25mm,minimum height=1.5cm,draw=black] (c1) at (3.3,-1.5) {};
        \node [minimum width=2.25mm,minimum height=1.5cm,draw=black] (c2) at (3.6,-1.5) {};
        \node [minimum width=2.25mm,minimum height=1.5cm,draw=black] (c3) at (3.9,-1.5) {};
        \node [minimum width=2.25mm,minimum height=1.5cm,draw=black] (c4) at (4.2,-1.5) {};
        \node [minimum width=2.25mm,minimum height=1.5cm,draw=black,fill=myParula01Blue!30] (c5) at (4.5,-1.5) {};
        \node [minimum width=2.25mm,minimum height=1.5cm,draw=black] (c6) at (4.8,-1.5) {};
        \node [minimum width=2.25mm,minimum height=1.5cm,draw=black] (c7) at (5.1,-1.5) {};
        
        \node (commitments) at (0.9,-3.4) {Commitments};
        
        \draw [-Latex] (1.8,-0.9) -- (3,-0.9);
        \draw [-Latex] (1.8,-1.1) -- (3,-1.1);
        \draw [-Latex] (1.8,-1.3) -- (3,-1.3);
        \draw [-Latex] (1.8,-1.5) -- (3,-1.5);
        \draw [-Latex] (1.8,-1.7) -- (3,-1.7);
        \draw [-Latex] (1.8,-1.9) -- (3,-1.9);
        \draw [-Latex] (1.8,-2.1) -- (3,-2.1);
        
        \node [] at (2.3,-1.5) {\circledScriptGray{2}};
        
        \node [fill=white] at (0.9,-1.5) {Data};
        
        \draw [-Latex,shorten <=0.5em] (u1) -- (commitments.north -| u1);
        \draw [-Latex,shorten <=0.5em] (u2) -- (commitments.north -| u2);
        \draw [-Latex,shorten <=0.5em] (u3) -- (commitments.north -| u3);
        \draw [-Latex,shorten <=0.5em] (u4) -- (commitments.north -| u4);
        \draw [-Latex,shorten <=0.5em] (u5) -- (commitments.north -| u5);
        
        \draw [draw=none] (u3) -- (commitments.north -| u3) node [midway] {\circledScriptGray{1}};
        
        \node [rotate=90,font=\scriptsize] (chunk) at (4.5,-1.5) {Chunk $i$};
        
        \node[align=right,anchor=north east] at (9,0) {\textsc{Storage Node $i$}};
        
        \node [align=center] (verify) at (7.75,-3.4) {Verify consistency\\of chunk $i$ with\\commitments};
        
        \draw [-Latex] (commitments) -- (verify);
        \draw [-Latex] (c5.south) -- (verify) node [midway,below] {\circledScriptGray{3}};
        
        \node (certificate) at (0.9,-4.5) {\includegraphics[width=1cm]{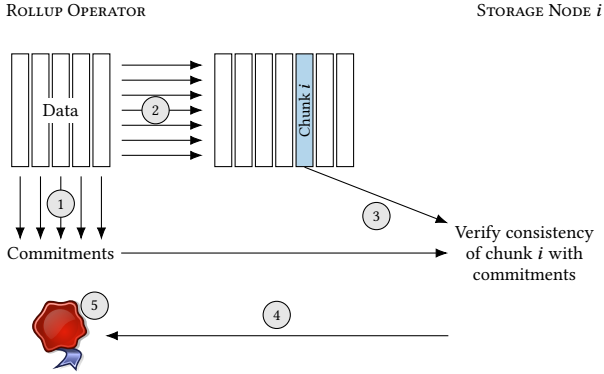}};
        
        \draw [-Latex] (verify.west |- certificate) -- (certificate) node [midway,above] {\circledScriptGray{4}};
        
        \node [] at (1.4,-4.1) {\circledScriptGray{5}};
        
    \end{tikzpicture}
    \caption[]{Dispersal in our Semi-AVID-PR scheme.
    \circledScriptGraySlim{1} Client arranges data in matrix and computes vector commitments of columns.
    \circledScriptGraySlim{2} Client encodes data row-wise (same code for each row).
    \circledScriptGraySlim{3} Commitments and chunk $i$ are sent to storage node $i$.
    \circledScriptGraySlim{4} If chunk $i$ is consistent with commitments, storage node confirms receipt.
    \circledScriptGraySlim{5} Enough acknowledgements form certificate of retrievability.}
    \label{fig:savidpr-disperse-overview}
\end{figure}

We propose a construction  for a communication- and storage-efficient Semi-AVID-PR scheme
with practical computational cost, which is compatible with the established on-chain smart contracts of and thus can be readily adopted for existing Validium rollups, \eg, such as StarkWare's StarkEx. Our construction relies on a collision-resistant hash function, unforgeable signatures, and a deterministic homomorphic vector commitment. We provide a reduction-based security proof.
A high-level illustration of $\savidDisperse$ of our scheme is provided in Figure~\ref{fig:savidpr-disperse-overview}.
In our protocol, the rollup operator computes commitments to chunks of the initial data and encodes it using an erasure-correcting code. Encoded chunks are dispersed among the storage nodes along with the commitments. Similarly to AVID-FP, the commitments allow storage nodes to verify the consistency of their local chunk with the file for which they are about to acknowledge the receipt of a chunk. If their chunk is consistent, a storage node confirms receipt to the operator. Upon collecting enough confirmations, the operator can produce a certificate of retrievability for the respective file, which is later verified by the main chain before accepting the new state snapshot.
Additional blinding can be used in our scheme to provide privacy against honest-but-curious storage nodes.
Finally, the core construction of our scheme can be used to derive a data availability verification scheme based on random sampling with practical computational requirements.

\subsection{Related Work}

The AVID protocol \cite{avid} satisfies not only the VID properties,
but furthermore guarantees that eventually a dispersed file
can be retrieved from \emph{any} subset containing $k$ honest storage nodes
(\cf Figure~\ref{fig:development-avid-schemes}, \tikz[baseline=(char.base)]{
\node[text=myParula01Blue,fill=myParula01Blue!10,draw=myParula01Blue,shape=circle,inner sep=1pt] (char) {\scriptsize $1$};}),
rather than from only \emph{some} subset that can be identified from the certificate of retrievability, as for Semi-AVID-PR
(following earlier successors of AVID which also dispense with this excess guarantee).
This is achieved by an additional round of echoing chunks which leads
to a high communication cost for AVID (\cf Table~\ref{tab:comparison}), whose
communication complexity 
is $O(n |B| + n^3 \lambda)$, 
with block size $|B|$,
number of storage nodes $n$,
and
$\lambda$-bit security parameter of the underlying
cryptographic primitives.
Our Semi-AVID-PR scheme's
communication and storage complexity is $O(|B| + n^2 \lambda)$.

AVID-FP \cite{avidfp},
a successor of AVID,
brings the communication complexity to $O(|B| + n^3 \lambda)$
using homomorphic fingerprinting (as in our Semi-AVID-PR scheme) so that storage nodes can verify
their chunks without echoing them.
Since chunks can be verified, retrievability is guaranteed
(\tikz[baseline=(char.base)]{
\node[text=myParula03Yellow,fill=myParula03Yellow!10,draw=myParula03Yellow,shape=circle,inner sep=1pt] (char) {\scriptsize $3$};}).
A round of AVID (with the fingerprints) is still used
to achieve comprehensive termination properties
(\tikz[baseline=(char.base)]{
\node[text=myParula02Orange,fill=myParula02Orange!10,draw=myParula02Orange,shape=circle,inner sep=1pt] (char) {\scriptsize $2$};}).

A recent advancement of VID, AVID-M \cite{avidm}, improves the communication complexity to $O(|B| + n^2 \lambda)$ by reducing the size of fingerprints
using Merkle trees \cite{merkle}.
However, chunks cannot be verified anymore,
and AVID-M retrieves as empty block any data that was maliciously
encoded during dispersal, so that retrievability is no longer guaranteed.
This makes AVID-M less suitable for application to rollups.
The situation is similar for Data Availability Oracles (DAOs) such as ACeD \cite{aced},
and for Asynchronous Provable Dispersal Broadcast (APDB), the primitive implemented by the dispersal sub-protocol of Dumbo-MVBA \cite{dumbo}. If a block was invalidly encoded during dispersal by a malicious client,
then AVID-M,
DAOs and APDB 
ensure consistency across retrieving clients,
but no guarantee is provided about how the retrieved content
relates to the `dispersed content'.
In contrast, for Semi-AVID-PR, 
$\savidCommit$
links
available dispersed content and retrieved content,
and thus ensures that data availability and
state transition integrity operate
on the same content.
We think of
exhibiting
$\savidCommit$ as \emph{the `interface'}
where data availability and state transition integrity `come together' in the rollup application
as a key contribution of this work.
DAOs
and protocols from the VID family
differ in terms of how the rollup's on-chain contract can verify
completion of the dispersal.
Semi-AVID-PR issues certificates of retrievability
(\tikz[baseline=(char.base)]{
\node[text=myParula04Purple,fill=myParula04Purple!10,draw=myParula04Purple,shape=circle,inner sep=1pt] (char) {\scriptsize $4$};}%
)
that can be independently verified (\eg, by a smart contract),
similar to Dumbo-MVBA's dispersal where
the dispersing client produces a `lock proof' consisting of a quorum of signatures from storage nodes attesting to having received their respective chunks of the dispersed data.
AVID, AVID-FP and AVID-M can readily be extended with such functionality.
DAOs report data availability
directly on-chain via a smart contract
(\tikz[baseline=(char.base)]{
\node[text=myParula05Green,fill=myParula05Green!10,draw=myParula05Green,shape=circle,inner sep=1pt] (char) {\scriptsize $5$};}%
).
Since AVID, AVID-FP and AVID-M all perform a round of VID to achieve comprehensive termination
properties, the AVID family is limited to an adversarial resilience $t < n/3$, in contrast to $t < n/2$ for Semi-AVID-PR, APDB
and ACeD.

Unlike proofs of retrievability
\cite{proofofretrievability}
and proof of replication
\cite{proofofreplication},
the goal in this work is
not to test whether a storage
system has custody of
certain intact data
and/or whether a storage system
actually delivers the data upon request,
as would arguably be the starting
point to study how to incentivize
data storage and retrieval in
the honest-but-rational model
resulting from open-participation
systems,
a surely interesting problem in its own right.
Rather,
we leave incentives for future work,
and focus
in this work
on communication,
storage, and compute efficient
dispersal in the setting
of honest vs. Byzantine participants,
where honest nodes follow
the protocol as specified.
In practice, for Validium rollups,
earlier dispersed data
becomes deprecated
as state/transaction information
is `overwritten' by later
dispersed data,
so that honest storage nodes are only required to
make available the data
of recent dispersals (on the order of days
or weeks).
\paragraph{Sampling Based Data Availability Checks} The data availability problem is not unique to rollups, but arises in other scaling approaches such as sharding or light clients as well.
Solutions like \cite{albassam,cmt} provide interactive protocols based on random sampling of
chunks of erasure-coded data. A block is deemed available if all randomly sampled chunks are available, with the assumption being that if enough nodes' random queries are answered, then enough chunks are available to restore the block.
However, this interactive technique is
not feasible for rollups since the on-chain contract cannot engage in random
sampling to convince itself of data availability.
Instead, the assurance of data availability could either be made on-chain
(as by Data Availability Oracles) or in the form of
a verifiable certificate of retrievability.
However, techniques from our Semi-AVID-PR scheme can be used to obtain a data availability check where the consistency of a randomly sampled chunk can be efficiently verified, obviating fraud proofs for invalid encoding \cite{albassam}.

\subsection{Outline}

Cryptographic essentials and erasure-correcting codes are reviewed
in Section~\ref{sec:preliminaries}.
Model and formal properties of
Semi-AVID-PR 
for the application in rollups are
introduced in Section~\ref{sec:model}.
We describe our Semi-AVID-PR protocol in Section~\ref{sec:protocol}
and prove its security
in Section~\ref{sec:security-argument}.
Use of blinding to protect privacy of dispersed data 
against honest-but-curious storage nodes is discussed
in Section~\ref{sec:privacy}.
An evaluation of computational cost and storage- and communication-efficiency
in comparison to other schemes is discussed in Section~\ref{sec:evaluation}.
We close with comments on an application of techniques of our Semi-AVID-PR
scheme
to data availability sampling
in Section~\ref{sec:data-availability-sampling}.

\section{Preliminaries}
\label{sec:preliminaries}
In this section, we briefly recapitulate
tools from cryptography and erasure-correcting codes
used throughout the paper.

\subsection{Basics \& Notation}
\label{sec:preliminaries-basics}

Let $\IG$ be a cyclic group (denoted multiplicatively,
\ie, with group operation~`$\cdot$')
of prime order $p \geq 2^{2\lambda}$
with generator $g \in \IG$,
where $\lambda$ denotes the security parameter used
subsequently for all primitives.
The function $H(x) \triangleq g^x$ is a bijection between the finite field $\IZ_p$ (\ie, integers modulo $p$) and $\IG$.
It has the \emph{linear
homomorphism} property
which this paper makes ample use of:
\begin{IEEEeqnarray}{l}%
    \forall n \geq 1\colon
    \forall c_1, ..., c_n \in \IZ_p\colon
    \forall x_1, ..., x_n \in \IZ_p\colon   \nonumber\\
    H\left( \sum_{i=1}^n c_i x_i \right) = \prod_{i=1}^n H(x_i)^{c_i}
    \IEEEeqnarraynumspace
\end{IEEEeqnarray}
A hash function $\crhf = (\mathsf{Gen}, \mathsf{H})$, see Definition~\ref{def:hash_func},
is \emph{collision resistant} 
if for any probabilistic poly-time (PPT) adversary $\mathcal A$ there exists a negligible function $\negl(.)$ such that
\begin{IEEEeqnarray}{C}
    \Prob{\gameCF_{\crhf, \CA}(\lambda) = \TRUE} \leq \negl(\lambda),
    \IEEEeqnarraynumspace
\end{IEEEeqnarray}
where $\gameCF_{\crhf, \CA}(\lambda)$ is the collision finding game recapitulated in Alg.~\ref{alg:game-crhf-collision}.
We use $\CRHF(x) \triangleq \crhfH^s(x)$ as a notational shorthand.

A signature scheme $\sig = (\mathsf{KeyGen}, \mathsf{Sign}, \mathsf{Verify})$, see Definition~\ref{def:sig}, is \emph{secure under existential forgery}
if for any PPT adversary $\mathcal A$ there exists a negligible function $\negl(.)$ such that
\begin{IEEEeqnarray}{C}
    \Prob{\gameEF_{\sig, \CA}(\lambda) = \TRUE} \leq \negl(\lambda),
    \IEEEeqnarraynumspace
\end{IEEEeqnarray}
where $\gameEF_{\sig, \CA}(\lambda)$ is the existential forgery game 
of 
Alg.~\ref{alg:game-sig-forgery}.

We denote by $[\Bx]_i$ the $i$-th entry of a vector $\Bx$,
by $[\BX]_{i}$ the $i$-th column of a matrix $\BX$,
and by $[n] \triangleq \{ 1, ..., n \}$.
In algorithms,
$\Pi[f(.)/g(.)]$ denotes substitution of $f(.)$ for $g(.)$ in the code of $\Pi$.

\subsection{Reed-Solomon Codes}
\label{sec:preliminaries-rs-codes}

A \emph{linear $(n,k)$-code} is a linear mapping
$\IZ_p^k \to \IZ_p^n$ with $n \geq k$. It can be represented
by a $k \times n$ \emph{generator matrix} $\BG$, with the encoding
operation then $\Bc^\top = \BG.\mathsf{Encode}(\Bu^\top) \triangleq \Bu^\top \BG$
to obtain a length-$n$ row vector of codeword symbols $\Bc^\top$
from a length-$k$ row vector of information symbols $\Bu^\top$.

A linear code is \emph{maximum distance separable} (MDS) if
any $k$ columns of its generator matrix $\BG$
are linearly independent, \ie, any $k \times k$ submatrix of $\BG$
is invertible. Thus, any set of codeword symbols $c_{i_j}$
from $k$ distinct indices $i_j$ can be used to uniquely decode using the relation
induced by the encoding $\Bc^\top = \Bu^\top \BG$,
\begin{IEEEeqnarray}{C}
    \Bu^\top
    \underbrace{
    \setlength{\arraycolsep}{2pt}\begin{bmatrix}
        \Bg_{i_1} & ... & \Bg_{i_k}
    \end{bmatrix}
    }_{
    \triangleq \tilde\BG
    }
    \overset{!}{=}
    \underbrace{
    \setlength{\arraycolsep}{2pt}\begin{bmatrix}
        c_{i_1} & ... & c_{i_k}
    \end{bmatrix}
    }_{
    \triangleq \tilde\Bc^\top
    }
    \nonumber\\
    \iff   %
    \Bu^\top = 
    \BG.\mathsf{Decode}(((i_j, c_{i_j}))_{j=1}^k) \triangleq \tilde\Bc^\top \tilde\BG^{-1},
    \IEEEeqnarraynumspace
\end{IEEEeqnarray}
where $\Bg_i$ corresponds to the $i$-th column of the generator matrix $\BG$.

Reed-Solomon codes \cite{rs} are an important class of MDS codes.
Here, an information vector $\Bu^\top$ is associated with a polynomial
$U(X) = \sum_{j=1}^k [\Bu^\top]_i X^{i-1}$
and the codeword vector is obtained by evaluating $U(X)$ at
$n$ distinct locations $\alpha_1, ..., \alpha_n$, such that
$\Bc^\top = (U(\alpha_1), ..., U(\alpha_n))^\top$.
This corresponds to a generator matrix $\BG_{\mathrm{RS}}$
with columns $\Bg_{\mathrm{RS},i} = (\alpha_i^0, ..., \alpha_i^{k-1})$.

\subsection{Linear Vector Commitment Schemes}
\label{sec:preliminaries-slvc}

A deterministic \emph{vector commitment} (VC) scheme  $\mathsf{VC} = (\mathsf{Setup}, \mathsf{Commit},\allowbreak \mathsf{OpenEntry}, \mathsf{VerifyEntry})$ \cite{vcs,merkle}
for vectors of length $L$
allows to commit to an element of $\IZ_p^L$. Later, the commitment can be compared to the commitment of another vector to check a vector opening, and it can be opened to individual entries of the vector.
Ideally, the proof for the opening of an entry of the vector
is short and computationally easy to generate and verify.
For our purposes it is important that the  VC is binding, \ie, 
a commitment cannot be opened to values that are inconsistent
with the committed vector.
Specifically, we call a $\mathsf{VC}$, see Definition~\ref{def:vc}, \emph{binding}
if for any PPT adversary $\mathcal A$ there exists a negligible function $\negl(.)$ such that
\begin{IEEEeqnarray}{C}
    \Prob{\gameVCB_{\slvc, \CA}(\lambda) = \TRUE} \leq \negl(\lambda).
    \IEEEeqnarraynumspace
\end{IEEEeqnarray}
where $\gameVCB_{\slvc, \CA}(\lambda)$ is the binding game defined in Alg.~\ref{alg:game-lvc-binding}.
We use $\VC(\Bv) \triangleq \slvcCommit(\Bv)$ as a notational shorthand.

For this manuscript of particular interest are
linearly homomorphic (also simply called \emph{linear}) VCs
($\slvc$)
with
\begin{IEEEeqnarray}{l}
    \forall \alpha, \beta \in \IZ_p\colon
    \forall \Bv, \Bw \in \IZ_p^L\colon   \nonumber\\
    \mathsf{Commit}(\alpha \Bv + \beta \Bw) = \alpha \mathsf{Commit}(\Bv) + \beta \mathsf{Commit}(\Bw).
    \IEEEeqnarraynumspace
\end{IEEEeqnarray}

Kate-Zaverucha-Goldberg (KZG) polynomial commitments \cite{kate}
(here the `basic' variant $\mathsf{PolyCommit}_{\mathbf{DL}}$ of \cite{kate}
as $\KZG$)
can be readily turned into an example linear VC,
which we use subsequently and introduce here briefly.
From a vector $\Bu$ of length $L$ interpolate a polynomial
$U(X)$ of degree $(L-1)$ such that
$U(i) = [\Bu]_i$ for $i=1,...,L$.
Commit to $\Bu$ by $\KZGcommit(U)$.\footnote{The polynomial
interpolation can be avoided
by preprocessing
the public parameters of KZG to obtain them in the Lagrange polynomial basis.}
The vector opening can be verified by recomputing the commitment.
The entry $[\Bu]_i$ can be opened and the opening verified
using $\KZGproof$ and $\KZGverifyeval$ for the corresponding
$U(X)$ at $X=i$, respectively.

To see that the resulting VC's $\mathsf{Commit}$ is linear,
consider this.
During trusted setup, $\KZGsetup$ samples $r \drawrandom \IZ_p$
and computes public parameters $(g^{r^0}, ..., g^{r^{L-1}})$.
$\KZGcommit$ computes the commitment to a polynomial $U(X)$
of degree $(L-1)$
with coefficients $\gamma_0, ..., \gamma_{L-1}$
as $g^{U(r)}$ which, due to the linear homomorphism of $H(x) = g^x$,
can be obtained from the public parameters as
\begin{IEEEeqnarray}{C}
    \KZGcommit(\gamma_0, ..., \gamma_{L-1}) = \KZGcommit(U) = \prod_{j=0}^{L-1} (g^{r^j})^{\gamma_j}.
    \IEEEeqnarraynumspace
\end{IEEEeqnarray}
Since
interpolation of
coefficients
$\Bgamma = (\gamma_0, ..., \gamma_{L-1})$
of $U(X)$
from a vector $\Bu$ such that
$U(i) = [\Bu]_i$ for $i=1,...,L$
is linear and invertible,
$\mathsf{Commit}$ is linear.
When using 
KZG-based LVCs,
our protocol
inherits its trusted setup.
Parameters generated
in recent ce\-re\-mo\-nies
can be reused
(\cf Zcash's or Filecoin's `Powers of Tau' \cite{powersoftau}).

\section{Model}
\label{sec:model}

The system under discussion consists of $n$ \emph{storage nodes} $P_1, ..., P_n$
and some \emph{clients}.
A PPT adversary
can corrupt protocol participants adaptively,
\ie, as the protocol execution progresses.
Corrupt participants surrender their internal state to the adversary immediately
and from thereon behave as coordinated by the adversary.
We denote by $f$ the number of storage nodes corrupted over the course of the execution,
and by $t$ the design resilience,
\ie, our construction is parametric in $t$
and satisfies the desired security properties
in all executions with $f \leq t$.
Protocol participants can send each other messages (a priori without sender
identification)
which undergo delay controlled by the adversary,
subject to the constraint that every message has to arrive eventually.
We design a scheme with
the following interface and security properties.
\begin{definition}[Semi-AVID-PR Syntax]
\label{def:semiavidpr-syntax}
A \emph{Semi-AVID (Asynchronous Verifiable Information Dispersal) Scheme with Provable Retrievability} (\savid)
consists of two algorithms,
$\savidCommit$ and $\savidVerify$,
and three protocols,
$\savidSetup$, $\savidDisperse$ and $\savidRetrieve$.
\begin{itemize}
    \item
        $\savidSetup\colon 1^\lambda \mapsto (\cryptoPP, \cryptoSP_1,...,\cryptoSP_n)$: 
        The protocol $\savidSetup$ is run by a temporary trusted party and all storage nodes,
        at the beginning of time (\ie, before adversarial corruption).
        It takes as input the security parameter $1^\lambda$
        and outputs global public parameters $\cryptoPP$,
        and local secret parameters $\cryptoSP_1,...,\cryptoSP_n$,
        one per
        storage node.
        
        The public parameters $\cryptoPP$ are common knowledge and input to all other
        algorithms and protocols.
        The secret parameters $\cryptoSP_1,...,\cryptoSP_n$ are part of the state
        of a storage node and as such available to that node during
        $\savidDisperse$
        and $\savidRetrieve$ invocations.
        Explicit mention of these inputs is subsequently omitted for simplicity of notation.
        
    \item 
        $\savidCommit\colon B \mapsto C$:
        The algorithm $\savidCommit$
        takes as input 
        a block $B$ of data,
        and returns a commitment $C$ to the data.
        
    \item
        $\savidDisperse\colon B \mapsto P$: 
        The protocol $\savidDisperse$ is run by a client and all storage nodes.
        It takes as input 
        a block $B$ of data at the client,
        and outputs $\bot$ or a \emph{certificate of retrievability} $P$ for commitment $C = \savidCommit(B)$
        to the client.
        
    \item
        $\savidVerify\colon (P, C) \mapsto b \in \{\TRUE, \FALSE\}$:
        The algorithm $\savidVerify$
        takes as input
        a certificate of retrievability $P$
        and a commitment $C$,
        and returns $\TRUE$ or $\FALSE$, depending on whether the certificate is considered valid.
        
    \item
        $\savidRetrieve\colon (P, C) \mapsto B$:
        The protocol $\savidRetrieve$ is run by a client and all storage nodes.
        It takes as input
        a certificate of retrievability $P$
        and a commitment $C$ at the client,
        and outputs $\bot$ or a block $B$ of data to the client.
\end{itemize}
\end{definition}
\begin{algorithm}[t]
    \caption{Commitment-binding game ($\gameCB$) against Semi-AVID-PR scheme $\savid = (\savidSetup, \savidCommit, \savidDisperse, \savidVerify, \savidRetrieve)$}
    \label{alg:game-savid-binding}
    \begin{algorithmic}[1]
        \State $(\cryptoPP, \cryptoSP_1, ..., \cryptoSP_n) \gets \savidSetup(1^\lambda)$
            \Comment{Run setup for all parties}
        \State $(B, B') \gets \CA_{\gameCB}(\cryptoPP, \cryptoSP_1, ..., \cryptoSP_n)$
            \Comment{$\CA$ can simulate any party}
        \State \Return $B \neq B' \land \savidCommit(B) = \savidCommit(B')$
    \end{algorithmic}
\end{algorithm}
\begin{algorithm}[t]
    \caption{Availability game ($\gameAv$) with resilience $t$ against Semi-AVID-PR scheme $\savid = (\savidSetup, \savidCommit, \savidDisperse, \savidVerify, \savidRetrieve)$}
    \label{alg:game-savid-availability}
    \begin{algorithmic}[1]
        \State $\CC \gets \emptyset$
            \Comment{Bookkeeping of corrupted parties}
        \State $\forall i\in[n]: P_i \gets \mathsf{new}\,\savid(\emptyset)$
            \Comment{Instantiate $P_i$ as $\savid$ with blank state}
        \State $\cryptoPP \gets \savidSetup^{P_1,...,P_n}(1^\lambda)$
            \Comment{Run setup among all parties}
        \Function{$\CO^{\mathrm{corrupt}}$}{$i$}
                \Comment{Oracle for $\CA$ to corrupt parties}
            \State \Assert{$i \not\in \CC$}
            \State $\CC \gets \CC \cup \{i\}$
                \Comment{Mark party as corrupted}
            \State \Return $P_i$
                \Comment{Hand $P_i$'s state to $\CA$}
        \EndFunction
        \Function{$\CO^{\mathrm{interact}}$}{$i, m$}
                \Comment{Oracle for $\CA$ to interact with parties}
            \State \Assert{$i \not\in \CC$}
            \State \Return $P_i(m)$
                \Comment{Execute $P_i$ on input $m$, return output to $\CA$}
        \EndFunction
        \State $\left(P, C, \left(\CO_i^{\mathrm{node}}(.)\right)_{i\in\CC}\right)
            \gets \CA^{\CO^{\mathrm{corrupt}}(.), \CO^{\mathrm{interact}}(.)}_{\gameAv}(\cryptoPP)$
            \Comment{$\CA$ returns certificate of retrievability $P$, commitment $C$, and oracle access to corrupted nodes for retrieval}
        \State $\hat B \gets \savidRetrieve^{P_1,...,P_n}\left[\CO_i^{\mathrm{node}}(.)/\Call{Query}{i, .}\right]_{i\in\CC}(P, C)$
            \Comment{During retrieval, interact with corrupted nodes through oracles}
        \State \Return $\begin{aligned}[t]
                    & |\CC| \leq t   \\[-2pt]
                    & \land \savidVerify(P, C) = \TRUE   \\[-5pt]
                    & \land \savidCommit(\hat B) \neq C
                \end{aligned}$
            \Comment{$\CA$ wins iff: while corrupting no more than $t$ parties, $\CA$ produces a valid certificate of retrievability $P$ for $C$ such that retrieval
            does not return a file matching $C$}
    \end{algorithmic}
\end{algorithm}
\begin{definition}[Semi-AVID-PR Security]
\label{def:semiavidpr-security}
A Semi-AVID-PR scheme $\savid$ is \emph{secure with resilience $t$}
if for all executions with $f \leq t$:
\begin{enumerate}
    \item
        \textbf{Commitment-Binding.}
        $\savidCommit$ of $\savid$ implements a binding deterministic commitment to a block $B$ of data. 
        More formally, $\savid$ is \emph{commitment-binding} if for any PPT adversary $\mathcal A$ there exists a negligible function $\negl(.)$ such that
        \begin{IEEEeqnarray}{C}
            \Prob{\gameCB_{\savid, \CA}(\lambda) = \TRUE} \leq \negl(\lambda),
            \IEEEeqnarraynumspace
        \end{IEEEeqnarray}
        where $\gameCB_{\savid, \CA}(\lambda)$ is the commitment-binding game defined in Alg.~\ref{alg:game-savid-binding}.
        
    \item
        \textbf{Correctness.}
        If an honest client invokes $\savidDisperse$ with a block $B$ of data,
        then eventually it outputs a certificate of retrievability $P$
        with the property
        that
        $\savidVerify(P, \savidCommit(B)) = \TRUE$.
        
    \item
        \textbf{Availability.}
        For a certificate of retrievability $P$
        and a commitment $C$,
        if $\savidVerify(P, C) = \TRUE$,
        then if an honest client invokes $\savidRetrieve$ with $P$ and $C$,
        then eventually it outputs a block $B$ of data
        such that $\savidCommit(B) = C$.
        More formally, $\savid$ provides \emph{availability} if for any PPT adversary $\mathcal A$ there exists a negligible function $\negl(.)$ such that
        \begin{IEEEeqnarray}{C}
            \Prob{\gameAv_{\savid, \CA}(\lambda, t) = \TRUE} \leq \negl(\lambda),
            \IEEEeqnarraynumspace
        \end{IEEEeqnarray}
        where $\gameAv_{\savid, \CA}(\lambda, t)$ is the availability game defined in Alg.~\ref{alg:game-savid-availability}.
        ($\Pi[f(.)/g(.)]$ denotes replacing $g(.)$ with $f(.)$ in $\Pi$.)
        
\end{enumerate}
\end{definition}

A few remarks are due on this formulation.
First, note the interplay of Commitment-Binding
and Availability.
Because of Commitment-Binding,
Availability ensures that any block,
for which there is a valid certificate of retrievability,
can 
be retrieved by honest clients.
This holds even if
the certificate 
was produced adversarially,
\eg, by a malicious rollup operator
tampering with the dispersal.

Unlike earlier formulations of AVID \cite{avid,avidfp,avidm},
our formulation does not have independent session identifiers.
Instead, the scheme provides a binding commitment scheme which is used
to establish a link between the data in question,
invocations of the dispersal/retrieval protocols,
and 
certificates of retrievability.
The completion of dispersal of a block
and the possibility to retrieve content matching a commitment
are tied together through the Commitment-Binding property of the commitment scheme
and
can be proven to a third party
using the certificate of retrievability.
This matches the Validium rollup application, where on the one hand retrievability of content matching a certain commitment needs to be verifiable on-chain, and on the other hand validity of the block content is proved and verified with respect to the commitment.

This can also be seen as following the paradigm shift from location-addressed
to content-addressed storage and is particularly
suitable for applications such as rollups or sharding
where one wants to succinctly but unequivocally
identify \emph{what} content is being referenced
rather than \emph{where to find it}.
In terms of the original four properties of AVID schemes \cite{avid},
our Correctness property takes the place of the Termination and Agreement properties,
and our Availability property takes the place of the Availability and Correctness properties.
Above weakenings (hence the name `Semi'-AVID)
allow us to achieve greater resilience up to $t < n/2$
rather than $t < n/3$ as for AVID, AVID-FP or AVID-M.

\section{Protocol}
\label{sec:protocol}
        \begin{algorithm}[t]%
            \caption{$\OURsavid.\savidSetup(1^\lambda)$}
            \label{alg:savid-setup}
            \begin{algorithmic}[1]
                \State \textbf{At the trusted party:}
                \State $\cryptoPP_{\slvc} \gets \slvcSetup(1^\lambda)$
                \State $\cryptoPP_{\crhf} \gets \crhfGen(1^\lambda)$
                \State \textbf{At each storage node $i$:} \quad $(\cryptoPK_i, \cryptoSK_i) \gets \sigKeygen(1^\lambda)$
                \State \Return $\cryptoPP = (\cryptoPP_{\slvc}, \cryptoPP_{\crhf}, \cryptoPK_1, ..., \cryptoPK_n), \cryptoSP_1 = \cryptoSK_1, ..., \cryptoSP_n = \cryptoSK_n$
            \end{algorithmic}
        \end{algorithm}
        \begin{algorithm}[t]%
            \caption{$\OURsavid.\savidCommit(B)$}
            \label{alg:savid-commit}
            \begin{algorithmic}[1]
                \State $\BU \leftarrow \operatorname{AsMatrix}_{L \times k}(B)$
                \State $(h_1, ..., h_k) \gets \VC^{\otimes k}(\BU)$
                \State \Return $\CRHF(h_1\|...\|h_k)$
            \end{algorithmic}
        \end{algorithm}
        \begin{algorithm}[t]%
            \caption{$\OURsavid.\savidVerify(P, C)$}
            \label{alg:savid-verify}
            \begin{algorithmic}[1]
                \State $\hat q \gets \left\vert \left\{
                i \,\middle\vert\,
                \exists (i \mapsto \sigma) \in P\colon \sigVerify(\cryptoPK_i, (\savidLabelStored, C), \sigma) = \TRUE
                \right\} \right\vert$
                \If{$\hat q \geq q$}
                    \Return $\TRUE$
                \EndIf
                \State\Return $\FALSE$
            \end{algorithmic}
        \end{algorithm}
        \begin{algorithm}[t]%
            \caption{$\OURsavid.\savidDisperse(B)$}
            \label{alg:savid-disperse}
            \begin{algorithmic}[1]
                \State \textbf{At the client:}
                \State $\BU \leftarrow \operatorname{AsMatrix}_{L \times k}(B)$
                \State $(h_1, ..., h_k) \gets \VC^{\otimes k}(\BU)$
                \State $\BC \gets \codeEncode^{\otimes L}(\BU)$
                \State Send $(\savidLabelDisperse, (h_1, ..., h_k), \Bc_i)$ to all storage nodes $i$
                \State \textbf{At storage node $i$ upon receiving $(\savidLabelDisperse, (h_1, ..., h_k), \Bc_i)$:}
                \State $\hat h \gets [\codeEncode(h_1,...,h_k)]_i$
                \If{$\hat h \neq \VC(\Bc_i)$}
                    \textbf{abort}   \label{loc:savid-disperse-consistency-check}
                \EndIf
                \State $C \gets \CRHF(h_1\|...\|h_k)$
                \State Store $C \mapsto ((h_1, ..., h_k), \Bc_i)$
                       \label{loc:savid-disperse-store}
                \State Send $\sigma_i \triangleq \sigSign(\cryptoSK_i, (\savidLabelStored, C))$ to client
                       \label{loc:savid-disperse-receipt}
                \State \textbf{At the client:}
                \State Wait for $\sigma_{i_j}$ from $q$ unique $\{i_j\}_{j=1}^q$ with $\sigVerify(\cryptoPK_{i_j}, (\savidLabelStored, C), \sigma_{i_j}) = \TRUE$ 
                \State \Return $\bigcup_{j=1}^q \{ i_j \mapsto \sigma_{i_j} \}$
            \end{algorithmic}
        \end{algorithm}
        \begin{algorithm}[t]%
            \caption{$\OURsavid.\savidRetrieve(P, C)$}
            \label{alg:savid-retrieve}
            \begin{algorithmic}[1]
                \State \textbf{At the client:}
                \State Extract from $P$ any $q$ unique %
                        $\left\{
                        i \,\middle\vert\,
                        \exists (i \mapsto \sigma) \in P\colon \sigVerify(\cryptoPK_i, (\savidLabelStored, C), \sigma) = \TRUE
                        \right\}$
                       \label{loc:savid-retrieve-extract}
                \State Send $(\savidLabelRetrieve, C)$ to all storage nodes $i$
                \State \textbf{At storage node $i$ upon receiving $(\savidLabelRetrieve, C)$:}
                \State Load $C \mapsto ((h_1, ..., h_k), \Bc_i)$
                    \label{loc:savid-retrieve-load}
                \State Send $(i, (h_1, ..., h_k), \Bc_i)$ to client
                \State \textbf{At the client:}
                \State Wait for $(h_1, ..., h_k)$ such that $C = \CRHF(h_1\|...\|h_k)$
                    \label{loc:savid-retrieve-wait}
                \State $\hat \Bh \gets \codeEncode(h_1,...,h_k)$
                \label{loc:savid-retrieve-coded-commitments}
                \State Discarding any $i$ with $[\hat \Bh]_{i} \neq \VC(\Bc_i)$, 
                    wait for $k$ remaining unique $\{i_j\}_{j=1}^k$
                    \label{loc:savid-retrieve-wait2}
                \State \Return $\codeDecode(( (i_j, \Bc_{i_j}) )_{j=1}^k)$
                    \label{loc:savid-retrieve-decode}
            \end{algorithmic}
        \end{algorithm}

\begin{figure}[t]
    \centering
    \begin{tikzpicture}[x=0.9cm]
        \footnotesize
        
        \node[align=left,anchor=north west,yshift=8pt] at (0,0) {\textsc{Client}};

        \node [minimum width=2.3mm,minimum height=1.5cm,draw=black] (u1) at (0.3,-1.5) {};
        \node [minimum width=2.3mm,minimum height=1.5cm,draw=black] (u2) at (0.6,-1.5) {};
        \node [minimum width=2.3mm,minimum height=1.5cm,draw=black] (u3) at (0.9,-1.5) {};
        \node (u4) at (1.2,-1.5) {...};
        \node [minimum width=2.3mm,minimum height=1.5cm,draw=black] (uK) at (1.5,-1.5) {};
        
        \node [rotate=90,font=\scriptsize] at (u1) {$\Bu_1$};
        \node [rotate=90,font=\scriptsize] at (u2) {$\Bu_2$};
        \node [rotate=90,font=\scriptsize] at (u3) {$\Bu_3$};
        \node [rotate=90,font=\scriptsize] at (uK) {$\Bu_k$};
        
        \draw [Latex-Latex] ([yshift=7pt] u1.north west) -- ([yshift=7pt] uK.north east) node [midway,fill=white,font=\scriptsize] {$k$};
        \draw [Latex-Latex] ([xshift=-7pt] u1.north west) -- ([xshift=-7pt] u1.south west) node [midway,fill=white,font=\scriptsize] {$L$};

        \node [minimum width=2.3mm,minimum height=1.5cm,draw=black] (c1) at (3.3,-1.5) {};
        \node [minimum width=2.3mm,minimum height=1.5cm,draw=black] (c2) at (3.6,-1.5) {};
        \node [minimum width=2.3mm,minimum height=1.5cm,draw=black] (c3) at (3.9,-1.5) {};
        \node (c4) at (4.2,-1.5) {...};
        \node [minimum width=2.3mm,minimum height=1.5cm,draw=black] (c5) at (4.5,-1.5) {};
        \node (c6) at (4.8,-1.5) {...};
        \node [minimum width=2.3mm,minimum height=1.5cm,draw=black] (cN) at (5.1,-1.5) {};
        
        \node [rotate=90,font=\scriptsize] at (c1) {$\Bc_1$};
        \node [rotate=90,font=\scriptsize] at (c2) {$\Bc_2$};
        \node [rotate=90,font=\scriptsize] at (c3) {$\Bc_3$};
        \node [rotate=90,font=\scriptsize] at (c5) {$\Bc_i$};
        \node [rotate=90,font=\scriptsize] at (cN) {$\Bc_n$};
        
        \draw [Latex-Latex] ([yshift=7pt] c1.north west) -- ([yshift=7pt] cN.north east) node [midway,fill=white,font=\scriptsize] {$n$};
        \draw [Latex-Latex] ([xshift=7pt] cN.north east) -- ([xshift=7pt] cN.south east) node [midway,fill=white,font=\scriptsize] {$L$};

        \draw [-Latex] (1.8,-0.9) -- (3,-0.9);
        \draw [-Latex] (1.8,-1.1) -- (3,-1.1);
        \draw [-Latex] (1.8,-1.3) -- (3,-1.3);
        \draw [-Latex] (1.8,-1.5) -- (3,-1.5);
        \draw [-Latex] (1.8,-1.7) -- (3,-1.7);
        \draw [-Latex] (1.8,-1.9) -- (3,-1.9);
        \draw [-Latex] (1.8,-2.1) -- (3,-2.1);

        \node [rotate=90,fill=white] at (2.3,-1.5) {$\codeEncode^{\otimes L}$};

        \node [] (h1) at (0.3,-3.4) {$h_1$};
        \node [] (h2) at (0.6,-3.4) {$h_2$};
        \node [] (h3) at (0.9,-3.4) {$h_3$};
        \node [] at (1.2,-3.4) {...};
        \node [] (hK) at (1.5,-3.4) {$h_k$};
        
        \draw [-Latex] (u1) -- (h1);
        \draw [-Latex] (u2) -- (h2);
        \draw [-Latex] (u3) -- (h3);
        \draw [-Latex] (u5) -- (hK);
        
        \draw [draw=none] (u3) -- (h3) node [pos=0.4,fill=white] {$\VC^{\otimes k}$};

        \node[align=right,anchor=north east,yshift=8pt,xshift=-10pt] at (9.25,0) {\textsc{Storage Node $i$}};
        
        \node [align=center,anchor=east] (verify) at (9.25,-3.4) {\\\\$[\codeEncode(h_1, ..., h_k)]_i$ \\$\overset{?}{=} \VC(\Bc_i)$\\If not, abort. Else,\\$C \leftarrow \CRHF(h_1\|...\|h_k)$,\\$S_i[C] \leftarrow ((h_1, ..., h_k), \Bc_i)$.};
        
        \draw [-Latex] (commitments) -- (verify) node [pos=0.6,above] {$(\savidLabelDisperse, (h_1, ..., h_k), \Bc_i)$};
        \draw [-Latex] (c5.south) -- (verify);

        \node (certificate) at (1.1,-4.5) {$P = \bigcup_{j=1}^q \{ i_j \mapsto \sigma_{i_j} \}$};
        
        \draw [-Latex] (verify.west |- certificate) -- (certificate) node [midway,above] {$\sigma_i \triangleq \sigSign(\cryptoSK_i, (\savidLabelStored, C))$};
        
    \end{tikzpicture}
    \caption[]{$\savidDisperse$ protocol of our Semi-AVID-PR construction $\OURsavid$ (\cf Figure~\ref{fig:savidpr-disperse-overview}).
    Client arranges data in $L \times k$ matrix $\BU$,
    computes commitments $h_1, ..., h_k$ column-wise
    and $L \times n$ coded matrix $\BC$ row-wise.
    Commitments and $i$-th column $\Bc_i$ of $\BC$ are sent to storage node $i$.
    Upon verification, storage node computes commitment $C$ to the data,
    stores commitments and chunk, and acknowledges receipt of chunk to client.
    Client forms certificate of retrievability $P$
    from $q$ unique server identifiers $i_j$ and their receipts $\sigma_{i_j}$.}
    \label{fig:savidpr-disperse}
\end{figure}
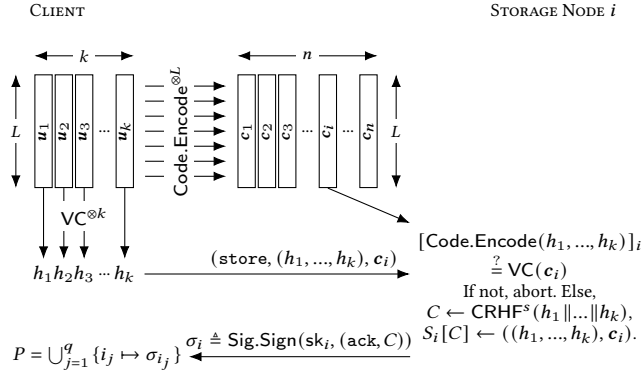

We provide a construction $\OURsavid$ of Semi-AVID-PR from
a binding deterministic linear vector commitment scheme $\slvc$,
a maximum distance separable $(n,k)$-code $\code$,
a collision resistant hash function $\CRHF$,
and a secure digital signature scheme $\sig$.
Our construction
satisfies the properties laid out in Section~\ref{sec:model}
as shown in Section~\ref{sec:security-argument}.
Moreover, it is storage- and communication-efficient
and incurs practically moderate cost for cryptographic computations
and erasure-correction coding
as demonstrated in Section~\ref{sec:evaluation}.
It is easy to extend our scheme with blinding such that
an honest-but-curious storage node 
cannot learn anything about the dispersed data from its chunk
(see Section~\ref{sec:privacy}).

Our construction $\OURsavid$ is provided in
Algs.~\ref{alg:savid-setup} to~\ref{alg:savid-retrieve}.
See also Figure~\ref{fig:savidpr-disperse}
for an illustration of the $\savidDisperse$ protocol
(\cf Figure~\ref{fig:savidpr-disperse-overview}).
Our approach is related to AVID-FP \cite{avidfp}
in that we also use the linear homomorphism between
the LVC and the erasure-correcting code.
More specifically,
during $\savidDisperse$,
the input file $B$ is arranged as an $L \times k$ matrix $\BU$
(using $\operatorname{AsMatrix}_{L \times k}$)
and a commitment $h_i$ is taken per column $\Bu_i$.
Vectorization of the $k$ column commitments of $\BU$ is denoted as $\VC^{\otimes k}(\BU)$.
The matrix is encoded row-wise into a coded matrix $\BC$,
of which each column $\Bc_i$ constitutes the chunk
for storage server $i$.
Vectorization of the $L$ row encodings of $\BU$ is denoted as $\codeEncode^{\otimes L}(\BU)$.
Now, due to the linear homomorphism,
the commitment of the encodings $\Bc_i$ is equal to the
encoding of the commitments $h_i$.
This allows storage nodes to easily verify the consistency of their chunk
with the uncoded data
(\ie, the verifiability property in AVID).
For this check, a storage node only needs to know
the commitments $h_i$ of the uncoded data,
which keeps the communication-overhead of the scheme low.
Our approach differs from AVID-FP in that
AVID-FP still performs a round of AVID (for the commitments)
in order to satisfy the full AVID requirements (in particular Termination
and Agreement),
while our Semi-AVID-PR scheme satisfies only the weaker
Correctness property
(\cf Figure~\ref{fig:development-avid-schemes}).

Our construction is parametric in the design resilience $t$,
the quorum size $q$ for certificates of retrievability,
the code dimension $k$
and the length of chunks $L$.
The analysis of Section~\ref{sec:security-argument} reveals
that $q \leq (n-t)$, $0 < (q-t)$ and $k \leq (q-t)$
are necessary.
So given any $t < n/2$ and target file size $|B|$ (in field elements),
choose
$q \triangleq (n-t)$,
minimize storage overhead
with $k \triangleq n-2t$,
and set $L \triangleq |B|/k$.

During $\savidSetup$ (Algorithm~\ref{alg:savid-setup}),
a trusted party performs the setup of
the $\slvc$ and the $\crhf$, and each storage node generates a cryptographic identity
for $\sig$.
The public parameters of the LVC and the public keys of the storage nodes
become common knowledge, each storage node stores its
secret key.

The $\savidCommit$ment of a block $B$ (Algorithm~\ref{alg:savid-commit})
is computed by arranging $B$ as an $L \times k$ matrix $\BU$,
then computing the commitments $h_i$ as $\VC(\Bu_i)$
for each of the $k$ columns $\Bu_i$ of $\BU$,
and finally $\CRHF(h_1\|...\|h_k)$ is the commitment.

To $\savidDisperse$ a block $B$ (Algorithm~\ref{alg:savid-disperse},
Figures~\ref{fig:savidpr-disperse-overview}, \ref{fig:savidpr-disperse}),
the client first computes
$\BU$ and the commitments $h_i$ as for $\savidCommit$.
Then, $\BU$ is encoded row-wise using $\codeEncode$ to obtain
an $L \times n$ coded matrix $\BC$.
Each column $\Bc_i$ of $\BC$ is sent to storage node $i$ together with
$(h_1,...,h_k)$. Each storage node $i$ verifies its chunk using
the linearly homomorphic property (aborting if violated)
\begin{IEEEeqnarray}{C}
    \label{eq:protocol-homomorphic-property}
    [\codeEncode(h_1, ..., h_k)]_i \overset{?}{=} \VC(\Bc_i),
\end{IEEEeqnarray}
before computing the file's commitment
$C \triangleq \CRHF(h_1\|...\|h_k)$ and storing commitments and chunk indexed by $C$.
The storage node then acknowledges receipt of the chunk by sending
a signature on $(\savidLabelStored, C)$ to the client.
Upon collecting valid signatures $\sigma_{i_j}$ from $q$ unique storage nodes $i_j$,
the client collects them into a certificate of retrievability $P$.

To $\savidVerify$ a certificate of retrievability $P$
for a commitment $C$ (Algorithm~\ref{alg:savid-verify}),
one counts whether $P$ contains valid signatures on $(\savidLabelStored, C)$
from at least $q$ unique storage nodes.

Finally, to $\savidRetrieve$ a file based on a
certificate of retrievability $P$ for a commitment $C$,
the client first extracts any $q$ unique storage nodes
for which $P$ contains a valid signature on $(\savidLabelStored, C)$.
The client then queries the chunks of $C$ from these storage nodes.
The storage nodes reply with the commitments and chunks they have
stored for $C$.
The client first waits until some commitments
$(h_1,...,h_k)$ satisfy $C = \CRHF(h_1\|...\|h_k)$.
Then, the client discards any chunks that do not satisfy
the homomorphic property \eqref{eq:protocol-homomorphic-property}.
Upon receiving valid chunks from $k$ unique storage nodes,
the client uses $\codeDecode$ to decode the file.

To protect the data against honest-but-curious storage nodes
(\ie, assuming storage nodes do not collude---clearly,
a sufficiently large set of
storage nodes can retrieve the data, which is a design
goal of Semi-AVID-PR),
$\BU$ can be augmented by the client with a column and a row
drawn uniformly a random, to blind the encoded chunks.
Details in Section~\ref{sec:privacy}.

\section{Security Proof}
\label{sec:security-argument}
We first provide a proof sketch to convey the relevant high level intuition,
before providing a formal reduction-based proof of
Commitment-Binding and Availability in Lemmas~\ref{lem:commitment-binding}
and~\ref{lem:availability}.

\begin{theorem}
The Semi-AVID-PR construction $\OURsavid$ of Section~\ref{sec:protocol}
is \emph{secure with resilience $t$}
for any $t < n/2$
as defined in Section~\ref{sec:model},
assuming $\crhf$ is collision resistant,
$\sig$ is existentially unforgeable,
and $\slvc$ is binding.
\end{theorem}

\begin{proof}[Proof sketch]
        \textbf{Commitment-Binding.}
        $\OURsavid$ is deterministic because $\VC$ and $\CRHF$ are.
        For binding, assume for contradiction that $\OURsavid$ was not commitment-binding,
        \ie, there was an adversary $\CA$ that can produce blocks $B \neq B'$
        such that $\savidCommit(B) = \savidCommit(B')$
        with non-negligible probability. 
        Since $B \neq B'$, for their respective representations as
        $L \times k$ matrices,
        $\BU \neq \BU'$.
        Either $\Bh \triangleq \VC^{\otimes k}(\BU) \neq \VC^{\otimes k}(\BU') \triangleq \Bh'$ but $\CRHF(\Bh) = \CRHF(\Bh')$, a collision in $\CRHF$, which can happen with negligible probability only by assumption, 
        or $\Bh = \Bh'$,
        so that for some $i$, $\VC([\BU]_{i}) = \VC([\BU']_{i})$
        but $[\BU]_{i} \neq [\BU']_{i}$, so $([\BU]_{i}, [\BU']_{i})$ is a pair
        that breaks the binding property of $\VC$, which can happen with negligible probability only by assumption. Thus, $\OURsavid$ is commitment-binding.

        Lemma~\ref{lem:commitment-binding} below establishes Commitment-Binding formally.

        \textbf{Correctness.}
        Since the client is honest and $\slvc$ is linearly homomorphic,
        the consistency check in Algorithm~\ref{alg:savid-disperse}
        l.~\ref{loc:savid-disperse-consistency-check} passes at all honest
        storage nodes.
        So the client receives signatures $\sigma_{i_j}$ (which are valid,
        by correctness of $\sig$) of $(\savidLabelStored, \savidCommit(B))$
        from at least $(n-t)$ unique
        storage nodes $i_j$,
        which it can bundle into a certificate of retrievability
        $P$ that satisfies the check of Algorithm~\ref{alg:savid-verify}
        by construction, \emph{if $q \leq (n-t)$}.
        \textbf{Availability.}
        Since $\savidVerify(P, C) = \TRUE$ by assumption,
        the client can extract some $q$ unique storage nodes $i_j$
        from $P$ in Algorithm~\ref{alg:savid-retrieve} l.~\ref{loc:savid-retrieve-extract}.
        Of these $q$ storage nodes, at least $(q-t)$
        remain honest. Security of $\sig$ implies that they must have
        previously executed Algorithm~\ref{alg:savid-disperse}
        l.~\ref{loc:savid-disperse-receipt} and hence
        stored $(h_1,...,h_k)$ for
        $C = \CRHF(h_1,...,h_k)$ in
        Algorithm~\ref{alg:savid-disperse}
        l.~\ref{loc:savid-disperse-store}
        and their chunks satisfied the consistency check
        in Algorithm~\ref{alg:savid-disperse}
        l.~\ref{loc:savid-disperse-consistency-check}.
        Note that by collision resistance of $\CRHF$,
        there can be only one set of $(h_1,...,h_k)$
        for $C$.
        \emph{As long as $(q-t) > 0$,}
        the client eventually completes the wait in
        Algorithm~\ref{alg:savid-retrieve} l.~\ref{loc:savid-retrieve-wait},
        and \emph{if $k \leq (q-t)$},
        then the client eventually also completes
        the wait in
        Algorithm~\ref{alg:savid-retrieve} l.~\ref{loc:savid-retrieve-wait2}.
        Finally, since $\code$ is an
        MDS $(n,k)$-code,
        Algorithm~\ref{alg:savid-retrieve} l.~\ref{loc:savid-retrieve-decode}
        succeeds to decode a block $B$, corresponding to
        an $L \times k$ matrix $\BU$.
        It remains to show that $\savidCommit(B) = C$,
        for which (by $\CRHF$)
        it suffices that $(h_1,...,h_k) = \VC^{\otimes k}(\BU)$.
        Note the decoder uses that $\code$ is an MDS $(n,k)$-code
        and thus any $k \times k$ submatrix
        of its generator matrix $\BG$ is invertible,
        and the relation
        \begin{IEEEeqnarray}{C}
            \underbrace{
            \setlength{\arraycolsep}{2pt}\begin{bmatrix}
                \Bu_1 & ... & \Bu_k
            \end{bmatrix}
            }_{
            = \BU
            }
            \underbrace{
            \setlength{\arraycolsep}{2pt}\begin{bmatrix}
                \Bg_{i_1} & ... & \Bg_{i_k}
            \end{bmatrix}
            }_{
            \triangleq \tilde\BG
            }
            \overset{!}{=}
            \underbrace{
            \setlength{\arraycolsep}{2pt}\begin{bmatrix}
                \Bc_{i_1} & ... & \Bc_{i_k}
            \end{bmatrix}
            }_{
            \triangleq \tilde\BC
            }
            \iff
            \BU = \tilde\BC \tilde\BG^{-1}.
            \IEEEeqnarraynumspace
        \end{IEEEeqnarray}
        At the same time, by the checks in Algorithm~\ref{alg:savid-retrieve} l.~\ref{loc:savid-retrieve-wait2},
        \begin{IEEEeqnarray}{C}
            \setlength{\arraycolsep}{2pt}\begin{bmatrix}
                h_1 & ... & h_k
            \end{bmatrix}
            \tilde\BG
            =
            \setlength{\arraycolsep}{2pt}\begin{bmatrix}
                \VC(\Bc_{i_1}) & ... & \VC(\Bc_{i_k})
            \end{bmatrix}.
            \IEEEeqnarraynumspace
        \end{IEEEeqnarray}
        By the linear homomorphism of $\slvc$,
        \begin{IEEEeqnarray}{C}
            \VC^{\otimes k}(\BU) = \VC^{\otimes k}(\tilde\BC) \tilde\BG^{-1}
            =
            \setlength{\arraycolsep}{2pt}\begin{bmatrix}
                h_1 & ... & h_k
            \end{bmatrix}.
            \IEEEeqnarraynumspace
        \end{IEEEeqnarray}
        Thus, as long as the retrieving client receives
        enough chunks from honest storage nodes
        such that $\savidRetrieve$ outputs any block $\hat B$
        rather than $\bot$,
        $\savidCommit(\hat B) = C$ for the commitment $C$
        associated with the certificate of retrievability $P$.
        This insight is crucial for the subsequent formal proof
        of Availability.

        Lemma~\ref{lem:availability} below establishes Availability formally.

        \textbf{Resilience.}
        From above analysis, we obtain the constraints
        $q \leq (n-t)$ (for Correctness),
        $0 < (q-t)$ (for Availability, to obtain $h_1, ..., h_k$),
        and
        $k \leq (q-t)$ (for Availability, to decode),
        which the choice of parameters in Section~\ref{sec:protocol}
        satisfies, and which lead to the resilience bound $t < n/2$.
\end{proof}

\begin{algorithm}[t]
    \caption{$\CA_{\gameCF\leftarrow\gameCB}(s)$
    constructed from $\CA_{\gameCB}$}
    \label{alg:reduction-binding-crhf}
    \begin{algorithmic}[1]
        \State $\cryptoPP_{\slvc} \gets \slvcSetup(1^\lambda)$
            \Comment{Setup $\OURsavid$ (\cf Alg.~\ref{alg:savid-setup}) ...}
        \State $\forall i \in [n]: (\cryptoPK_i, \cryptoSK_i) \gets \sigKeygen(1^\lambda)$ 
        \State $\cryptoPP_{\crhf} \gets s$
            \Comment{... except use $s$ for $\cryptoPP_{\crhf}$}
        \State $\cryptoPP \gets (\cryptoPP_{\slvc}, \cryptoPP_{\crhf}, \cryptoPK_1, ..., \cryptoPK_n)$
        \State $(B, B') \gets \CA_{\gameCB}(\cryptoPP, \cryptoSP_1, ..., \cryptoSP_n)$
        \State $\BU \leftarrow \operatorname{AsMatrix}_{L \times k}(B)$
        \State $\BU' \leftarrow \operatorname{AsMatrix}_{L \times k}(B')$
        \State $(h_1, ..., h_k) \gets \VC^{\otimes k}(\BU)$
        \State $(h'_1, ..., h'_k) \gets \VC^{\otimes k}(\BU')$
        \If $\begin{aligned}[t]
                    & h_1\|...\|h_k \neq h'_1\|...\|h'_k    \\
                    & \land \CRHF(h_1\|...\|h_k) = \CRHF(h'_1\|...\|h'_k)  
                \end{aligned}$
            \State \Return $(h_1\|...\|h_k, h'_1\|...\|h'_k)$ \Comment{Collision in $\crhfH$}
        \Else
            \State \Abort
                \Comment{No collision identified}
        \EndIf
    \end{algorithmic}
\end{algorithm}

\begin{algorithm}[t]
    \caption{$\CA_{\gameVCB\leftarrow\gameCB}(\cryptoPP)$
    constructed from $\CA_{\gameCB}$}
    \label{alg:reduction-binding-vcbinding}
    \begin{algorithmic}[1]
        \State $\cryptoPP_{\crhf} \gets \crhfGen(1^\lambda)$
            \Comment{Setup $\OURsavid$ (\cf Alg.~\ref{alg:savid-setup}) ...}
        \State $\forall i \in [n]: (\cryptoPK_i, \cryptoSK_i) \gets \sigKeygen(1^\lambda)$ 
        \State $\cryptoPP_{\slvc} \gets \cryptoPP$
            \Comment{... except use $\cryptoPP$ for $\cryptoPP_{\slvc}$}
         \State $\cryptoPP \gets (\cryptoPP_{\slvc}, \cryptoPP_{\crhf}, \cryptoPK_1, ..., \cryptoPK_n)$
        \State $(B, B') \gets \CA_{\gameCB}(\cryptoPP, \cryptoSP_1, ..., \cryptoSP_n)$ 
        \State $\BU \leftarrow \operatorname{AsMatrix}_{L \times k}(B)$
        \State $\BU' \leftarrow \operatorname{AsMatrix}_{L \times k}(B')$
        \If $\exists i \in [k] : [\BU]_{i} \neq  [\BU']_{i} \land \VC([\BU]_{i}) = \VC([\BU']_{i})$ 
            \State \Return $([\BU]_{i}, [\BU']_{i})$ 
        \Else
            \State \Abort
                
        \EndIf
    \end{algorithmic}
\end{algorithm}

We proceed to give formal reduction-based proofs
of the Com\-mit\-ment-Binding and Availability properties
of $\OURsavid$.

\begin{lemma}
\label{lem:commitment-binding}
If $\crhf$ is collision resistant and $\slvc$ is binding, then 
$\OURsavid$
is commitment-binding, \ie, for all PPT adversaries $\CA_{\gameCB}$ there exists a negligible function $\negl(.)$ such that
\begin{IEEEeqnarray}{C}
\Prob{\gameCB_{\OURsavid, \CA_{\gameCB}}(\lambda) = \TRUE} \leq \negl(\lambda).
\end{IEEEeqnarray}
\end{lemma}

\begin{proof}
    Let $\CA_{\gameCB}$ be an arbitrary PPT $\gameCB$ adversary. We construct from it the adversaries $\CA_{\gameCF\leftarrow\gameCB}$ against $\gameCF$ and $\CA_{\gameVCB\leftarrow\gameCB}$ against $\gameVCB$. The adversary $\CA_{\gameCF\leftarrow\gameCB}$ is detailed in Alg.~\ref{alg:reduction-binding-crhf}. It receives a challenge $s$ and runs $\slvcSetup(1^\lambda)$ and $\sigKeygen(1^\lambda)$ to produce the remaining public parameters $\cryptoPP$ and the secret parameters $(\cryptoSP_1, ..., \cryptoSP_n)$ for $\CA_{\gameCB}$. After $\CA_{\gameCB}$ outputs a pair $(B, B')$, $\CA_{\gameCF\leftarrow\gameCB}$ computes $\BU \leftarrow \operatorname{AsMatrix}_{L \times k}(B)$ and $\BU' \leftarrow \operatorname{AsMatrix}_{L \times k}(B)$. Next, it computes $(h_1, ..., h_k) \gets \VC^{\otimes k}(\BU)$ and $(h'_1, ..., h'_k) \gets \VC^{\otimes k}(\BU')$ to check whether $h_1\|...\|h_k$ and $h'_1\|...\|h'_k$ are a collision of $\CRHF$. If so, it outputs $(h_1\|...\|h_k, h'_1\|...\|h'_k)$; else it aborts. 
    
    The adversary $\CA_{\gameVCB\leftarrow\gameCB}$ is detailed in Alg.~\ref{alg:reduction-binding-vcbinding}. It receives the challenge $\cryptoPP_{\slvc}$ and runs $\crhfGen(1^\lambda)$ and $\sigKeygen(1^\lambda)$ to produce the remaining public parameters $\cryptoPP$ and the secret parameters $(\cryptoSP_1, ..., \cryptoSP_n)$ for $\CA_{\gameCB}$. 
    After $\CA_{\gameCB}$ outputs a pair $(B, B')$, $\CA_{\gameVCB\leftarrow\gameCB}$ computes $\BU \leftarrow \operatorname{AsMatrix}_{L \times k}(B)$ and $\BU' \leftarrow \operatorname{AsMatrix}_{L \times k}(B')$. Next, it checks if there exists $i \in [k]$ such that the column $[\BU]_{i} \neq [\BU']_{i}$ but $\VC([\BU]_{i}) = \VC([\BU']_{i})$. If so, it outputs the collision $([\BU]_{i}, [\BU']_{i})$; else it aborts.
    
    The adversaries $\CA_{\gameCF\leftarrow\gameCB}$ and $\CA_{\gameVCB\leftarrow\gameCB}$ run in polynomial time. The input of the adversary $\CA_{\gameCB}$ when run as a subroutine of $\CA_{\gameCF\leftarrow\gameCB}$ or $\CA_{\gameVCB\leftarrow\gameCB}$ is distributed identically to the input of the adversary $\CA_{\gameCB}$ when run in $\gameCB$.
    
    For the subsequent arguments we define the following events:
    \newcommand{\EVENTcb}[0]{\ensuremath{E_{\mathrm{CB}}}}
    \newcommand{\EVENTvcb}[0]{\ensuremath{E_{\mathrm{VCB}}}}
    \newcommand{\EVENTcf}[0]{\ensuremath{E_{\mathrm{CF}}}}
    \newcommand{\EVENT}[0]{\ensuremath{E}}
    \begin{IEEEeqnarray}{rCl}
    \EVENTcb &\triangleq& \{\gameCB_{\OURsavid, \CA_{\gameCB}}(\lambda) = \TRUE\}\\
    \EVENTcf &\triangleq& \{\gameCF_{\crhf, \CA_{\gameCF\leftarrow\gameCB}}(\lambda) = \TRUE\}\\
    \EVENTvcb &\triangleq& \{\gameVCB_{\slvc, \CA_{\gameVCB\leftarrow\gameCB}}(\lambda) = \TRUE\}\\
    \EVENT &\triangleq& \{ \VC^{\otimes k}(\BU) \neq  \VC^{\otimes k}(\BU') \land \BU \neq \BU'\}   \\
    U &\triangleq& \EVENTcf \lor \EVENTvcb  
    \end{IEEEeqnarray}

    Suppose $\EVENTcb$ holds and $\CA_{\gameCB}$ outputs $(B, B')$ such that $\savidCommit(B) = \savidCommit(B')$ but $B \neq B'$. Hence, $\CRHF(h_1\|...\|h_k) = \CRHF(h'_1\|...\|h'_k)$, where 
    $\BU \leftarrow \operatorname{AsMatrix}_{L \times k}(B)$,
    $\BU' \leftarrow \operatorname{AsMatrix}_{L \times k}(B')$,
    $(h_1, ..., h_k) \gets \VC^{\otimes k}(\BU)$ and $(h'_1, ..., h'_k) \gets \VC^{\otimes k}(\BU')$. We consider two cases. If $\EVENT$ holds,
    then $(h_1, ..., h_k) \neq (h_1, ..., h_k)$.
    Thus, in the event of $\EVENT$,  $(h_1\|...\|h_k, h'_1\|...\|h'_k)$ is a collision of $\CRHF$, so the event $\EVENTcf$ holds. 
    In the case of $\lnot \EVENT$, $(h_1, ..., h_k) = (h'_1, ..., h'_k)$. Thus, there exists $i \in [k]$ such that $[\BU]_{i} \neq [\BU']_{i}$ but $\VC([\BU]_{i}) = \VC([\BU']_{i})$. So under $\lnot \EVENT$, the event $\EVENTvcb$ holds.
    
    Observe that if $\EVENT$ holds, then $\EVENTcf$ holds. Hence, $ \EVENT \subseteq \EVENTcf$ and $\Prob{\lnot \EVENTcf \land \EVENT} = 0$. 
    Similarly, if $\lnot \EVENT$ holds, then $\EVENTvcb$ holds. Hence, $ \lnot \EVENT \subseteq \EVENTvcb$ and $\Prob{\lnot \EVENTvcb \land \lnot \EVENT} = 0$.

    We can now bound the probability of $\EVENTcb$:
    \begin{IEEEeqnarray}{rCl}
        && \Prob{\EVENTcb}   \nonumber\\
        &\eqA& \CProb{\EVENTcb}{U}\Prob{U} + \CProb{\EVENTcb}{\lnot U}\Prob{\lnot U} \\
        &\leqB& \Prob{U} + \Prob{\EVENTcb \land \lnot U} \\
        &\leqC& \Prob{U} + \Prob{\EVENTcb \land \lnot U \land \EVENT}+ \Prob{\EVENTcb \land \lnot U \land \lnot \EVENT} \IEEEeqnarraynumspace\\
        &\leqD& \Prob{U} + \Prob{\lnot \EVENTcf \land \EVENT}+ \Prob{\lnot \EVENTvcb \land \lnot \EVENT}\\
        &\leqE& \Prob{U}\\
        &\leqF& \Prob{\EVENTcf} + \Prob{\EVENTvcb}
    \end{IEEEeqnarray}
    where (a) uses the law of total probability (TP) to introduce $U$;
    (b) uses $\CProb{\EVENTcb}{U} \leq 1$;
    (c) uses TP to introduce $\EVENT$;
    (d) uses 
    \begin{IEEEeqnarray}{C}
        \Prob{\EVENTcb \land \lnot \EVENTcf \land \lnot \EVENTvcb \land \EVENT} \leq \Prob{\lnot \EVENTcf \land \EVENT}
    \end{IEEEeqnarray}
    and  
    \begin{IEEEeqnarray}{C}
        \Prob{\EVENTcb \land \lnot \EVENTcf \land \lnot \EVENTvcb \land \lnot \EVENT} \leq \Prob{\lnot \EVENTvcb \land \lnot \EVENT};
    \end{IEEEeqnarray}
    (e) uses $\Prob{\lnot \EVENTcf \land \EVENT} = 0$ and $\Prob{\lnot \EVENTvcb \land \lnot \EVENT} = 0$;
    (f) uses a union bound.
    
    Since by assumption $\crhf$ is collision resistant and $\slvc$ is binding, there exist $\negl_1(.), \negl_2(.)$ such that $\Prob{\EVENTcf} \leq \negl_1(\lambda)$ and $\Prob{\EVENTvcb} \leq \negl_2(\lambda)$. Thus,
    \begin{IEEEeqnarray}{rCl}
        \Prob{\EVENTcb} \leq  \negl_1(\lambda) + \negl_2(\lambda) \leq \negl(\lambda).
    \end{IEEEeqnarray}
    Hence, $\OURsavid$ is commitment-binding.
\end{proof}

\begin{algorithm}[t]
    \caption{$\CA_{\gameEF\leftarrow\gameAv'}(\cryptoPK)$
    constructed from $\CA_{\gameAv'}$}
    \label{alg:reduction-availability-sig-forgery}
    \begin{algorithmic}[1]
        \State $I \drawrandom [n]$
            \Comment{Choose random party $I$ to emulate using $\CO^{\mathrm{sign}}(.)$}
        \State $\CC \gets \emptyset$
            \Comment{Bookkeeping of corrupted parties $\CC$}
        \State $\forall i\in[n]: S_i \gets \emptyset$
            \Comment{Blank state for each party $P_i$}
        \State $\cryptoPP_{\slvc} \gets \slvcSetup(1^\lambda)$
            \Comment{Setup $\OURsavid$ (\cf Alg.~\ref{alg:savid-setup}) ...}
        \State $\cryptoPP_{\crhf} \gets \crhfGen(1^\lambda)$
        \State $\forall i \in [n] \setminus \{I\}: (\cryptoPK_i, \cryptoSK_i) \gets \sigKeygen(1^\lambda)$
        \State $\cryptoPK_I \gets \cryptoPK$
            \Comment{... except use $\cryptoPK$ for party $I$'s $\cryptoPK_I$}
        \State $\cryptoPP \gets (\cryptoPP_{\slvc}, \cryptoPP_{\crhf}, \cryptoPK_1, ..., \cryptoPK_n)$
        \Function{$\CO^{\mathrm{corrupt}}$}{$i$}
            \State \Assert{$i \not\in \CC$}
            \If{$i \neq I$}
                \State $\CC \gets \CC \cup \{i\}$
                \State\Return $(\cryptoSK_i, S_i)$
            \Else
                \State\Abort
                    \Comment{Cannot hand over $P_I$ since $\cryptoSK_I$ is unknown}
            \EndIf
        \EndFunction
        \Function{$\CO^{\mathrm{interact}}$}{$i, m$}
            \State \Assert{$i \not\in \CC$}
            \If{$i \neq I$}
                \State\Return $\OURsavidSUPER{\cryptoSK_i, S_i}(m)$
                    \Comment{Execute $P_i$ on input $m$ and state $\cryptoSK_i, S_i$, and return output to $\CA$}
            \Else
                \State\Return $\OURsavidSUPER{S_i}[\CO^{\mathrm{sign}}(.)/\sigSign(\cryptoSK_i,.)](m)$
                    \Comment{Execute $P_I$ on input $m$ and state $S_i$, substituting $\CO^{\mathrm{sign}}(.)$ for invocations of $\sigSign(\cryptoSK_i,.)$, and return output to $\CA$}
            \EndIf
        \EndFunction
        \State $\left(P, C, \left(\CO_i^{\mathrm{node}}(.)\right)_{i\in\CC}\right)
            \gets \CA^{\CO^{\mathrm{corrupt}}(.), \CO^{\mathrm{interact}}(.)}_{\gameAv'}(\cryptoPP)$
        \State $\hat B \gets \savidRetrieve^{P_1,...,P_n}\left[\CO_i^{\mathrm{node}}(.)/\Call{Query}{i, .}\right]_{i\in\CC}(P, C)$
            \Comment{During retrieval, interact with corrupted nodes through oracles}
        \If{$\exists \sigma: (I \mapsto \sigma) \in P
            \land S_I[C] = \emptyset$}
            \State \Return $((\savidLabelStored, C), \sigma)$
                \Comment{Forgery for $\cryptoPK_I = \cryptoPK$}
        \Else
            \State \Abort
                \Comment{No forgery for $\cryptoPK_I = \cryptoPK$ identified}
        \EndIf
    \end{algorithmic}
\end{algorithm}

\begin{algorithm}[t]
    \caption{$\CA_{\gameCF\leftarrow\gameAv'}(s)$
    constructed from $\CA_{\gameAv'}$}
    \label{alg:reduction-availability-crhf}
    \begin{algorithmic}[1]
        \State $I \drawrandom [n]$
        \State $\CC, \CH \gets \emptyset, \emptyset$
            \Comment{Bookkeeping of corrupted parties $\CC$ and image/preimage pairs $\CH$ for $\crhfH$}
        \State $\forall i\in[n]: S_i \gets \emptyset$
            \Comment{Blank state for each party $P_i$}
        \State $\cryptoPP_{\slvc} \gets \slvcSetup(1^\lambda)$
            \Comment{Setup $\OURsavid$ (\cf Alg.~\ref{alg:savid-setup}) ...}
        \State $\forall i \in [n]: (\cryptoPK_i, \cryptoSK_i) \gets \sigKeygen(1^\lambda)$
        \State $\cryptoPP_{\crhf} \gets s$
            \Comment{... except use $s$ for $\cryptoPP_{\crhf}$}
        \State $\cryptoPP \gets (\cryptoPP_{\slvc}, \cryptoPP_{\crhf}, \cryptoPK_1, ..., \cryptoPK_n)$
        \Function{$\CO^{\mathrm{corrupt}}$}{$i$}
            \State \Assert{$i \not\in \CC$}
            \If{$i \neq I$}
                \State $\CC \gets \CC \cup \{i\}$
                \State\Return $(\cryptoSK_i, S_i)$
            \Else
                \State\Abort
            \EndIf
        \EndFunction
        \Function{$\CO^{\mathrm{interact}}$}{$i, m$}
            \State \Assert{$i \not\in \CC$}
            \If{$m$ parses as $(\savidLabelDisperse, (h_1, ..., h_k), \Bc)$}
                \State $\CH \gets \CH \cup \{ (\CRHF(h_1\|...\|h_k) \mapsto h_1\|...\|h_k) \}$
                    \Comment{Record image/preimage pair for $\crhfH$}
            \EndIf
            \State\Return $\OURsavidSUPER{\cryptoSK_i, S_i}(m)$
                \Comment{Execute $P_i$ on input $m$ and state $\cryptoSK_i, S_i$, and return output to $\CA$}
        \EndFunction
        \State $\left(P, C, \left(\CO_i^{\mathrm{node}}(.)\right)_{i\in\CC}\right)
            \gets \CA^{\CO^{\mathrm{corrupt}}(.), \CO^{\mathrm{interact}}(.)}_{\gameAv'}(\cryptoPP)$
        \Function{$\tilde\CO_i^{\mathrm{node}}$}{$m$}
            \State $r \gets \CO_i^{\mathrm{node}}(m)$
            \If{$r$ parses as $(i, (h_1, ..., h_k), \Bc_i)$}
                \State $\CH \gets \CH \cup \{ (\CRHF(h_1\|...\|h_k) \mapsto h_1\|...\|h_k) \}$
                    \Comment{Record image/preimage pair for $\crhfH$}
            \EndIf
            \State\Return $r$
        \EndFunction
        \State $\savidRetrieve^{P_1,...,P_n}\left[\tilde\CO_i^{\mathrm{node}}(.)/\Call{Query}{i, .}\right]_{i\in\CC}(P, C)$
            \Comment{During retrieval, interact with corrupted nodes through wrapped oracles}
        \If{$\exists x, x': (C \mapsto x) \in \CH \land (C \mapsto x') \in \CH \land x \neq x'$}
            \State \Return $(x, x')$
                \Comment{Collision in $\crhfH$}
        \Else
            \State \Abort
                \Comment{No collision identified}
        \EndIf
    \end{algorithmic}
\end{algorithm}

\begin{lemma}
\label{lem:availability}
If $\crhf$ is collision resistant,
$\sig$ is secure against existential forgery,
and $t < n/2$,
then $\OURsavid$ provides availability, \ie, for all PPT adversaries $\CA_{\gameAv}$ there exists a negligible function $\negl(.)$ such that
\begin{IEEEeqnarray}{C}
\Prob{\gameAv_{\OURsavid, \CA_{\gameAv}}(\lambda, t) = \TRUE} \leq \negl(\lambda).
\end{IEEEeqnarray}
\end{lemma}
\begin{proof}

First, we modify the availability game $\gameAv$
(Alg.~\ref{alg:game-savid-availability})
to obtain $\gameAv'$
(Alg.~\ref{alg:game-savid-availability-modified})
in which initially the index $I$ of a storage node is sampled
uniformly at random,
and subsequently the game is aborted if the adversary
$\CA_{\gameAv'}$ attempts to corrupt $I$.
This modification will subsequently 
streamline the reduction of availability of $\OURsavid$
to security of $\sig$ against existential forgery.

We reduce availability of $\OURsavid$
to availability$^\prime$ of $\OURsavid$, \ie,
if
for all PPT $\CA_{\gameAv'}$
there exists $\negl(.)$ such that
\begin{IEEEeqnarray}{C}
\Prob{\gameAv'_{\OURsavid, \CA_{\gameAv'}}(\lambda, t) = \TRUE} \leq \negl(\lambda),
\end{IEEEeqnarray}
then
for all PPT $\CA_{\gameAv}$
there exists $\negl(.)$ such that
\begin{IEEEeqnarray}{C}
\Prob{\gameAv_{\OURsavid, \CA_{\gameAv}}(\lambda, t) = \TRUE} \leq \negl(\lambda).
\end{IEEEeqnarray}
To this end, pick any $\gameAv$ adversary $\CA_{\gameAv}$.
Note that
$\CA_{\gameAv'} \triangleq \CA_{\gameAv}$ is an
$\gameAv'$ adversary.
Define the events:
\newcommand{\EVENTua}[0]{\ensuremath{E_{\mathrm{A}}}}
\newcommand{\EVENTuap}[0]{\ensuremath{E_{\mathrm{A'}}}}
\begin{IEEEeqnarray}{rCl}
    \EVENTua &\triangleq& \{ \gameAv_{\OURsavid,\CA_{\gameAv}}(\lambda, t) = \TRUE \}   \\
    \EVENTuap &\triangleq& \{ \gameAv'_{\OURsavid,\CA_{\gameAv'}}(\lambda, t) = \TRUE \}
\end{IEEEeqnarray}
Obviously, $\EVENTua \land \{ I \not\in \CC \} \subseteq \EVENTuap$. Furthermore, $\EVENTua$ implies $|\CC| \leq t < n/2$, so that $\CProb{I \not\in \CC}{\EVENTua} \geq 1/2$.
Thus, by availability$^\prime$ of $\OURsavid$,
\begin{IEEEeqnarray}{rCl}
    \Prob{\EVENTuap}
    &\geq& \Prob{\EVENTua \land \{ I \not\in \CC \}}   \\
    &\geq& \CProb{\{ I \not\in \CC \}}{\EVENTua} \Prob{\EVENTua}   \geq \frac{1}{2} \Prob{\EVENTua}   \\
    \Prob{\EVENTua} &\leq& 2 \Prob{\EVENTuap} \leq \negl(\lambda).
\end{IEEEeqnarray}

We proceed with the reduction of
availability$^\prime$ of $\OURsavid$
to collision resistance of $\crhf$
and security of $\sig$ against existential forgery.
Let $\CA_{\gameAv'}$ be an arbitrary PPT $\gameAv'$ adversary.
We construct from it the adversaries
$\CA_{\gameEF\leftarrow\gameAv'}$ for the $\gameEF$ and
$\CA_{\gameCF\leftarrow\gameAv'}$ for the $\gameCF$
as detailed in Algs.~\ref{alg:reduction-availability-sig-forgery} and~\ref{alg:reduction-availability-crhf}, respectively.
$\CA_{\gameEF\leftarrow\gameAv'}$ emulates the $\gameAv'$ challenger
(Alg.~\ref{alg:game-savid-availability-modified}), except it
does not generate a signature public/secret key pair for node $I$, but instead uses
the input challenge $\cryptoPK$, and
it
attempts to forge a signature for $I$.
It uses the signature oracle
$\CO^{\mathrm{sign}}(.)$ provided in the $\gameEF$
to produce signatures for node $I$ whenever 
$\OURsavid$ requires to do so.
$\CA_{\gameCF\leftarrow\gameAv'}$ emulates the $\gameAv'$ challenger
(Alg.~\ref{alg:game-savid-availability-modified}), except it uses the input challenge $s$
for $\crhf$'s key in the public parameters of $\OURsavid$.
Throughout the protocol execution, for both
dispersal and retrieval operations,
the reduction adversary keeps track of any $(h_1, ..., h_k)$ that may
present colliding inputs
for $\CRHF$.
Clearly, the adversaries $\CA_{\gameEF\leftarrow\gameAv'}$ and $\CA_{\gameCF\leftarrow\gameAv'}$
run in time polynomial in the security parameter $\lambda$.
Furthermore, the input $\cryptoPP$ of $\CA_{\gameAv'}$
and its interactions through the oracles
$\CO^{\mathrm{corrupt}}(.)$ and $\CO^{\mathrm{interact}}(.)$
are distributed identically
when run by the challenger
of $\gameAv'$
and
when run as a subroutine of $\CA_{\gameEF\leftarrow\gameAv'}$ or $\CA_{\gameCF\leftarrow\gameAv'}$ invoked by the challenger
of $\gameEF$ or $\gameCF$, respectively.

For the subsequent arguments we define the following events:
\newcommand{\EVENTef}[0]{\ensuremath{E_{\mathrm{E}}}}
\newcommand{\EVENTcf}[0]{\ensuremath{E_{\mathrm{C}}}}
\newcommand{\EVENT}[0]{\ensuremath{E}}
\begin{IEEEeqnarray}{rCl}
    \EVENTef &\triangleq& \{ \gameEF_{\sig,\CA_{\gameEF\leftarrow\gameAv'}}(\lambda) = \TRUE \}   \\
    \EVENTcf &\triangleq& \{ \gameCF_{\crhf,\CA_{\gameCF\leftarrow\gameAv'}}(\lambda) = \TRUE \}   \\
    \EVENT &\triangleq& \{ (\exists \sigma: (I\mapsto\sigma) \in P) \land S_I[C] = \emptyset \}
\end{IEEEeqnarray}

The following facts
will
be useful.
Observe that if 
$\EVENT \land \EVENTuap$ holds, then $\CA_{\gameEF\leftarrow\gameAv'}$
wins the $\gameEF$, \ie, $\EVENTef$ holds. So $\EVENT \land \EVENTuap \subseteq \EVENTef$.
Thus, $\Prob{\EVENTuap \land \lnot\EVENTcf \land \lnot\EVENTef \land \EVENT} = 0$.
Inverting the implication, $\lnot\EVENTef \subseteq \lnot\EVENT \lor \lnot\EVENTuap$.
Thus, $\Prob{\EVENTuap \land \lnot\EVENTcf \land \lnot\EVENTef \land \lnot\EVENT} \leq \Prob{\EVENTuap \land \lnot\EVENTcf \land \lnot\EVENT}$,
where we have used a union bound and $\EVENTuap \land \lnot\EVENTuap = \emptyset$.

Finally, suppose $\CA_{\gameAv'}$ as a subroutine of $\CA_{\gameEF\leftarrow\gameAv'}$
behaves such that it would win the corresponding $\gameAv'$.
Furthermore, suppose no collision is identified by $\CA_{\gameCF\leftarrow\gameAv'}$,
so $\lnot\EVENTcf$.
Then, given that $|\CC| \leq t < q$, $\savidVerify(P, C) = \TRUE$, the only
way in which $\savidCommit(\hat B) \neq \CC$ could be true is if
there is at least one node that is part of $P$, has not been corrupted,
and has not previously stored a chunk associated with $C$.
This is because, as was argued earlier in the proof sketch,
if the retrieving client receives chunks from at least $k$
honest storage nodes, it decodes a block that matches
the expected commitment $C$. Hence, it must be the case that
the client does not receive sufficiently many valid chunks. But since
$P$ contains $q > t$ valid signatures, but at most $t \geq |\CC|$
storage nodes are corrupted, and $k \leq (q-t)$ by design,
there must be an honest storage node
whose signature on $(\savidLabelStored, C)$ is in $P$,
yet the node does not respond to the retrieving client's query
because it has never stored a chunk associated with $C$
(and hence not signed $(\savidLabelStored, C)$).
Thus, $\CProb{\EVENT}{\EVENTuap \land \lnot\EVENTcf} \geq \frac{1}{n}$.

\newcommand{\ERR}[0]{\ensuremath{x}}
Now we can bound, with $\ERR \triangleq \Prob{\EVENTcf} + \Prob{\EVENTef}$,
\begin{IEEEeqnarray}{rCl}
    \Prob{\EVENTuap}   
    &\eqA& \Prob{\EVENTuap \land \EVENTcf} + \Prob{\EVENTuap \land \lnot\EVENTcf}   \\
    &\leqB& \Prob{\EVENTcf} + \Prob{\EVENTef} + \Prob{\EVENTuap \land \lnot\EVENTcf \land \lnot\EVENTef}   \\
    &\leqC& \ERR + \Prob{\EVENTuap \land \lnot\EVENTcf \land \lnot\EVENTef \land \lnot\EVENT}   \IEEEeqnarraynumspace\\
    &\leqD& \ERR + \Prob{\EVENTuap \land \lnot\EVENTcf \land \lnot\EVENT}   \IEEEeqnarraynumspace\\
    &=& \ERR + \CProb{\lnot\EVENT}{\EVENTuap \land \lnot\EVENTcf} \Prob{\EVENTuap \land \lnot\EVENTcf}   \IEEEeqnarraynumspace\\
    &\leqE& \ERR + \frac{n-1}{n} \Prob{\EVENTuap}   \IEEEeqnarraynumspace\\
    \Prob{\EVENTuap}
    &\leq& n \ERR = n \Prob{\EVENTcf} + n \Prob{\EVENTef}
\end{IEEEeqnarray}
where
(a) uses the law of total probability (TP) to introduce $\EVENTcf$;
(b) uses TP to introduce $\EVENTef$, $\EVENTuap \land \EVENTcf \subseteq \EVENTcf$, $\EVENTuap \land \lnot\EVENTcf \land \EVENTef \subseteq \EVENTef$;
(c) uses TP to introduce $\EVENT$,
$\Prob{\EVENTuap \land \lnot\EVENTcf \land \lnot\EVENTef \land \EVENT} = 0$;
(d) uses
$\Prob{\EVENTuap \land \lnot\EVENTcf \land \lnot\EVENTef \land \lnot\EVENT} \leq \Prob{\EVENTuap \land \lnot\EVENTcf \land \lnot\EVENT}$;
(e) uses
$\EVENTuap \land \lnot\EVENTcf \subseteq \EVENTuap$,
$\CProb{\EVENT}{\EVENTuap \land \lnot\EVENTcf} \geq \frac{1}{n}$.

Since by assumption $\crhf$ is collision resistant
and $\sig$ is secure against existential forgery,
there exist $\negl_1(.), \negl_2(.)$ such that
$\Prob{\EVENTcf} \leq \negl_1(\lambda)$
and
$\Prob{\EVENTef} \leq \negl_2(\lambda)$. Furthermore, $n$ is a constant independent of the security parameter $\lambda$.
Thus,
\begin{IEEEeqnarray}{rCl}
    \Prob{\EVENTuap} \leq n \, \negl_1(\lambda) + n \, \negl_2(\lambda) \leq \negl(\lambda).
\end{IEEEeqnarray}
Hence, $\OURsavid$ provides availability$^\prime$,
to which availability of $\OURsavid$ was reduced
in the first part of the proof.
\end{proof}

\section{Privacy}
\label{sec:privacy}
\makeatletter
\renewcommand*\env@matrix[1][*\c@MaxMatrixCols c]{%
  \hskip -\arraycolsep
  \let\@ifnextchar\new@ifnextchar
  \array{#1}}
\makeatother

\DeclareRobustCommand{\svdots}{%
  \vcenter{%
    \offinterlineskip
    \hbox{.}
    \vskip0.25\normalbaselineskip
    \hbox{.}
    \vskip0.25\normalbaselineskip
    \hbox{.}%
  }%
}
\DeclareRobustCommand{\sddots}{%
  \vcenter{%
    \offinterlineskip
    \hbox{.\hspace{0.5\normalbaselineskip}}
    \vskip0.25\normalbaselineskip
    \hbox{\hspace{0.25\normalbaselineskip}.\hspace{0.25\normalbaselineskip}}
    \vskip0.25\normalbaselineskip
    \hbox{\hspace{0.5\normalbaselineskip}.}%
  }%
}
\DeclareRobustCommand{\scdots}{%
  \vcenter{%
    \offinterlineskip
    \hbox{.\,.\,.}%
  }%
}

Our Semi-AVID-PR scheme $\OURsavid$ can be extended to hide the dispersed data
from non-colluding honest-but-curious storage nodes,
\ie, formally,
for each storage node, the distribution of 
the
dispersed information
is independent of
the data received by the storage node.
For this purpose, with a slight abuse of notation,
let $\tilde \BU$ denote the $(L-1) \times (k-1)$ matrix of
information to be dispersed
(with columns $\tilde \Bu_i$).
The dispersing client
augments it first with a blinding column
$\Bb \drawrandom \IZ_p^{L-1}$
to the right of $\tilde \BU$
and then with a blinding row
$\Bs \drawrandom \IZ_p^k$
to the bottom of both $\tilde \BU$ and $\Bb$,
to obtain the $L \times k$ matrix 
$\BU$,
\begin{IEEEeqnarray}{C}
    \BU \triangleq \begin{bmatrix}[cc]
        \tilde \BU & \Bb \\
        \multicolumn{2}{c}{-\,\Bs^\top-}
    \end{bmatrix}.
\end{IEEEeqnarray}
The coded matrix $\BC$ (with columns $\Bc_i$) and
the column commitments $(h_1,...,h_k)$
continue to be computed as detailed in
Figure~\ref{fig:savidpr-disperse}
and
Section~\ref{sec:protocol}.
Thus,
storage node $m$ receives $(\Bc_m, h_1, ..., h_k)$ as part of the protocol
(see Alg.~\ref{alg:savid-disperse}).
\begin{theorem}
\label{thm:privacy}
The distribution of $(\Bc_m, h_1, ..., h_k)$
induced by the randomness in the blinding $\Bb$ and 
$\Bs$
is independent of $\tilde \BU$, so that storage node $m$
learns nothing about the dispersed information.
\end{theorem}
\begin{proof}
Assume that storage node $m$ could even compute
$\log_g(.)$ in $\IG$, and hence knows the secret $r$ sampled during trusted
setup of the KZG polynomial commitment scheme
(see Section~\ref{sec:preliminaries-slvc}),
as well as $\log_g(h_i)$ for the column commitments $h_i$.
Furthermore,
assume a Reed-Solomon code (see Section~\ref{sec:preliminaries-rs-codes})
is used as part of $\OURsavid$
so that the column $\Bg_{\mathrm{RS},m} = (\alpha_m^0, ..., \alpha_m^{k-1})$
corresponds to storage node $m$
in the code's generator matrix $\BG_{\mathrm{RS}}$.
Then, the data obtained by node $m$ is related to the unknowns
by the 
following
equations:
\begin{IEEEeqnarray}{C}
    \tiny
    \setlength{\arraycolsep}{1pt}\setlength{\extrarowheight}{1.5pt}\begin{bmatrix*}[l]
        [\Bc_m]_1 \\
        \qquad\svdots \\
        [\Bc_m]_{L-1} \\[1pt]
        \hline %
        [\Bc_m]_L \\[1pt]
        \hline %
        \log_g(h_1) \\
        \qquad\svdots \\
        \log_g(h_{k-1}) \\[2pt]
        \hline %
        \log_g(h_k)
    \end{bmatrix*}
    =
    \underbrace{
    \setlength{\arraycolsep}{1pt}\setlength{\extrarowheight}{2.5pt}\begin{bmatrix}[ccccccc|ccc|ccc|c]
        \alpha_m^0 & ... & \alpha_m^{k-2} &&&&& \alpha_m^{k-1} & 0 & ... &&&& \\
        &&& \sddots &&&&& \sddots &&&&& \\
        &&&& \alpha_m^0 & ... & \alpha_m^{k-2} & ... & 0 & \alpha_m^{k-1} &&&& \\[1pt]
        \hline %
        &&&&&&&&&& \alpha_m^0 & ... & \alpha_m^{k-2} & \alpha_m^{k-1} \\[1pt]
        \hline %
        r^0 & 0 & ... && r^{L-2} & 0 & ... & && & r^{L-1} & 0 & ... & \\
        & \sddots & & \sddots & & \sddots & & && & & \sddots & & \\
        ... & 0 & r^0 & & ... & 0 & r^{L-2} & && & ... & 0 & r^{L-1} & \\
        \hline %
        && & & && & r^0 & ... & r^{L-2} & && & r^{L-1}
    \end{bmatrix}
    }_{\triangleq \BM \in \IZ_p^{(L+k) \times (L k + (L-1) + (k-1) + 1)}}
    \setlength{\arraycolsep}{1pt}\setlength{\extrarowheight}{2pt}\begin{bmatrix*}[l]
        [\tilde \Bu_1]_1 \\
        \qquad\svdots \\
        [\tilde \Bu_{k-1}]_1 \\
        \qquad\svdots \\
        [\tilde \Bu_1]_{L-1} \\
        \qquad\svdots \\
        [\tilde \Bu_{k-1}]_{L-1} \\[1pt]
        \hline %
        \qquad\Bb \\%[1pt]
        \hline %
        \qquad\Bs
    \end{bmatrix*}
    \quad
    \IEEEeqnarraynumspace
\end{IEEEeqnarray}

Denote by $[\BM]_i^\top$ the $i$-th row of $\BM$.
Observe that
\begin{IEEEeqnarray}{C}
    [\BM]_{L+k}^\top
    =
    \sum_{i=1}^{L} r^{i-1} \alpha_m^{-(k-1)} [\BM]_{i}^\top
    -
    \sum_{i=1}^{k-1} \alpha_m^{-(k-i)} [\BM]_{L+i}^\top.
\end{IEEEeqnarray}
Thus, the last 
equation of the system
is redundant.
Striking 
it,
\begin{IEEEeqnarray}{rCl}
    \tiny
    \setlength{\arraycolsep}{1pt}\setlength{\extrarowheight}{1.5pt}\begin{bmatrix*}[l]
        [\Bc_m]_1 \\
        \qquad\svdots \\
        [\Bc_m]_{L-1} \\[1pt]
        \hline %
        [\Bc_m]_L \\[2pt]
        \hline %
        \log_g(h_1) \\
        \qquad\svdots \\
        \log_g(h_{k-1})
    \end{bmatrix*}
    &=&
    \tiny
    \setlength{\arraycolsep}{1pt}\setlength{\extrarowheight}{2.5pt}\begin{bmatrix}[ccccccc]
        \alpha_m^0 & \scdots & \alpha_m^{k-2} &&&& \\
        &&& \sddots &&& \\
        &&&& \alpha_m^0 & \scdots & \alpha_m^{k-2} \\[1pt]
        \hline %
        &&&&&& \\[1pt]
        \hline %
        r^0 & 0 & \scdots && r^{L-2} & 0 & \scdots \\
        & \sddots & & \sddots & & \sddots & \\
        \scdots & 0 & r^0 & & \scdots & 0 & r^{L-2}
    \end{bmatrix}
    \setlength{\arraycolsep}{1pt}\setlength{\extrarowheight}{1.5pt}\begin{bmatrix*}[l]
        [\tilde \Bu_1]_1 \\
        \qquad\svdots \\
        [\tilde \Bu_{k-1}]_1 \\
        \qquad\svdots \\
        [\tilde \Bu_1]_{L-1} \\
        \qquad\svdots \\
        [\tilde \Bu_{k-1}]_{L-1}
    \end{bmatrix*}
    \nonumber\\
    &+&
    \tiny
    \underbrace{
    \setlength{\arraycolsep}{1pt}\setlength{\extrarowheight}{2.5pt}\begin{bmatrix}[ccc|ccc|c]
        \alpha_m^{k-1} & 0 & \scdots &&&& \\
        & \sddots &&&&& \\
        \scdots & 0 & \alpha_m^{k-1} &&&& \\[1pt]
        \hline %
        &&& \alpha_m^0 & \scdots & \alpha_m^{k-2} & \alpha_m^{k-1} \\[1pt]
        \hline %
        && & r^{L-1} & 0 & \scdots & \\
        && & & \sddots & & \\
        && & \scdots & 0 & r^{L-1} &
    \end{bmatrix}
    }_{\triangleq \BM' \in \IZ_p^{(L + k - 1) \times (L - 1 + k)}}
    \setlength{\arraycolsep}{1pt}\setlength{\extrarowheight}{1.5pt}\begin{bmatrix*}[l]
        [\Bb]_1 \\
        \qquad\svdots \\
        [\Bb]_{L-1} \\[1pt]
        \hline %
        [\Bs]_1 \\
        \qquad\svdots \\
        [\Bs]_{k-1} \\[1pt]
        \hline %
        [\Bs]_k
    \end{bmatrix*}.
    \IEEEeqnarraynumspace
\end{IEEEeqnarray}
Observe that $\BM'$ is full-rank.
Thus, the randomness of $\Bb$ and 
$\Bs$
renders
the distribution of
$(\Bc_m, h_1, ..., h_{k-1})$ 
uniform,
while $h_k$ is a function of $(\Bc_m, h_1, ..., h_{k-1})$,
all independent of the dispersed information $\tilde \BU$.
Thus, as desired,
\begin{IEEEeqnarray}{rCl}
    \CProb{\tilde\BU = \By}{(\Bc_m, h_1, ..., h_k) = \Bx}
    &=& \Prob{\tilde\BU = \By}.   \IEEEeqnarraynumspace
\end{IEEEeqnarray}
\end{proof}

\section{Evaluation}
\label{sec:evaluation}
In this section, we
show that the 
computational cost
required for 
our Semi-AVID-PR
scheme $\OURsavid$ is low,
and the communication and storage requirements in comparison
with AVID \cite{avid}, AVID-FP \cite{avidfp}, AVID-M \cite{avidm}
and ACeD \cite{aced} are among the best-of-class
(tied with AVID-M) and practically low,
while providing superior resilience (up to $t<n/2$ vs. $t<n/3$)
and provable retrievability.

\subsection{Computation}
\label{sec:evaluation-computation}

\input{fig_experiments_steps}

\begin{figure}[t]
    \centering
    \begin{tikzpicture}
        \small
        \begin{axis}[
            semiavidprexperiment1,
            ymode=log, xmode=log,
            height=0.45\linewidth,
            width=0.8\linewidth,
            ymin=0.1,
            xtick={32,64,128,256,512,1024},
            xticklabels={$32$,$64$,$128$,$256$,$512$,$1024$},
            ytick={0.001,0.01,0.1,1,10,100,1000},
            yticklabels={$0.001$,$0.01$,$0.1$,$1$,$10$,$100$,$1000$},
            xlabel={System size $n$},
            ylabel={Runtime [s]},
        ]
        
            \addplot [myparula11,experiment-rate33-n128] table [x=args_n,y=runtime_all_row_encodings_seconds] {figures/experiments3/data-experiments3-rate33-bls12-381.txt};
            \label{leg:experiments3-rate33-n128-disperse-encode}
            \addlegendentry{Disperse-Encode};
        
            \addplot [myparula21,experiment-rate33-n128] table [x=args_n,y=runtime_all_column_commitments_seconds] {figures/experiments3/data-experiments3-rate33-bls12-381.txt};
            \label{leg:experiments3-rate33-n128-disperse-commit}
            \addlegendentry{Disperse-Commit};
        
            \addplot [myparula31,experiment-rate33-n128] table [x=args_n,y=runtime_all_downloaded_chunk_verifications_seconds] {figures/experiments3/data-experiments3-rate33-bls12-381.txt};
            \label{leg:experiments3-rate33-n128-retrieve-verify}
            \addlegendentry{Retrieve-Verify};
        
            \addplot [myparula41,experiment-rate33-n128] table [x=args_n,y=runtime_all_row_decodings_seconds] {figures/experiments3/data-experiments3-rate33-bls12-381.txt};
            \label{leg:experiments3-rate33-n128-retrieve-decode}
            \addlegendentry{Retrieve-Decode};
        
            \addplot [myparula51,experiment-rate33-n128] table [x=args_n,y=runtime_prepare_decoding_seconds] {figures/experiments3/data-experiments3-rate33-bls12-381.txt};
            \label{leg:experiments3-rate33-n128-retrieve-prep-decode}
            \addlegendentry{Retrieve-Prep-Decode};
            
            \legend{};

            \addlegendimage{experiment-rate25-n128,black,mark=none};
            \label{leg:experiments3-rate25}
            \addlegendimage{experiment-rate33-n128,black,mark=none};
            \label{leg:experiments3-rate33}
            \addlegendimage{experiment-rate45-n128,black,mark=none};
            \label{leg:experiments3-rate45}
            
            \addlegendimage{experiment-rate25-n128,black,only marks};
            \label{leg:experiments3-n128}
            \addlegendimage{experiment-rate25-n256,black,only marks};
            \label{leg:experiments3-n256}
            \addlegendimage{experiment-rate25-n1024,black,only marks};
            \label{leg:experiments3-n1024}
            
            \addlegendimage{myparula11,only marks,mark=square*};
            \label{leg:experiments3-disperse-encode}
            \addlegendimage{myparula21,only marks,mark=square*};
            \label{leg:experiments3-disperse-commit}
            \addlegendimage{myparula31,only marks,mark=square*};
            \label{leg:experiments3-retrieve-verify}
            \addlegendimage{myparula41,only marks,mark=square*};
            \label{leg:experiments3-retrieve-decode}
            \addlegendimage{myparula51,only marks,mark=square*};
            \label{leg:experiments3-retrieve-prep-decode}
            
            \legend{};
            
        \end{axis}
    \end{tikzpicture}%
    \vspace{-0.5em}%
    \caption{%
    Single-thread runtime
    of steps of $\OURsavid$
    on AMD Opteron 6378 processor
    for varying system size $n$
    in BLS12-381 curve.
    Fixed code rate $k/n\approx0.33$, and
    file size $\approx22\,\mathrm{MB}$.
    $\savidDisperse$:
    RS encoding
    (\ref{leg:experiments3-disperse-encode}),
    computing vector commitments
    (\ref{leg:experiments3-disperse-commit}).
    $\savidRetrieve$:
    Verifying downloaded chunks
    (\ref{leg:experiments3-retrieve-verify}),
    RS decoding $k \times k$ matrix inversion
    (\ref{leg:experiments3-retrieve-prep-decode}),
    RS decoding matrix-matrix product
    (\ref{leg:experiments3-retrieve-decode}).
    (\cf Figure~\ref{fig:experiments})%
    }
    \label{fig:experiments3}
\end{figure}
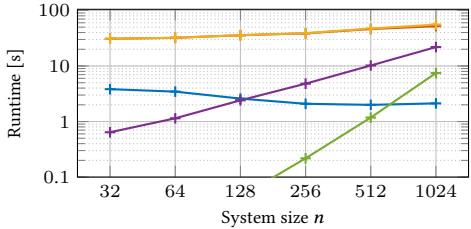

To evaluate the computational requirements
imposed by computation and verification
of vector commitments and encoding and decoding of the erasure-correcting code
during $\savidDisperse$ (\ie, computational burden to the Validium
rollup operator) and $\savidRetrieve$ (\ie, computational burden to the Validium
rollup user, in case of malicious operator)
of our Semi-AVID-PR scheme $\OURsavid$,
we implemented a prototype in the Rust programming language
using libraries from arkworks \cite{arkworks}.
We make the source code available on Github.%
\footnote{%
Source code:
\url{https://github.com/tse-group/semiavidpr-experiments}%
}
We used KZG commitments \cite{kate}
on the BLS12-381 curve \cite{blscurve,bls12381}
as vector commitments,
and Reed-Solomon codes \cite{rs}
over the underlying scalar prime field.
Reed-Solomon (RS) encoding was
implemented using the fast Fourier transform (FFT).
The consistency check of chunks
performed by storage nodes in Alg.~\ref{alg:savid-disperse}
line~\ref{loc:savid-disperse-consistency-check}
was implemented naively,
by computing a vector commitment of the received chunk
and comparing it with the corresponding `encoded' combination
(computed using $k$ exponentiations and multiplications in the group)
of the column commitments $h_1, ..., h_k$.
Note that during dispersal,
these consistency checks are naturally
parallelized across storage nodes,
but have to be computed independently.

For retrieval,
on the other hand, 
the computation of the coded chunks' commitments
(Alg.~\ref{alg:savid-disperse} line~\ref{loc:savid-retrieve-coded-commitments})
was sped up using the FFT.
RS decoding
was naive (generic for any MDS code):
After downloading $k$ valid chunks from distinct storage nodes,
the retrieving client first inverts the corresponding $k \times k$
sub-matrix of the code's generator matrix $\BG$,
and then obtains the matrix $\BU$ of uncoded chunks from
the downloaded sub-matrix of the coded chunks $\BC$
by way of a matrix-matrix product with the inverse.
This way, the cubic complexity (in $k$) of naive matrix inversion via
Gaussian elimination is amortized
over the decoding of $L$ rows (each with quadratic complexity in $k$;
in the regime of interest $k \leq L$).
Note that this naive approach to 
RS decoding
is permissible as
our experiments indicate
that $\savidRetrieve$'s runtime is 
bottlenecked
by another step (verifying downloaded chunks),
and 
during normal operation (when the rollup operator
is honest) only $\savidDisperse$ is invoked.
Note that a systematic erasure-correcting 
code can be used to 
further speed up
decoding
in the realistic scenario where
few storage nodes are corrupted
(as long as the privacy techniques of Section~\ref{sec:privacy} are not used).

Figure~\ref{fig:experiments} shows
the single-thread runtime (ordinate)
on an AMD Opteron 6378 processor
of the different steps of $\OURsavid$
for varying file sizes (abscissa),
system sizes
(columns),
and code rates
(rows).
Figure~\ref{fig:experiments-aggregate} shows
the measurements aggregated on the level of
$\savidDisperse$ and $\savidRetrieve$, respectively.
The plots reveal a minor slowdown
of $\savidDisperse$
with increasing system size and code rate.
The runtime of $\savidDisperse$ is dominated
by computing the vector commitments.
The runtime of $\savidRetrieve$ is dominated
by verifying downloaded chunks.
Recall that both bottlenecking steps have been optimized using the FFT.
Naive RS decoding does not introduce a performance bottleneck, but becomes relevant for large systems ($n=1024$).
Fixing file size to $\approx22\,\mathrm{MB}$
and code rate to $k/n\approx0.33$
while varying system size $n$
(\cf Figure~\ref{fig:experiments3}),
corroborates the earlier observations.

Concretely, the client computation for
dispersing a file of $22\,\mathrm{MB}$ among $256$ storage nodes, up to $85$ of which may be adversarial, requires 
$\approx41\,\mathrm{s}$ of single-thread runtime 
on an AMD Opteron 6378 processor
when using the BLS12-381 curve.
The corresponding retrieval takes
$\approx44\,\mathrm{s}$ of single-thread runtime.
The dispersal throughput of $\approx0.54\,\mathrm{MB/s}$
corresponds to
$\approx2{,}700\,\mathrm{tx/s}$
(assuming $200\,\mathrm{B}$ transaction size).
It should be noted that we report single-thread
runtime on a seven year old processor here.
The workload is embarrassingly parallel and hence
wall-clock time reduces trivially
with an increasing number of parallel workers.
With $16$ threads, the Validium operator can complete dispersal
and Validium users can complete retrieval
in less than $3\,\mathrm{s}$, respectively.

\subsection{Communication \& Storage}
\label{sec:evaluation-communication-storage}

\begin{table}[t]
    \centering
    \caption{Communication and storage required to disperse $22\,\mathrm{MB}$ among $n=1024$ nodes using different solutions;
    resilience and whether provable retrievability is supported.
    (Calculations: Appendix~\ref{sec:appendix-calculations}.)}%
    \vspace{-0.5em}%
    \footnotesize%
    \begin{tabular}{p{0.205\linewidth}rrrc}
        \toprule
        Scheme  & Resilience  & Communication & Storage & Retrievability\\
        \midrule
        Repetition  & $502=0.49n$ & $22.5\,\mathrm{GB}$ & $22.5\,\mathrm{GB}$ & \textcolor{myParula05Green}{\cmark{}} \\
        AVID \cite{avid}         & $338=0.33n$  &   $101\,\mathrm{GB}$ &   $98.3\,\mathrm{MB}$     & \textcolor{myParula05Green}{\cmark{}}      \\
        AVID-FP \cite{avidfp}       & $338=0.33n$  &    $46.1\,\mathrm{GB}$  &    $110\,\mathrm{MB}$   & \textcolor{myParula05Green}{\cmark{}}    \\
        AVID-M \cite{avidm}       & $338=0.33n$    &   $98.7\,\mathrm{MB}$   & $65.1\,\mathrm{MB}$ & \textcolor{myParula07Red}{\xmark{}}  \\
        ACeD \cite{aced}      & $338=0.33n$    &     $585\,\mathrm{MB}$   &  $585\,\mathrm{MB}$     & \textcolor{myParula07Red}{\xmark{}}   \\
        ACeD \cite{aced}     & $502=0.49n$    &    $10.2\,\mathrm{GB}$    &   $10.2\,\mathrm{GB}$    & \textcolor{myParula07Red}{\xmark{}}      \\ 
        This work
        ($\OURsavid$)
        & $338=0.33n$  &     $81.8\,\mathrm{MB}$ &  $81.8\,\mathrm{MB}$    & \textcolor{myParula05Green}{\cmark{}} \\ 
        This work
        ($\OURsavid$)
        & $502=0.49n$    &   $1.13\,\mathrm{GB}$ &    $1.13\,\mathrm{GB}$ & \textcolor{myParula05Green}{\cmark{}} \\
        \bottomrule
    \end{tabular}
    \label{tab:comparison}
\end{table}

Communication and storage required for different data availability solutions
are tabulated for a numerical example in Table~\ref{tab:comparison}.
The calculations are provided in Appendix~\ref{sec:appendix-calculations},
and in the
source code available on Github.%
\footnote{%
Source code:
\url{https://github.com/tse-group/semiavidpr-experiments}%
}
Note the $t<n/2$ resilience upper bound for any scheme
that simultaneously provides \emph{availability}
and \emph{correctness}
(\cf Definition~\ref{def:semiavidpr-security}).
To see this, suppose
the cooperation of $q$ storage nodes is necessary
and sufficient to complete a dispersal.
Correctness requires $q \leq n-t$ (else adversarial storage
nodes can `block' dispersal), and availability requires
$q > t$ (else adversarial storage nodes can `forge' a dispersal).
Combining the two conditions yields $t < n/2$, which is achieved by
the naive repetition (full replication)
scheme, at high communication and storage cost.
Our $\OURsavid$ recovers
the same trade-off when parameterized for resilience close to $n/2$;
but our scheme can also be parameterized for lower resilience, in which case it achieves considerably lower communication and storage,
whereas the repetition scheme does not
allow for such parameterization.
AVID improves over repetition in that each node only needs to store
a chunk rather than the full file. However, nodes still echo chunks to each other,
leading to a lot of communication.
AVID-FP improves in communication because storage nodes only echo fingerprints rather than full chunks.
AVID-M improves over AVID-FP in that it drastically reduces the fingerprint size and hence the communication.
ACeD allows for a trade-off of communication and storage
with adversarial resilience.

In terms of communication and storage,
our Semi-AVID-PR scheme $\OURsavid$ 
(Sec.~\ref{sec:protocol})
is among the best-of-class
(tied with AVID-M),
while providing superior resilience ($t<n/2$ vs. $t<n/3$)
and provable retrievability (the lack thereof limits application
of AVID-M to Validium rollups).
Our Semi-AVID-PR scheme
outperforms ACeD in communication and storage
by at least $7\times$.
The net data throughput of $\approx0.54\,\mathrm{MB/s}$
for $(n, k)=(256,85)$
corresponding to $\approx2{,}700\,\mathrm{tx/s}$
(\cf Section~\ref{sec:evaluation-computation})
entails $\approx1.7\,\mathrm{MB/s}$
communication bandwidth usage
($\approx70\,\mathrm{MB}$ in $\approx41\,\mathrm{s}$),
which is feasible even via consumer-grade Internet connectivity.
Finally, it should be noted that the VID-based schemes
in Table~\ref{tab:comparison} have resilience $t$ at most
$t < n/3$. In that regime, $\OURsavid$ matches or exceeds
the communication- and storage-efficiency of VID-based schemes.
However, like ACeD, $\OURsavid$ also supports higher resilience up to $t < n/2$.
In this regime, the overhead from erasure coding increases,
as for ACeD,
but still outperforms ACeD.

\section{Application to Data Availability Sampling}
\label{sec:data-availability-sampling}

In common blockchain designs every block consists of a meta data header and transaction content.
\emph{Full nodes} download
the full chain and validate all transactions.
However, a resource-limited node 
can instead participate as a \emph{light node}.\footnote{Light nodes also occur in the context of sharding, where each node is assigned to a shard and behaves in-shard
as a full node and out-of-shard
as a light node.}
Then, it only processes block headers.
If a block contained an invalid transaction,
it would be rejected by full nodes
but its header would be accepted
by a light node unable to inspect the block content and verify transaction validity.
To prevent this, full nodes can produce an \emph{invalid transaction fraud proof} \cite{albassam}.
To take full nodes' ability to issue such fraud proofs, a malicious block producer can
withhold parts of the block content.
Full nodes would then 
reject the block until its content 
is
fully available, but light nodes would not notice the missing content.
The absence of an invalid transaction fraud proof can thus mean two things: either
the block is valid,
or full nodes are unable 
to verify the block due to missing data.
To rule out the second possibility,
data availability sampling schemes for light nodes were introduced.

Data availability schemes using Reed-Solomon codes were proposed in \cite{albassam},
where
the block producer encodes the $k$ chunks block content with a $(2k, k)$ Reed-Solomon (RS) code. Light nodes randomly query a few chunks of the encoded block content. The block is accepted only if the queried chunks are available. For a block to be widely accepted by light nodes, most of the light nodes' queried chunks have to be available. Quickly, light nodes' queries cover more than 50\% of coded chunks of the block and any remaining missing chunks can be recovered using the RS code. 
It is therefore no longer possible to trick light nodes into accepting a block while withholding data to prevent invalid transaction fraud proofs.
However, a malicious block producer could invalidly encode the block.
Decoding would then not consistently recover the original chunks' data.
Full nodes can detect invalid encoding and issue a fraud proof for light nodes.
But, the size of such proofs in this scheme
is commensurate to the block content size---defying the idea of light nodes downloading less than the full block.
Subsequent works \cite{albassam,cmt,DBLP:conf/itw/MitraTD20,DBLP:journals/corr/abs-2108-13332} focussed on reducing the fraud proof size,
but drawbacks remain (\eg, complexity, timing assumptions).

A different approach is to make it impossible for block producers to invalidly encode data. Such schemes can be achieved using polynomial commitments \cite{kate},
where the block is interpreted as a low-degree polynomial,
the commitment to which is included in the block header and gets opened
at locations randomly sampled by light clients.
This effectively enforces valid RS encoding.
Schemes of this flavor however require to compute an evaluation witness for each query,
which despite recent algorithmic improvements is still computationally heavy \cite{fk,alin,tab}.

Algorithm $\OURsavid.\savidCommit$ of our Semi-AVID-PR scheme
is suitable for the application at hand.
It can commit to a block $B$ such that:
(a) The commitment can be opened to \emph{chunks}, but only of a valid RS encoding of $B$.
Computing and verifying these openings is practically efficient.
\emph{This enables data availability sampling.}
(b) The commitment can be opened to \emph{entries} of the original block $B$.
The openings are short and can be produced and verified practically efficiently.
\emph{This enables the invalid transaction fraud proofs of \cite{albassam}.}

Let $\BU \equiv \operatorname{AsMatrix}_{L \times k}(B)$ with columns $\Bu_1, ..., \Bu_k$.
An opening $(\Bc_i, i, (h_1, ..., h_k))$ to chunk $i$ of an RS encoding of $B$ is computed as:
\begin{IEEEeqnarray}{C}
    \Bc_i \gets [\codeEncode^{\otimes L}(\BU)]_i
    \quad
    (h_1, ..., h_k) \gets \VC^{\otimes k}(\BU)
\end{IEEEeqnarray}
An opening $(\Bc, i, (h_1, ..., h_k))$ to chunk $i$ of an RS encoding of a block with commitment $C$ is verified as:
\begin{IEEEeqnarray}{C}
    C \overset{?}{=} \CRHF(h_1\|...\|h_k)
    \land
    [\codeEncode(h_1, ..., h_k)]_i \overset{?}{=} \VC(\Bc)
    \IEEEeqnarraynumspace
\end{IEEEeqnarray}
An opening $([\Bu_j]_i, i, j, (h_1, ..., h_k), w)$
to the entry at $(i,j)$ of the matrix $\BU$ corresponding to $B$ is
computed as:
\begin{IEEEeqnarray}{C}
    (h_1, ..., h_k) \gets \VC^{\otimes k}(\BU)
    \quad
    w \gets \mathsf{VC.OpenEntry}(\cryptoPP, \Bu_j, i)
    \IEEEeqnarraynumspace
\end{IEEEeqnarray}
An opening $(y, i, j, (h_1, ..., h_k), w)$ 
to entry $[\Bu_j]_i$ of
a block with commitment $C$ is verified as:
\begin{IEEEeqnarray}{C}
    C \overset{?}{=} \CRHF(h_1\|...\|h_k)
    \land
    \mathsf{VC.VerifyEntry}(\cryptoPP, h_j, i, y, w) \overset{?}{=} \TRUE
    \IEEEeqnarraynumspace
\end{IEEEeqnarray}

For more details on the application of $\OURsavid$
to data availability sampling
see Appendix~\ref{sec:appendix-data-availability-sampling}.

\section*{Acknowledgment}
The authors thank Mohammad Ali Maddah-Ali and Dionysis Zindros for fruitful discussions.
KN is supported by a gift from IOG (Input-Output Global) and a grant from the National Science Foundation Center for Science of Information.
JN is supported by the Protocol Labs PhD Fellowship, a gift from Ethereum Foundation, and the Reed-Hodgson Stanford Graduate Fellowship.

\bibliographystyle{ACM-Reference-Format}
\bibliography{references}

\appendix

\section{Calculations for Table~\ref{tab:comparison}}
\label{sec:appendix-calculations}
\begin{figure}%
    \centering
    \begin{tikzpicture}
        \begin{axis}[
            mysimpleplot,
            ylabel={Communication\\and storage cost [MB]},
            xticklabels={},
            height=0.45\linewidth,
            width=\linewidth,
            ymin=0,
            xmin=4e3, xmax=100e3,
            scaled y ticks=real:1e6,   %
            scaled x ticks=real:1e3,   %
            xtick scale label code/.code={},
            ytick scale label code/.code={},
        ]
        
            \addplot [myparula11] table [x=c,y=storage] {figures/aced.txt};

        \end{axis}
        \begin{axis}[
            mysimpleplot,
            ylabel={Invalid encoding\\fraud proof size [kB]},
            xlabel={Base layer symbol size [kB]},
            height=0.45\linewidth,
            width=\linewidth,
            yshift=-2.5cm,
            xmin=4e3, xmax=100e3,
            scaled y ticks=real:1e3,   %
            scaled x ticks=real:1e3,   %
            xtick scale label code/.code={},
            ytick scale label code/.code={},
        ]
        
            \addplot [myparula11] table [x=c,y=proof] {figures/aced.txt};
            
        \end{axis}
    \end{tikzpicture}%
    \vspace{-0.5em}%
    \caption{Communication and storage cost (top) and invalid encoding fraud proof size (bottom) as a function of the base layer symbol size $c$, when dispersing a file of size $22\,\mathrm{MB}$ among $1024$ nodes using ACeD \cite{aced} with resilience $t = 0.33n$.}
    \label{fig:aced}
\end{figure}
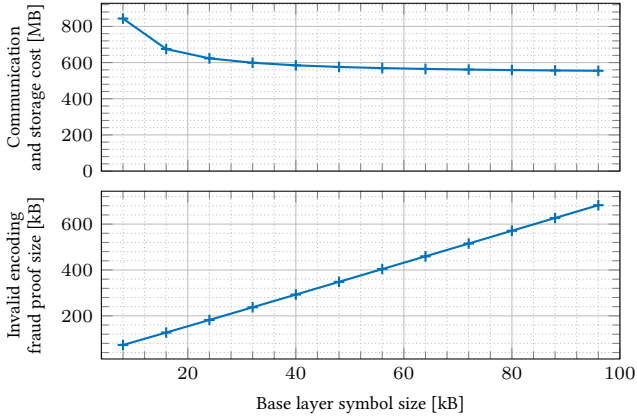

\newcommand{\lambdaHASH}[0]{\ensuremath{\lambda_{\mathrm{h}}}}
\newcommand{\lambdaCOMMIT}[0]{\ensuremath{\lambda_{\mathrm{c}}}}

Table~\ref{tab:comparison} shows
communication and storage required to disperse a block of size $|B| = 22\,\mathrm{MB}$ among $n = 1024$ storage nodes using different schemes.
We provide the corresponding calculations here.
We denote communication and storage costs as $C$ and $S$, respectively.
We assume the size of a hash is $\lambdaHASH = 32\,\mathrm{B}$,
and the size of an LVC is $\lambdaCOMMIT = 48\,\mathrm{B}$ (as in BLS12-381).
Given adversarial resilience $t$,
we choose
$k \triangleq n-2t$.
In the following, we do not
account for
signatures on messages
or for blockchain interaction (\eg, ACeD).

\begin{itemize}
    \item Repetition (uncoded) scheme:
        \begin{IEEEeqnarray}{C}
        C = S = n|B|
        \end{IEEEeqnarray}
    \item AVID:
        \begin{IEEEeqnarray}{rCl}
        C &=& \left(\frac{|B|}{k} + n\lambdaHASH\right)\left(n + n^2\right)   \\
        S &=& n\left(\frac{|B|}{k} + n\lambdaHASH\right)
        \end{IEEEeqnarray}
    \item AVID-FP:
        \begin{IEEEeqnarray}{rCl}
        C &=& n\left(\frac{|B|}{k} + (n+k)\lambdaHASH\right) + n^2(n+k)\lambdaHASH   \\ 
        S &=& n\left(\frac{|B|}{k} + (n+k)\lambdaHASH\right)
        \end{IEEEeqnarray}
    \item AVID-M:
        \begin{IEEEeqnarray}{rCl}
        C &=& n\left(\frac{|B|}{k} + (1+\log_2 n)\lambdaHASH\right) + n^2 \lambdaHASH   \\ 
        S &=& n\left(\frac{|B|}{k} + (1+\log_2 n)\lambdaHASH\right)
        \end{IEEEeqnarray}
    \item ACeD:
        \begin{IEEEeqnarray}{C}
        C = S = n\left(  t' \lambdaHASH+\frac{|B|}{n r \lambda}+\frac{(2 q-1) |B| \lambdaHASH}{n r c \lambda} \log_{q r} \frac{|B|}{c t' r}  \right)
        \end{IEEEeqnarray}
        
        Parameters: %
        \begin{IEEEeqnarray}{rClrClrClrCl}
            t' &=& 16   \quad
            &
            r &=& 0.25   \quad
            &
            q &=& 8   \quad
            &
            d &=& 8   \quad   \\
            c &=& 40\,\mathrm{kB}   \quad
            &
            \eta &=& 0.875   \quad
            &
            \lambda &=& \frac{1-2 t/n}{\ln\left( \frac{1}{1-\eta} \right)}   \quad
            &
            &&
        \end{IEEEeqnarray}
        
        As illustrated in Figure~\ref{fig:aced}, the communication and storage cost of ACeD can be decreased by increasing the base layer symbol size $c$, at the expense of an increased invalid encoding fraud proof size.
        We picked what seemed to us to be a
        reasonable tradeoff,
        where increasing the fraud proof
        size further does not lead to any
        more substantial gains
        in communication and storage.

    \item Semi-AVID-PR:
        \begin{IEEEeqnarray}{rCl}
        C &=& n\left(\frac{|B|}{k} + k\lambdaCOMMIT\right)   \\
        S &=& n\left(\frac{|B|}{k} + k\lambdaCOMMIT\right)
        \end{IEEEeqnarray}
\end{itemize}

\section{Preliminaries}
\label{sec:appendix-preliminaries}
\label{sec:appendix-cryptographic-preliminaries}

\begin{algorithm}[t]
    \caption{Existential forgery game ($\gameEF$)
    against $\sig = (\sigKeygen, \sigSign, \sigVerify)$}
    \label{alg:game-sig-forgery}
    \begin{algorithmic}[1]
        \State $\CM \gets \emptyset$
        \State $(\cryptoPK, \cryptoSK) \gets \sigKeygen(1^\lambda)$
        \Function{$\CO^{\mathrm{sign}}$}{m}
            \State $\CM \gets \CM \cup \{ m \}$
            \State \Return $\sigSign(\cryptoSK, m)$
        \EndFunction
        \State $(m, \sigma) \gets \CA^{\CO^{\mathrm{sign}}(.)}_{\gameEF}(\cryptoPK)$
        \State \Return $m \not\in \CM \land \sigVerify(\cryptoPK, m, \sigma) = \TRUE$
    \end{algorithmic}
\end{algorithm}

\begin{algorithm}[t]
    \caption{Collision finding game ($\gameCF$) against $\crhf = (\crhfGen, \crhfH)$}
    \label{alg:game-crhf-collision}
    \begin{algorithmic}[1]
        \State $s \gets \crhfGen(1^\lambda)$
        \State $(x, x') \gets \CA_{\gameCF}(s)$
        \State \Return $x \neq x' \land \CRHF(x) = \CRHF(x')$
    \end{algorithmic}
\end{algorithm}

\begin{algorithm}[t]
    \caption{Binding game ($\gameVCB$) against $\slvc = (\slvcSetup, \slvcCommit)$}
    \label{alg:game-lvc-binding}
    \begin{algorithmic}[1]
        \State $\cryptoPP \gets \slvcSetup(1^\lambda)$
        \State $(\Bv, \Bv') \gets \CA_{\gameVCB}(\cryptoPP)$
        \State \Return $\Bv \neq \Bv' \land \slvcCommit(\Bv) = \slvcCommit(\Bv')$
    \end{algorithmic}
\end{algorithm}

\begin{definition}[Hash Function]
\label{def:hash_func}
A hash function is a pair of PPT algorithms $(\mathsf{Gen}, \mathsf{H})$ \cite{katz}, such that:
\begin{itemize}
    \item $\mathsf{Gen}\colon 1^\lambda \mapsto s$: 
    takes as input a security parameter $\lambda$ and outputs a randomly sampled key $s$,
    \item $\mathsf{H}\colon (s, x) \mapsto h$: 
    takes as input a key $s$ and a string $x \in \{0, 1\}^*$ and outputs a string $h\in \{0, 1\}^{\lambda}$.
\end{itemize}
\end{definition}

\begin{definition}[Digital Signature Scheme]
\label{def:sig}
A digital signature scheme $\sig = (\mathsf{KeyGen}, \mathsf{Sign}, \mathsf{Verify})$
\cite{katz} consists of three PPT algorithms, such that:
\begin{itemize}
    \item $\mathsf{KeyGen}\colon 1^\lambda \mapsto (\cryptoPK, \cryptoSK)$: takes as input a security parameter $\lambda$ and outputs a public key $\cryptoPK$ and a secret key $\cryptoSK$,
    \item $\mathsf{Sign}\colon (\cryptoSK, m) \mapsto \sigma$:
    takes as input a secret key $\cryptoSK$ and a message $m$ and outputs a signature $\sigma$,
    \item $\mathsf{Verify}\colon (\cryptoPK, m, \sigma) \mapsto b \in \{\TRUE, \FALSE\}$:
    takes as input a public key $\cryptoPK$, a message $m$ and a signature $\sigma$ and outputs a boolean $b$ indicating whether the signature is valid.
\end{itemize}
\end{definition}

\begin{definition}[Deterministic Vector Commitment Scheme]
\label{def:vc}
A deterministic vector commitment scheme  $\mathsf{VC} = (\mathsf{Setup}, \mathsf{Commit},\allowbreak \mathsf{OpenEntry}, \mathsf{VerifyEntry})$ \cite{vcs, kate} consists of four PPT algorithms, such that:
\begin{itemize}
    \item $\mathsf{Setup}\colon 1^\lambda\mapsto \cryptoPP$: 
    takes as input a security parameter $\lambda$ and outputs some public parameters $\cryptoPP$,
    \item $\mathsf{Commit}\colon (\cryptoPP, \Bv) \mapsto C$: 
    takes as input the public parameters $\cryptoPP$ and a vector $\Bv$ and outputs a commitment $C$,
    \item $\mathsf{OpenEntry}\colon (\cryptoPP, \Bv, i) \mapsto \pi_i$:
    takes as input the public parameters $\cryptoPP$, a vector $\Bv$, a position $i$ and returns a proof $\pi_i$ attesting to the fact that $[\Bv]_{i}$ is the $i$-th entry of $\Bv$,
    \item $\mathsf{VerifyEntry}\colon (\cryptoPP, C, i, y,  \pi) \mapsto b \in \{\TRUE, \FALSE\}$:
    takes as input the public parameters $\cryptoPP$, a commitment $C$, a position $i$, a value $y$, and an opening proof $\pi$, and returns a boolean $b$ indicating whether
    $\pi$ is a proof attesting to the fact that $C$ is a commitment to a vector $\Bv$ such that $[\Bv]_{i} = y$.
\end{itemize}
\end{definition}

\section{Additional Evaluation Plots}
\label{sec:appendix-additional-evaluation-plots}
\begin{figure}[t]
    \centering
    \begin{tikzpicture}
        \small
        \def\CODERATE{25}
        \def\CODEN{128}
        \begin{scope}[xshift=-0.05\linewidth,yshift=0.21\linewidth]
            \node [align=left,anchor=south west] at (0,0) {(a) $(n,k)=(128,32)$};
        \end{scope}
        \begin{axis}[
            semiavidprexperiment2,
            ymode=log, xmode=log,
            xshift=0.0\linewidth,
            yshift=0.0\linewidth,
        ]
        
            \addplot [myparula61,experiment-rate\CODERATE-n\CODEN,discard if neq={args_n}{\CODEN}] table [x=net_file_size_bytes,y=scenario_disperse_runtime_client_seconds] {figures/experiments2/data-experiments2-rate\CODERATE-bls12-381.txt};
            \label{leg:experiments-aggregate-rate\CODERATE-n\CODEN-disperse}
            \addlegendentry{Disperse};
        
            \addplot [myparula71,experiment-rate\CODERATE-n\CODEN,discard if neq={args_n}{\CODEN}] table [x=net_file_size_bytes,y=scenario_retrieve_runtime_client_seconds] {figures/experiments2/data-experiments2-rate\CODERATE-bls12-381.txt};
            \label{leg:experiments-aggregate-rate\CODERATE-n\CODEN-retrieve}
            \addlegendentry{Retrieve};
            
            \legend{};
            
        \end{axis}
        \def\CODEN{256}
        \begin{scope}[xshift=0.28\linewidth,yshift=0.21\linewidth]
            \node [align=left,anchor=south west] at (0,0) {(b) $(n,k)=(256,64)$};
        \end{scope}
        \begin{axis}[
            semiavidprexperiment2,
            ymode=log, xmode=log,
            xshift=0.33\linewidth,
            yshift=0.0\linewidth,
        ]
        
            \addplot [myparula61,experiment-rate\CODERATE-n\CODEN,discard if neq={args_n}{\CODEN}] table [x=net_file_size_bytes,y=scenario_disperse_runtime_client_seconds] {figures/experiments2/data-experiments2-rate\CODERATE-bls12-381.txt};
            \label{leg:experiments-aggregate-rate\CODERATE-n\CODEN-disperse}
            \addlegendentry{Disperse};
        
            \addplot [myparula71,experiment-rate\CODERATE-n\CODEN,discard if neq={args_n}{\CODEN}] table [x=net_file_size_bytes,y=scenario_retrieve_runtime_client_seconds] {figures/experiments2/data-experiments2-rate\CODERATE-bls12-381.txt};
            \label{leg:experiments-aggregate-rate\CODERATE-n\CODEN-retrieve}
            \addlegendentry{Retrieve};
            
            \legend{};
            
        \end{axis}
        \def\CODEN{1024}
        \begin{scope}[xshift=0.61\linewidth,yshift=0.21\linewidth]
            \node [align=left,anchor=south west] at (0,0) {(c) $(n,k)=(1024,256)$};
        \end{scope}
        \begin{axis}[
            semiavidprexperiment2,
            ymode=log, xmode=log,
            xshift=0.66\linewidth,
            yshift=0.0\linewidth,
        ]
        
            \addplot [myparula61,experiment-rate\CODERATE-n\CODEN,discard if neq={args_n}{\CODEN}] table [x=net_file_size_bytes,y=scenario_disperse_runtime_client_seconds] {figures/experiments2/data-experiments2-rate\CODERATE-bls12-381.txt};
            \label{leg:experiments-aggregate-rate\CODERATE-n\CODEN-disperse}
            \addlegendentry{Disperse};
        
            \addplot [myparula71,experiment-rate\CODERATE-n\CODEN,discard if neq={args_n}{\CODEN}] table [x=net_file_size_bytes,y=scenario_retrieve_runtime_client_seconds] {figures/experiments2/data-experiments2-rate\CODERATE-bls12-381.txt};
            \label{leg:experiments-aggregate-rate\CODERATE-n\CODEN-retrieve}
            \addlegendentry{Retrieve};
            
            \legend{};
            
        \end{axis}
        \def\CODERATE{33}
        \def\CODEN{128}
        \begin{scope}[xshift=-0.05\linewidth,yshift=-0.11\linewidth]
            \node [align=left,anchor=south west] at (0,0) {(d) $(n,k)=(128,43)$};
        \end{scope}
        \begin{axis}[
            semiavidprexperiment2,
            ymode=log, xmode=log,
            xshift=0.0\linewidth,
            yshift=-0.32\linewidth,
        ]
        
            \addplot [myparula61,experiment-rate\CODERATE-n\CODEN,discard if neq={args_n}{\CODEN}] table [x=net_file_size_bytes,y=scenario_disperse_runtime_client_seconds] {figures/experiments2/data-experiments2-rate\CODERATE-bls12-381.txt};
            \label{leg:experiments-aggregate-rate\CODERATE-n\CODEN-disperse}
            \addlegendentry{Disperse};
        
            \addplot [myparula71,experiment-rate\CODERATE-n\CODEN,discard if neq={args_n}{\CODEN}] table [x=net_file_size_bytes,y=scenario_retrieve_runtime_client_seconds] {figures/experiments2/data-experiments2-rate\CODERATE-bls12-381.txt};
            \label{leg:experiments-aggregate-rate\CODERATE-n\CODEN-retrieve}
            \addlegendentry{Retrieve};
            
            \legend{};
            
        \end{axis}
        \def\CODEN{256}
        \begin{scope}[xshift=0.28\linewidth,yshift=-0.11\linewidth]
            \node [align=left,anchor=south west] at (0,0) {(e) $(n,k)=(256,85)$};
        \end{scope}
        \begin{axis}[
            semiavidprexperiment2,
            ymode=log, xmode=log,
            xshift=0.33\linewidth,
            yshift=-0.32\linewidth,
        ]
        
            \addplot [myparula61,experiment-rate\CODERATE-n\CODEN,discard if neq={args_n}{\CODEN}] table [x=net_file_size_bytes,y=scenario_disperse_runtime_client_seconds] {figures/experiments2/data-experiments2-rate\CODERATE-bls12-381.txt};
            \label{leg:experiments-aggregate-rate\CODERATE-n\CODEN-disperse}
            \addlegendentry{Disperse};
        
            \addplot [myparula71,experiment-rate\CODERATE-n\CODEN,discard if neq={args_n}{\CODEN}] table [x=net_file_size_bytes,y=scenario_retrieve_runtime_client_seconds] {figures/experiments2/data-experiments2-rate\CODERATE-bls12-381.txt};
            \label{leg:experiments-aggregate-rate\CODERATE-n\CODEN-retrieve}
            \addlegendentry{Retrieve};
            
            \legend{};
            
        \end{axis}
        \def\CODEN{1024}
        \begin{scope}[xshift=0.61\linewidth,yshift=-0.11\linewidth]
            \node [align=left,anchor=south west] at (0,0) {(f) $(n,k)=(1024,341)$};
        \end{scope}
        \begin{axis}[
            semiavidprexperiment2,
            ymode=log, xmode=log,
            xshift=0.66\linewidth,
            yshift=-0.32\linewidth,
        ]
        
            \addplot [myparula61,experiment-rate\CODERATE-n\CODEN,discard if neq={args_n}{\CODEN}] table [x=net_file_size_bytes,y=scenario_disperse_runtime_client_seconds] {figures/experiments2/data-experiments2-rate\CODERATE-bls12-381.txt};
            \label{leg:experiments-aggregate-rate\CODERATE-n\CODEN-disperse}
            \addlegendentry{Disperse};
        
            \addplot [myparula71,experiment-rate\CODERATE-n\CODEN,discard if neq={args_n}{\CODEN}] table [x=net_file_size_bytes,y=scenario_retrieve_runtime_client_seconds] {figures/experiments2/data-experiments2-rate\CODERATE-bls12-381.txt};
            \label{leg:experiments-aggregate-rate\CODERATE-n\CODEN-retrieve}
            \addlegendentry{Retrieve};
            
            \legend{};
            
        \end{axis}
        \def\CODERATE{45}
        \def\CODEN{128}
        \begin{scope}[xshift=-0.05\linewidth,yshift=-0.43\linewidth]
            \node [align=left,anchor=south west] at (0,0) {(g) $(n,k)=(128,58)$};
        \end{scope}
        \begin{axis}[
            semiavidprexperiment2,
            ymode=log, xmode=log,
            xshift=0.0\linewidth,
            yshift=-0.64\linewidth,
        ]
        
            \addplot [myparula61,experiment-rate\CODERATE-n\CODEN,discard if neq={args_n}{\CODEN}] table [x=net_file_size_bytes,y=scenario_disperse_runtime_client_seconds] {figures/experiments2/data-experiments2-rate\CODERATE-bls12-381.txt};
            \label{leg:experiments-aggregate-rate\CODERATE-n\CODEN-disperse}
            \addlegendentry{Disperse};
        
            \addplot [myparula71,experiment-rate\CODERATE-n\CODEN,discard if neq={args_n}{\CODEN}] table [x=net_file_size_bytes,y=scenario_retrieve_runtime_client_seconds] {figures/experiments2/data-experiments2-rate\CODERATE-bls12-381.txt};
            \label{leg:experiments-aggregate-rate\CODERATE-n\CODEN-retrieve}
            \addlegendentry{Retrieve};
            
            \legend{};
            
        \end{axis}
        \def\CODEN{256}
        \begin{scope}[xshift=0.28\linewidth,yshift=-0.43\linewidth]
            \node [align=left,anchor=south west] at (0,0) {(h) $(n,k)=(256,115)$};
        \end{scope}
        \begin{axis}[
            semiavidprexperiment2,
            ymode=log, xmode=log,
            xshift=0.33\linewidth,
            yshift=-0.64\linewidth,
        ]
        
            \addplot [myparula61,experiment-rate\CODERATE-n\CODEN,discard if neq={args_n}{\CODEN}] table [x=net_file_size_bytes,y=scenario_disperse_runtime_client_seconds] {figures/experiments2/data-experiments2-rate\CODERATE-bls12-381.txt};
            \label{leg:experiments-aggregate-rate\CODERATE-n\CODEN-disperse}
            \addlegendentry{Disperse};
        
            \addplot [myparula71,experiment-rate\CODERATE-n\CODEN,discard if neq={args_n}{\CODEN}] table [x=net_file_size_bytes,y=scenario_retrieve_runtime_client_seconds] {figures/experiments2/data-experiments2-rate\CODERATE-bls12-381.txt};
            \label{leg:experiments-aggregate-rate\CODERATE-n\CODEN-retrieve}
            \addlegendentry{Retrieve};
            
            \legend{};
            
        \end{axis}
        \def\CODEN{1024}
        \begin{scope}[xshift=0.61\linewidth,yshift=-0.43\linewidth]
            \node [align=left,anchor=south west] at (0,0) {(i) $(n,k)=(1024,461)$};
        \end{scope}
        \begin{axis}[
            semiavidprexperiment2,
            ymode=log, xmode=log,
            xshift=0.66\linewidth,
            yshift=-0.64\linewidth,
        ]
        
            \addplot [myparula61,experiment-rate\CODERATE-n\CODEN,discard if neq={args_n}{\CODEN}] table [x=net_file_size_bytes,y=scenario_disperse_runtime_client_seconds] {figures/experiments2/data-experiments2-rate\CODERATE-bls12-381.txt};
            \label{leg:experiments-aggregate-rate\CODERATE-n\CODEN-disperse}
            \addlegendentry{Disperse};
        
            \addplot [myparula71,experiment-rate\CODERATE-n\CODEN,discard if neq={args_n}{\CODEN}] table [x=net_file_size_bytes,y=scenario_retrieve_runtime_client_seconds] {figures/experiments2/data-experiments2-rate\CODERATE-bls12-381.txt};
            \label{leg:experiments-aggregate-rate\CODERATE-n\CODEN-retrieve}
            \addlegendentry{Retrieve};
            
            \legend{};

            \addlegendimage{experiment-rate25-n128,black,mark=none};
            \label{leg:experiments-aggregate-rate25}
            \addlegendimage{experiment-rate33-n128,black,mark=none};
            \label{leg:experiments-aggregate-rate33}
            \addlegendimage{experiment-rate45-n128,black,mark=none};
            \label{leg:experiments-aggregate-rate45}
            
            \addlegendimage{experiment-rate25-n128,black,only marks};
            \label{leg:experiments-aggregate-n128}
            \addlegendimage{experiment-rate25-n256,black,only marks};
            \label{leg:experiments-aggregate-n256}
            \addlegendimage{experiment-rate25-n1024,black,only marks};
            \label{leg:experiments-aggregate-n1024}
            
            \addlegendimage{myparula61,only marks,mark=square*};
            \label{leg:experiments-aggregate-disperse}
            \addlegendimage{myparula71,only marks,mark=square*};
            \label{leg:experiments-aggregate-retrieve}
            
            \legend{};
            
        \end{axis}
    \end{tikzpicture}%
    \vspace{-0.5em}%
    \caption{%
    Single-thread runtime (ordinate, in seconds)
    of
    $\savidDisperse$
    (\ref{leg:experiments-aggregate-disperse})
    and
    $\savidRetrieve$
    (\ref{leg:experiments-aggregate-retrieve})
    on AMD Opteron 6378 processor
    for varying file sizes (abscissa, in $10^6$ bytes;
    for varying $L$)
    in BLS12-381 curve.
    Rows: code rates $k/n\approx0.25,0.33,0.45$
    (\ref{leg:experiments-aggregate-rate25},
    \ref{leg:experiments-aggregate-rate33},
    \ref{leg:experiments-aggregate-rate45}).
    Columns: system sizes $n=128,256,1024$
    (\ref{leg:experiments-aggregate-n128},
    \ref{leg:experiments-aggregate-n256},
    \ref{leg:experiments-aggregate-n1024}).
    (Disaggregated: Figure~\ref{fig:experiments})%
    }
    \label{fig:experiments-aggregate}
\end{figure}

\input{fig_experiments_bn254_steps}
\begin{figure}[t]
    \centering
    \begin{tikzpicture}
        \small
        \def\CODERATE{25}
        \def\CODEN{128}
        \begin{scope}[xshift=-0.05\linewidth,yshift=0.21\linewidth]
            \node [align=left,anchor=south west] at (0,0) {(a) $(n,k)=(128,32)$};
        \end{scope}
        \begin{axis}[
            semiavidprexperiment2,
            ymode=log, xmode=log,
            xshift=0.0\linewidth,
            yshift=0.0\linewidth,
        ]
        
            \addplot [myparula61,experiment-rate\CODERATE-n\CODEN,discard if neq={args_n}{\CODEN}] table [x=net_file_size_bytes,y=scenario_disperse_runtime_client_seconds] {figures/experiments2/data-experiments2-rate\CODERATE-bn254.txt};
            \label{leg:experiments-bn254-aggregate-rate\CODERATE-n\CODEN-disperse}
            \addlegendentry{Disperse};
        
            \addplot [myparula71,experiment-rate\CODERATE-n\CODEN,discard if neq={args_n}{\CODEN}] table [x=net_file_size_bytes,y=scenario_retrieve_runtime_client_seconds] {figures/experiments2/data-experiments2-rate\CODERATE-bn254.txt};
            \label{leg:experiments-bn254-aggregate-rate\CODERATE-n\CODEN-retrieve}
            \addlegendentry{Retrieve};
            
            \legend{};
            
        \end{axis}
        \def\CODEN{256}
        \begin{scope}[xshift=0.28\linewidth,yshift=0.21\linewidth]
            \node [align=left,anchor=south west] at (0,0) {(b) $(n,k)=(256,64)$};
        \end{scope}
        \begin{axis}[
            semiavidprexperiment2,
            ymode=log, xmode=log,
            xshift=0.33\linewidth,
            yshift=0.0\linewidth,
        ]
        
            \addplot [myparula61,experiment-rate\CODERATE-n\CODEN,discard if neq={args_n}{\CODEN}] table [x=net_file_size_bytes,y=scenario_disperse_runtime_client_seconds] {figures/experiments2/data-experiments2-rate\CODERATE-bn254.txt};
            \label{leg:experiments-bn254-aggregate-rate\CODERATE-n\CODEN-disperse}
            \addlegendentry{Disperse};
        
            \addplot [myparula71,experiment-rate\CODERATE-n\CODEN,discard if neq={args_n}{\CODEN}] table [x=net_file_size_bytes,y=scenario_retrieve_runtime_client_seconds] {figures/experiments2/data-experiments2-rate\CODERATE-bn254.txt};
            \label{leg:experiments-bn254-aggregate-rate\CODERATE-n\CODEN-retrieve}
            \addlegendentry{Retrieve};
            
            \legend{};
            
        \end{axis}
        \def\CODEN{1024}
        \begin{scope}[xshift=0.61\linewidth,yshift=0.21\linewidth]
            \node [align=left,anchor=south west] at (0,0) {(c) $(n,k)=(1024,256)$};
        \end{scope}
        \begin{axis}[
            semiavidprexperiment2,
            ymode=log, xmode=log,
            xshift=0.66\linewidth,
            yshift=0.0\linewidth,
        ]
        
            \addplot [myparula61,experiment-rate\CODERATE-n\CODEN,discard if neq={args_n}{\CODEN}] table [x=net_file_size_bytes,y=scenario_disperse_runtime_client_seconds] {figures/experiments2/data-experiments2-rate\CODERATE-bn254.txt};
            \label{leg:experiments-bn254-aggregate-rate\CODERATE-n\CODEN-disperse}
            \addlegendentry{Disperse};
        
            \addplot [myparula71,experiment-rate\CODERATE-n\CODEN,discard if neq={args_n}{\CODEN}] table [x=net_file_size_bytes,y=scenario_retrieve_runtime_client_seconds] {figures/experiments2/data-experiments2-rate\CODERATE-bn254.txt};
            \label{leg:experiments-bn254-aggregate-rate\CODERATE-n\CODEN-retrieve}
            \addlegendentry{Retrieve};
            
            \legend{};
            
        \end{axis}
        \def\CODERATE{33}
        \def\CODEN{128}
        \begin{scope}[xshift=-0.05\linewidth,yshift=-0.11\linewidth]
            \node [align=left,anchor=south west] at (0,0) {(d) $(n,k)=(128,43)$};
        \end{scope}
        \begin{axis}[
            semiavidprexperiment2,
            ymode=log, xmode=log,
            xshift=0.0\linewidth,
            yshift=-0.32\linewidth,
        ]
        
            \addplot [myparula61,experiment-rate\CODERATE-n\CODEN,discard if neq={args_n}{\CODEN}] table [x=net_file_size_bytes,y=scenario_disperse_runtime_client_seconds] {figures/experiments2/data-experiments2-rate\CODERATE-bn254.txt};
            \label{leg:experiments-bn254-aggregate-rate\CODERATE-n\CODEN-disperse}
            \addlegendentry{Disperse};
        
            \addplot [myparula71,experiment-rate\CODERATE-n\CODEN,discard if neq={args_n}{\CODEN}] table [x=net_file_size_bytes,y=scenario_retrieve_runtime_client_seconds] {figures/experiments2/data-experiments2-rate\CODERATE-bn254.txt};
            \label{leg:experiments-bn254-aggregate-rate\CODERATE-n\CODEN-retrieve}
            \addlegendentry{Retrieve};
            
            \legend{};
            
        \end{axis}
        \def\CODEN{256}
        \begin{scope}[xshift=0.28\linewidth,yshift=-0.11\linewidth]
            \node [align=left,anchor=south west] at (0,0) {(e) $(n,k)=(256,85)$};
        \end{scope}
        \begin{axis}[
            semiavidprexperiment2,
            ymode=log, xmode=log,
            xshift=0.33\linewidth,
            yshift=-0.32\linewidth,
        ]
        
            \addplot [myparula61,experiment-rate\CODERATE-n\CODEN,discard if neq={args_n}{\CODEN}] table [x=net_file_size_bytes,y=scenario_disperse_runtime_client_seconds] {figures/experiments2/data-experiments2-rate\CODERATE-bn254.txt};
            \label{leg:experiments-bn254-aggregate-rate\CODERATE-n\CODEN-disperse}
            \addlegendentry{Disperse};
        
            \addplot [myparula71,experiment-rate\CODERATE-n\CODEN,discard if neq={args_n}{\CODEN}] table [x=net_file_size_bytes,y=scenario_retrieve_runtime_client_seconds] {figures/experiments2/data-experiments2-rate\CODERATE-bn254.txt};
            \label{leg:experiments-bn254-aggregate-rate\CODERATE-n\CODEN-retrieve}
            \addlegendentry{Retrieve};
            
            \legend{};
            
        \end{axis}
        \def\CODEN{1024}
        \begin{scope}[xshift=0.61\linewidth,yshift=-0.11\linewidth]
            \node [align=left,anchor=south west] at (0,0) {(f) $(n,k)=(1024,341)$};
        \end{scope}
        \begin{axis}[
            semiavidprexperiment2,
            ymode=log, xmode=log,
            xshift=0.66\linewidth,
            yshift=-0.32\linewidth,
        ]
        
            \addplot [myparula61,experiment-rate\CODERATE-n\CODEN,discard if neq={args_n}{\CODEN}] table [x=net_file_size_bytes,y=scenario_disperse_runtime_client_seconds] {figures/experiments2/data-experiments2-rate\CODERATE-bn254.txt};
            \label{leg:experiments-bn254-aggregate-rate\CODERATE-n\CODEN-disperse}
            \addlegendentry{Disperse};
        
            \addplot [myparula71,experiment-rate\CODERATE-n\CODEN,discard if neq={args_n}{\CODEN}] table [x=net_file_size_bytes,y=scenario_retrieve_runtime_client_seconds] {figures/experiments2/data-experiments2-rate\CODERATE-bn254.txt};
            \label{leg:experiments-bn254-aggregate-rate\CODERATE-n\CODEN-retrieve}
            \addlegendentry{Retrieve};
            
            \legend{};
            
        \end{axis}
        \def\CODERATE{45}
        \def\CODEN{128}
        \begin{scope}[xshift=-0.05\linewidth,yshift=-0.43\linewidth]
            \node [align=left,anchor=south west] at (0,0) {(g) $(n,k)=(128,58)$};
        \end{scope}
        \begin{axis}[
            semiavidprexperiment2,
            ymode=log, xmode=log,
            xshift=0.0\linewidth,
            yshift=-0.64\linewidth,
        ]
        
            \addplot [myparula61,experiment-rate\CODERATE-n\CODEN,discard if neq={args_n}{\CODEN}] table [x=net_file_size_bytes,y=scenario_disperse_runtime_client_seconds] {figures/experiments2/data-experiments2-rate\CODERATE-bn254.txt};
            \label{leg:experiments-bn254-aggregate-rate\CODERATE-n\CODEN-disperse}
            \addlegendentry{Disperse};
        
            \addplot [myparula71,experiment-rate\CODERATE-n\CODEN,discard if neq={args_n}{\CODEN}] table [x=net_file_size_bytes,y=scenario_retrieve_runtime_client_seconds] {figures/experiments2/data-experiments2-rate\CODERATE-bn254.txt};
            \label{leg:experiments-bn254-aggregate-rate\CODERATE-n\CODEN-retrieve}
            \addlegendentry{Retrieve};
            
            \legend{};
            
        \end{axis}
        \def\CODEN{256}
        \begin{scope}[xshift=0.28\linewidth,yshift=-0.43\linewidth]
            \node [align=left,anchor=south west] at (0,0) {(h) $(n,k)=(256,115)$};
        \end{scope}
        \begin{axis}[
            semiavidprexperiment2,
            ymode=log, xmode=log,
            xshift=0.33\linewidth,
            yshift=-0.64\linewidth,
        ]
        
            \addplot [myparula61,experiment-rate\CODERATE-n\CODEN,discard if neq={args_n}{\CODEN}] table [x=net_file_size_bytes,y=scenario_disperse_runtime_client_seconds] {figures/experiments2/data-experiments2-rate\CODERATE-bn254.txt};
            \label{leg:experiments-bn254-aggregate-rate\CODERATE-n\CODEN-disperse}
            \addlegendentry{Disperse};
        
            \addplot [myparula71,experiment-rate\CODERATE-n\CODEN,discard if neq={args_n}{\CODEN}] table [x=net_file_size_bytes,y=scenario_retrieve_runtime_client_seconds] {figures/experiments2/data-experiments2-rate\CODERATE-bn254.txt};
            \label{leg:experiments-bn254-aggregate-rate\CODERATE-n\CODEN-retrieve}
            \addlegendentry{Retrieve};
            
            \legend{};
            
        \end{axis}
        \def\CODEN{1024}
        \begin{scope}[xshift=0.61\linewidth,yshift=-0.43\linewidth]
            \node [align=left,anchor=south west] at (0,0) {(i) $(n,k)=(1024,461)$};
        \end{scope}
        \begin{axis}[
            semiavidprexperiment2,
            ymode=log, xmode=log,
            xshift=0.66\linewidth,
            yshift=-0.64\linewidth,
        ]
        
            \addplot [myparula61,experiment-rate\CODERATE-n\CODEN,discard if neq={args_n}{\CODEN}] table [x=net_file_size_bytes,y=scenario_disperse_runtime_client_seconds] {figures/experiments2/data-experiments2-rate\CODERATE-bn254.txt};
            \label{leg:experiments-bn254-aggregate-rate\CODERATE-n\CODEN-disperse}
            \addlegendentry{Disperse};
        
            \addplot [myparula71,experiment-rate\CODERATE-n\CODEN,discard if neq={args_n}{\CODEN}] table [x=net_file_size_bytes,y=scenario_retrieve_runtime_client_seconds] {figures/experiments2/data-experiments2-rate\CODERATE-bn254.txt};
            \label{leg:experiments-bn254-aggregate-rate\CODERATE-n\CODEN-retrieve}
            \addlegendentry{Retrieve};
            
            \legend{};

            \addlegendimage{experiment-rate25-n128,black,mark=none};
            \label{leg:experiments-bn254-aggregate-rate25}
            \addlegendimage{experiment-rate33-n128,black,mark=none};
            \label{leg:experiments-bn254-aggregate-rate33}
            \addlegendimage{experiment-rate45-n128,black,mark=none};
            \label{leg:experiments-bn254-aggregate-rate45}
            
            \addlegendimage{experiment-rate25-n128,black,only marks};
            \label{leg:experiments-bn254-aggregate-n128}
            \addlegendimage{experiment-rate25-n256,black,only marks};
            \label{leg:experiments-bn254-aggregate-n256}
            \addlegendimage{experiment-rate25-n1024,black,only marks};
            \label{leg:experiments-bn254-aggregate-n1024}
            
            \addlegendimage{myparula61,only marks,mark=square*};
            \label{leg:experiments-bn254-aggregate-disperse}
            \addlegendimage{myparula71,only marks,mark=square*};
            \label{leg:experiments-bn254-aggregate-retrieve}
            
            \legend{};
            
        \end{axis}
    \end{tikzpicture}%
    \vspace{-0.5em}%
    \caption{%
    Single-thread runtime (ordinate, in seconds)
    of
    $\savidDisperse$
    (\ref{leg:experiments-bn254-aggregate-disperse})
    and
    $\savidRetrieve$
    (\ref{leg:experiments-bn254-aggregate-retrieve})
    on AMD Opteron 6378 processor
    for varying file sizes (abscissa, in $10^6$ bytes;
    for varying $L$)
    in BN254 curve.
    Rows: code rates $k/n\approx0.25,0.33,0.45$
    (\ref{leg:experiments-bn254-aggregate-rate25},
    \ref{leg:experiments-bn254-aggregate-rate33},
    \ref{leg:experiments-bn254-aggregate-rate45}).
    Columns: system sizes $n=128,256,1024$
    (\ref{leg:experiments-bn254-aggregate-n128},
    \ref{leg:experiments-bn254-aggregate-n256},
    \ref{leg:experiments-bn254-aggregate-n1024}).
    (Disaggregated: Figure~\ref{fig:experiments-bn254})%
    }
    \label{fig:experiments-bn254-aggregate}
\end{figure}

The runtime measurements for different steps of $\OURsavid$
shown in Figure~\ref{fig:experiments}
are aggregated for $\savidDisperse$ and $\savidRetrieve$
in
Figure~\ref{fig:experiments-aggregate}.

Figures~\ref{fig:experiments} and~\ref{fig:experiments-aggregate}
use curve BLS12-381.
Corresponding plots
for curve BN254 are provided
in Figures~\ref{fig:experiments-bn254} and~\ref{fig:experiments-bn254-aggregate}.
The plots show a slight speedup,
due to faster operations in BN254.

\section{Proof Details}
\label{sec:appendix-proof-details}

\begin{algorithm}[t]
    \caption{Modified availability game ($\gameAv'$) with resilience $t$ against Semi-AVID-PR scheme $\savid = (\savidSetup, \savidCommit, \savidDisperse, \savidVerify, \savidRetrieve)$}
    \label{alg:game-savid-availability-modified}
    \begin{algorithmic}[1]
        \State $I \drawrandom [n]$
        \State $\CC \gets \emptyset$
            \Comment{Bookkeeping of corrupted parties}
        \State $\forall i\in[n]: P_i \gets \mathsf{new}\,\savid(\emptyset)$
            \Comment{Instantiate $P_i$ as $\savid$ with blank state}
        \State $\cryptoPP \gets \savidSetup^{P_1,...,P_n}(1^\lambda)$
            \Comment{Run setup among all parties}
        \Function{$\CO^{\mathrm{corrupt}}$}{$i$}
                \Comment{Oracle for $\CA$ to corrupt parties}
            \State \Assert{$i \not\in \CC$}
            \If{$i \neq I$}
                 \State $\CC \gets \CC \cup \{i\}$
                \Comment{Mark party as corrupted}
            \State \Return $P_i$
                \Comment{Hand $P_i$'s state to $\CA$}
            \Else
                \State\Abort
            \EndIf
           
        \EndFunction
        \Function{$\CO^{\mathrm{interact}}$}{$i, m$}
                \Comment{Oracle for $\CA$ to interact with parties}
            \State \Assert{$i \not\in \CC$}
            \State \Return $P_i(m)$
                \Comment{Execute $P_i$ on input $m$, return output to $\CA$}
        \EndFunction
        \State $\left(P, C, \left(\CO_i^{\mathrm{node}}(.)\right)_{i\in\CC}\right)
            \gets \CA^{\CO^{\mathrm{corrupt}}(.), \CO^{\mathrm{interact}}(.)}_{\gameAv'}(\cryptoPP)$
            \Comment{$\CA$ returns certificate of retrievability $P$, commitment $C$, and oracle access to corrupted nodes for retrieval}
        \State $\hat B \gets \savidRetrieve^{P_1,...,P_n}\left[\CO_i^{\mathrm{node}}(.)/\Call{Query}{i, .}\right]_{i\in\CC}(P, C)$
            \Comment{During retrieval, interact with corrupted nodes through oracles}
        \State \Return $\begin{aligned}[t]
                    & |\CC| \leq t   \\[-2pt]
                    & \land \savidVerify(P, C) = \TRUE   \\[-5pt]
                    & \land \savidCommit(\hat B) \neq C
                \end{aligned}$
            \Comment{$\CA$ wins iff: while corrupting no more than $t$ parties, $\CA$ produces a valid certificate of retrievability $P$ for $C$ such that retrieval
            does not return a file matching $C$}
    \end{algorithmic}
\end{algorithm}

We modify the availability game $\gameAv$
(Alg.~\ref{alg:game-savid-availability})
to obtain $\gameAv'$
(Alg.~\ref{alg:game-savid-availability-modified})
in which initially the index $I$ of a storage node is sampled
uniformly at random,
and subsequently the game is aborted if the adversary
$\CA_{\gameAv'}$ attempts to corrupt $I$.

\section{Details on Application to Data Availability Sampling}
\label{sec:appendix-data-availability-sampling}

In common blockchain designs all nodes have to download the full blockchain and validate all included transactions (\eg, check that accounts have sufficient balances, no funds are created out of thin air, etc.). However, if a node does not have enough bandwidth, storage, or computational resources to do so, it can instead participate as a so called \emph{light node}.\footnote{Light nodes also occur in the context of sharding, where each node is assigned to a shard and behaves in-shard (\ie, towards their assigned shard) as a full node and out-of-shard (\ie, towards other shards) as a light node.} We assume that every block consists of a header comprised of meta data and a body comprised of a list of transactions. The header includes a commitment to the block content, binding the two together. Full nodes process block headers and content, while light nodes only process block headers. If an invalid transaction was added to a block, this block would be rejected by full nodes but a header of this block can be accepted by a light node, since the light node cannot inspect the block content and verify transaction validity. To prevent light nodes from accepting (the header of) an invalid block, full nodes can produce an invalid transaction fraud proof (\ie, a succinct string of the evidence necessary to verify relative to the block header that the block indeed contains an invalid transaction). To take full nodes' ability to issue invalid transaction fraud proofs, a malicious block producer can perform a data availability attack and withhold parts of the block content, including the invalid transaction. Full nodes would now temporarily reject the block (until its content becomes fully available), but light nodes would not notice the missing content since they do not attempt to download the block content anyway. In this setting, the absence of an invalid transaction fraud proof can thus mean two things: either that the block is alright, or that full nodes were not able to verify the block due to missing data. To rule out the possibility of data unavailability (so that finally lack of fraud proof implies the block is valid), various data availability sampling schemes for light nodes were introduced.

Data availability schemes using Reed-Solomon codes were proposed in \cite{albassam}. In a naive scheme, a block producer encodes a list of transactions, consisting of $k$ chunks, with a $(2k, k)$ Reed-Solomon code. Once a light node receives a header of the block, it randomly queries a few chunks of the encoded block content. The block is accepted only if the queried chunks are available. For a block to be widely accepted by light nodes, most of the light nodes' queried chunks have to be available. Quickly, light nodes' queries cover more than 50\% of coded chunks of the block, so that any remaining missing chunks can be recovered using the Reed-Solomon code. It is therefore no longer possible to trick light nodes into accepting a block while withholding a chunk in an attempt to prevent full nodes from generating an invalid transaction fraud proof.
The main drawback of this solution is that a malicious block producer could invalidly encode the block.
Decoding would then not consistently recover the original chunks' data,
even if nominally enough chunks are available.
Again, full nodes would be able to detect invalid encoding, but light nodes would not.
And again, full nodes could issue a fraud proof to prevent light nodes from accepting
an invalidly encoded block.
However, the amount of evidence needed to prove invalid encoding in this scheme
is as big as the block content itself -- defying the idea of light nodes downloading less than the full block content.
For example, an invalid encoding fraud proof consists of the full original block data, which the light node can verify with respect to the block header, re-encode, and then check that some of the `encoded' chunks received in response to data availability queries do not match the properly encoded chunks.
Subsequent works \cite{albassam,cmt,DBLP:conf/itw/MitraTD20,DBLP:journals/corr/abs-2108-13332} on data availability schemes of this flavor have
thus focussed on reducing the size of invalid encoding fraud proofs,
but drawbacks remain (\eg, additional complexity, timing assumptions).

A different approach is to eliminate invalid encoding fraud proofs by making it impossible for block producers to invalidly encode data. Such schemes can be achieved using polynomial commitment schemes such as KZG \cite{kate}. Treated as evaluations of a polynomial at agreed-upon locations, the $k$ chunks of the block content uniquely determine a polynomial of degree $k-1$.
A commitment to this polynomial is included in the block header.
To ensure data availability, light nodes query for evaluations of this polynomial at random locations. Consistency of every query response with the polynomial committed to in the block header can be verified using the polynomial commitment scheme by providing an evaluation witness.
Even with few queries each, light nodes together will soon have queried evaluations at at least $k$ distinct locations. If the block producer withholds any of these evaluations, light nodes will not accept the block.
But once evaluations at $k$ distinct locations are available,
the polynomial, and thus the block content, can be reconstructed.
Any invalid transaction becomes visible and full nodes can generate corresponding fraud proofs.
Schemes of this flavor however require to compute an evaluation witness for each query,
which despite recent algorithmic improvements is still computationally heavy \cite{fk,alin,tab}.

The Semi-AVID-PR scheme, described in Section~\ref{sec:protocol} and illustrated in Figure~\ref{fig:savidpr-disperse}, can be seen as combining `the best of both worlds' in that it does not require invalid encoding fraud proofs,
but sampled chunks can be verified efficiently.
We assume that a block contains an alternating sequence
of transactions and commitments to resulting intermediary chain states
(this is used for invalid transaction fraud proofs as in \cite{albassam}).
As
illustrated in Figure~\ref{fig:savidpr-disperse}, the
block producer arranges the block content $\BU$ as a matrix of size $L \times k$, where $k$ and $L$ are system parameters.
It commits to each of the columns $\Bu_1, ..., \Bu_k$ of that matrix using the linear vector commitment scheme defined in Section~\ref{sec:preliminaries-slvc} (to obtain commitments $h_1, ..., h_k$), and encodes the matrix $\BU$ row-wise using a $(n, k)$ Reed-Solomon code to obtain chunks $\Bc_1, ..., \Bc_n$ of a coded matrix $\BC$. A final commitment to the full block content is computed as $C \triangleq \CRHF(h_1\|...\|h_k)$ and used on-chain in the block's header to uniquely reference the block content. Full nodes receive the full block content, recompute the column commitments and their hash, and compare it with $C$ to verify the block content.
Light nodes receive only $C$ from the block header. Prior to accepting a new block header, a light node samples random coded chunks $\Bc_i$. The response to each query is accompanied by
purported column commitments $h_1, ..., h_k$. Every light node can verify the column commitments by locally recomputing their hash and comparing it with the commitment $C$ in the block header. Subsequently, the light node verifies
the
downloaded chunk $\Bc_i$ using the linear homomorphic property of the vector commitment scheme and the column commitments $h_1, ..., h_k$.
To employ the invalid transaction fraud proofs of \cite{albassam},
it remains to show how a full node can open any entry of $\BU$
to a light node
in a verifiable manner.
For this purpose, an opening witness for value $y = [\Bu_j]_i$ at position $(i, j)$ consists of:
\begin{itemize}
    \item Value $y = [\Bu_j]_i$ and coordinates $(i, j)$.
    \item Commitments $h_1, ..., h_k$ to columns $\Bu_1, ..., \Bu_k$.
    \item A witness $w$ for the opening of $y$ at the $i$-th position in $\Bu_j$ with respect to the vector commitment $h_j$.
\end{itemize}
The light client first verifies the commitments $h_1, ..., h_k$
by comparing their hash to the commitment $C$ in the block header.
The client then verifies the opening of the value $y$ at position $i$ in $\Bu_j$ with respect to the vector commitment $h_j$.

Since in Semi-AVID-PR valid encoding can be verified by light nodes using the homomorphic property of linear vector commitments, invalid encoding fraud proofs are not needed needed.
At the same time, verifying a chunk requires only to compute a vector commitment (to $\Bc_i$) and a linear combination of the vector commitments $h_1, ..., h_k$, which is lightweight to compute.
Performance is discussed in more detail in
Section~\ref{sec:evaluation}.

\end{document}